\documentclass[12pt, fleqn, final]{article}
\usepackage{structure}

\usepackage{tikz}
\usetikzlibrary{arrows}
\usepackage{tikz-cd}
\usepackage{frcursive}

\newcommand{\arr}{\text{\textsl{\textcursive{r}}}}
\newcommand{\nord}[1]{\mathbf{:} #1 \mathbf{:}}

\DeclareFontEncoding{LS1}{}{}
\DeclareFontSubstitution{LS1}{stix}{m}{n}
\DeclareSymbolFont{symbols2}{LS1}{stixfrak}{m}{n}
\DeclareMathSymbol{\typecolon}{\mathbin}{symbols2}{"25}
\newcommand{\altnord}[1]{\typecolon #1 \typecolon}

\title{Lorentzian 2d CFT from the pAQFT perspective}
\author{
    \small Sam Crawford \\
    \small\textit{University of York,}  \\
	\small\textit{Department of Mathematics,}\\
    \small\href{mailto:sam.crawford@york.ac.uk}{\texttt{sam.crawford@york.ac.uk}}
        \and
    \small Kasia Rejzner  \\
    \small\textit{University of York,}  \\
	\small\textit{Department of Mathematics,}\\
    \small\href{mailto:kasia.rejzner@york.ac.uk}{\texttt{kasia.rejzner@york.ac.uk}}
        \and
    \small Beno\^{\i}t Vicedo \\
    \small \textit{University of York,}  \\
	\small \textit{Department of Mathematics,}\\
    \small \href{mailto:benoit.vicedo@york.ac.uk}{\texttt{benoit.vicedo@york.ac.uk}}
}

\date{\today}

\begin{document}
    
\maketitle


\definecolor{todoKasia}{RGB}{207, 58, 174}
\definecolor{todoBenoit}{RGB}{87, 159, 222}



\ifdraft{
\noindent
{\color{themesection}\Large\bfseries Key:}
\begin{itemize}[itemsep=-1em]
    \item[] \hspace{-2.7em} \tikz[baseline=0.5ex]{\path[draw=black,fill=todogreen ] (0, 0) rectangle (0.4cm, 0.4cm);} :
                Boring things for Sam
    \item[] \hspace{-2.7em} \tikz[baseline=0.5ex]{\path[draw=black,fill=todoKasia ] (0, 0) rectangle (0.4cm, 0.4cm);} :
                Comments by Kasia
    \item[] \hspace{-2.7em} \tikz[baseline=0.5ex]{\path[draw=black,fill=todoBenoit ] (0, 0) rectangle (0.4cm, 0.4cm);} :
                Comments by Benoit
    \item[] \hspace{-2.7em} \tikz[baseline=0.5ex]{\path[draw=black,fill=todoorange] (0, 0) rectangle (0.4cm, 0.4cm);} :
                Format/layout queries
    \item[] \hspace{-2.7em} \tikz[baseline=0.5ex]{\path[draw=black,fill=todored   ] (0, 0) rectangle (0.4cm, 0.4cm);} :
                Mathematical queries/concerns
\end{itemize}
}

\listoftodos

\begin{abstract}
    We provide a detailed construction of the quantum theory of
    the massless scalar field on
    2-dimensional, globally-hyperbolic (in particular, Lorentzian) manifolds
    using the framework of perturbative algebraic quantum field theory.
    From this we obtain subalgebras of observables isomorphic to the Heisenberg
    and Virasoro algebras on the Einstein cylinder.
    We also show how the conformal version of general covariance,
    as first introduced by Pinamonti as an extension of the construction
    due to Brunetti, Fredenhagen and Verch,
    may be applied to the concept of natural Lagrangians in order to
    obtain a simple condition for the conformal covariance
    of classical dynamics, which is then shown to quantise
    in the case of a quadratic Lagrangian.
    We then compare the covariance condition for the stress-energy tensor
    in the classical and quantum theory in Minkowksi space,
    obtaining a transformation law dependent on
    the Schwarzian derivative of the transformed coordinate,
    in accordance with a well-known result in the Euclidean literature.
\end{abstract}

\tableofcontents

\section{Introduction}
One of the most important problems faced by mathematical physicists nowadays
is the search for mathematically rigorous formulation of quantum field theory (QFT).
Over the span of six decades, several axiomatic frameworks have been developed
(including algebraic quantum field theory
\cite{haagLocalQuantumPhysics1996,haagAlgebraicApproachQuantum1964}),
but none of them can yet claim to
include an interacting QFT model in 4 spacetime dimensions.
On the other hand, a lot is known about lower-dimensional cases
(prominently 2-dimensional) and in the presence of symmetries,
e.g. the conformal symmetry.
The huge success of conformal field theory (CFT) and
its ubiquity in theoretical physics is evidenced by
a vast trove of literature and impressive number of
results obtained throughout the history of the subject
\cite{
    belavinInfiniteConformalSymmetry1984,
    ginspargAppliedConformalField1988,francescoConformalFieldTheory1997,
    schottenloherMathematicalIntroductionConformal2008}.
In two dimensions, CFT plays a central role in the world-sheet description of string theory. More generally, they describe continuous phase transitions in condensed matter systems, critical points of renormalisation group flows in quantum field theories and provide duals to gravitational theories in anti-de Sitter spacetimes via the AdS/CFT correspondence. From a mathematical point of view, the rigorous formulation of two-dimensional Euclidean chiral CFT has led to the important development of vertex operator algebras (VOA), see e.g.
\cite{
    kacVertexAlgebrasBeginners1998,
    lepowskyIntroductionVertexOperator2004,
    frenkelVertexAlgebrasAlgebraic2004,
    beilinsonChiralAlgebras2004},
which has been instrumental in various areas of pure mathematics,
including the proof of the monstrous moonshine conjecture
\cite{frenkelVertexOperatorAlgebras1989,
    borcherdsVertexAlgebrasKacMoody1986,
    borcherdsMonstrousMoonshineMonstrous1992}
and in the study of the geometric Langlands correspondence
\cite{feiginAffineKacmoodyAlgebras1992,
    beilinsonQuantizationHitchinFibration1996,
    frenkelLecturesLanglandsProgram2005,
    frenkelLanglandsCorrespondenceLoop2007}.

CFT has also provided a rich class of models that
satisfy algebraic quantum field theory (AQFT) axioms,
as demonstrated for example in
\cite{
    gabbianiOperatorAlgebrasConformal1993,
    kawahigashiLocalConformalNets2006,
    kawahigashiClassificationLocalConformal2002,
    kawahigashiClassificationTwoDimensionalLocal2004,
    bischoffGeneralizedOrbifoldConstruction2017,
    bischoffCharacterization2DRational2015,
    longoLocalFieldsBoundary2004,
    longoAlgebraicConstructionBoundary2011}.
The  main principles of AQFT can also be applied to describe perturbative QFT.
This led to development of
\emph{perturbative algebraic quantum field theory} (pAQFT),
which started in the 90s
\cite{brunettiPerturbativeAlgebraicQuantum2009,
    brunettiMicrolocalAnalysisInteracting2000,
    brunettiInteractingQuantumFields1997,
    duetschPerturbativeAlgebraicField2001,
    duetschMasterWardIdentity2003,
    duetschCausalPerturbationTheory2004}
(see also \cite{rejznerPerturbativeAlgebraicQuantum2016} and
\cite{dutschClassicalFieldTheory2019} for a review).
The advantage of pAQFT is that it combines the ideas of AQFT with
the powerful methods of perturbation theory and renormalization and allows one to
construct physically interesting models in 4-dimensions, also on curved spacetimes.
However, the ultimate goal of pAQFT is
to understand how non-perturbative results could be obtained.
To this end,
it is useful to construct some known non-perturbative models using pAQFT methods and see how convergence and non-perturbative effects arise.
An example of such a model has been investigated in
\cite{bahnsQuantumSineGordonModel2018}. 
This paper is the first step in the research programme aimed at understanding how CFT fits into the framework of pAQFT. The advantages of such a combination are twofold: 
\begin{itemize}
    \item Many of the CFT results are proven only in the Euclidean signature. With the aid of pAQFT, we want to show how to obtain them in Lorentizian signature as well.
    \item Some powerful techniques used in CFT can be applied in pAQFT to obtain non-perturbative results.
\end{itemize}
In the present paper we concentrate on setting up the general framework, with particular focus on \emph{local conformal covariance}. We improve on existing results of \cite{pinamontiConformalGenerallyCovariant2009} and apply our methods to define normally-ordered covariant quantities, with Virasoro generators on a cylinder among them. We show that covariant normal-ordering allows one to reproduce the correct Virasoro algebra relations on the cylinder and we demonstrate how the usual ``Zeta regularisation'' trick can be rigorously understood as the change in the choice of normal ordering.

\section{Mathematical Preliminaries}

In this section, we provide an account of the constructions
of pAQFT relevant to our discussion.
For a more thorough exposition, the reader is directed towards
\cite{rejznerPerturbativeAlgebraicQuantum2016}.
In particular, whilst we may, from time to time,
discuss the possibility of interactions in the classical theory,
all of our quantum constructions shall be specific to the free scalar field.

We begin with the kinematics (i.e. states and observables)
of our classical theory.
Due to our use of deformation quantisation,
this will also establish the observables of the quantum theory.
Next, we address in \Cref{sec:classical_dynamics}
the matter of imposing  suitable dynamics on the system,
using the generalised Lagrangian formalism.
For an appropriately chosen Lagrangian,
we are then able to endow our space of observables with a Poisson structure.
At this point, the algebra is decidedly ``off-shell'',
as the field configurations we consider include
those which do not satisfy the equations of motion.
Therefore, in \Cref{sec:classical_BV}, we make a detour to examine how,
in the case of the free scalar field,
our construction does indeed recover the canonical
(i.e. ``equal-time'') Poisson bracket on-shell.
Here we also briefly explore the `dg' perspective of QFT
at the heart of the Costello-Gwilliam  formalism 
\cite{costelloFactorizationAlgebrasQuantum2016} and descriptions of
`higher' QFT as outlined in, for example,
\cite{beniniCategorificationAlgebraicQuantum2021}.

Satisfied with our choice of Poisson structure,
we then use it in \Cref{sec:deformation_quantisation} to deform
the pointwise product of functionals into an associative product $\star$,
which is analogous to the operator composition of canonical quantisation.
Once the quantum algebra has been established,
we discuss the comparison between classical and quantum observables.
The difficulty in `quantising' classical observables
is traditionally known as the \textit{ordering problem}.
In an attempt to find the most natural solution to this problem,
we then introduce in \Cref{sec:normal_ordering}
the concept of \textit{local covariance},
where we require our theory to be defined
in a coherent manner across multiple spacetimes.
This is so that we may be sure our ordering prescription is not dependent
on the global geometry of any particular spacetime
(which local algebras should in principle be unaware of).

\subsection{Classical Kinematics}
\label{sec:kin}

Let $M$ be a smooth manifold
(we shall specify dimension and topological constraints later).
For the theory of a real scalar field,
we take our configuration space, $\mathfrak{E}(M)$,
to be the space of smooth real-valued functions on $M$.
More generally, we might consider the space of smooth sections of
some vector bundle $E \overset{\pi}{\to} M$,
to which the following constructions can be readily generalised.
Note that this space is ``off-shell'' in the sense that it
includes field configurations which
may not satisfy any equations of motions later imposed by the dynamics.

Classically, observables are maps
$\mathcal{F}: \mathfrak{E}(M) \to \mathbb{C}$.
Typically, we also assume them to be smooth,
with respect to an appropriate notion of smoothness which we shall introduce shortly.
The derivative of a functional at a point
$\phi \in \mathfrak{E}(M)$ and in a direction
$h    \in \mathfrak{E}(M)$
is defined in the obvious way as
\begin{equation}
    \label{eq:directional_derivative}
    \big\langle \mathcal{F}^{(1)}[\phi], h \big\rangle
        :=
    \lim_{\epsilon \to 0}
        \frac{
            \mathcal{F}[\phi + \epsilon h] - \mathcal{F}[\phi]
        }{
            \epsilon
        },
\end{equation}
whenever this limit exists.
If it exists for all $\phi, h \in \mathfrak{E}(M)$,
and the map
\begin{equation*}
    \mathcal{F}^{(1)}:
    (\phi, h)
        \mapsto
    \big\langle \mathcal{F}^{(1)}[\phi], h \big\rangle
\end{equation*}
is continuous with respect to the product topology on
$\mathfrak{E}(M)^2$ then we say $\mathcal{F}$ is $\mathcal{C}^1$.

Higher derivatives of $\mathcal{F}$ are defined similarly by
\begin{equation}
    \left\langle 
        \mathcal{F}^{(n)}[\phi], 
        h_1 \otimes \cdots \otimes h_n
    \right\rangle
        :=
    \frac{
        \partial^n \mathcal{F}[\phi + \epsilon_1 h_1 + \cdots + \epsilon_n h_n]
    }{
        \partial \epsilon_1 \cdots \partial \epsilon_n
    }
    \big|_{\epsilon_1 = \cdots = \epsilon_n = 0},
\end{equation}
wherever these limits exist.
If, $\forall n \in \mathbb{N}$ and $\phi \in \mathfrak{E}(M)$,
$\mathcal{F}^{(n)}[\phi] \in \mathfrak{E}'(M^n)$ exists,
and the maps
\begin{align*}
    \mathcal{F}^{(n)}: 
        \mathfrak{E}(M) \times \mathfrak{E}(M^n)
            &\to
        \mathbb{C} \\
        \left( \phi, h_1 \otimes \cdots \otimes h_n \right) 
            &\mapsto
        \left\langle
            \mathcal{F}^{(n)}[\phi], h_1 \otimes \cdots \otimes h_n 
        \right\rangle
\end{align*}
are all continuous then we say $\mathcal{F}$ is \textit{Bastiani} smooth
as discussed in, for example \cite[\S II]{brouderPropertiesFieldFunctionals2018}.
We shall denote by $\mathfrak{F}(M)$ the space of Bastiani smooth functionals
of the real scalar field over $M$.

Various pieces of notation are commonly used when discussing functional derivatives.
For clarity, we collect some of them here.
Firstly, note that for a $\mathcal{C}^1$ functional $\mathcal{F}$,
the above definition implies that
$\mathcal{F}^{(1)}[\phi] \in \mathfrak{E}'(M)$,
using Schwartz's notation for compactly supported distributions.
Hence the bracket $\left\langle \cdot, \cdot \right\rangle$
in \eqref{eq:directional_derivative} can be seen as denoting the canonical
pairing $V' \times V \to \mathbb{C}$,
where $V$ is a topological vector space over $\mathbb{C}$
and $V'$ is its continuous dual space.
If $M$ is equipped with a preferred volume form%
\footnote{
    As we are only interested in Lorentzian manifolds,
    we always have the metric volume form.
    Our definitions of various classes of functionals assume a preferred volume form,
    other authors opt instead to define $\delta F / \delta \phi$
    as a \textit{distribution density} \cite[p.145]{hormanderAnalysisLinearPartial2015}.
}
$\mathcal{F}^{(1)}[\phi]$ may be given an integral kernel,
typically written as
\begin{equation}
    \label{eq:derivative_integral_kernel}
    \big\langle \mathcal{F}^{(1)}[\phi], h \big\rangle
        =
    \int_{M} 
        \frac{\delta \mathcal{F}[\phi]}{\delta \phi(x)}
        h(x)
        \, \mathrm{dVol}_M.
\end{equation}
Finally, we introduce the map,
for a $\mathcal{C}^1$ funcitonal $\mathcal{F}$,
$\tfrac{\delta}{\delta \phi}: \mathcal{F} \mapsto \mathcal{F}^{(1)}$.

A condition we will frequently impose is that,
for every $\phi \in \mathfrak{E}(M)$,
$\mathcal{F}^{(1)}[\phi]$ is a \textit{smooth} compactly-supported distribution.
This means that the integral kernel
$\delta \mathcal{F}/\delta \phi$
in \eqref{eq:derivative_integral_kernel} may be replaced by
some \textit{test function} $f_\phi \in \mathfrak{D}(M)$.
Here we again use the notation due to Schwartz, where $\mathfrak{D}(M)$
denotes the subspace of $\mathfrak{E}(M)$ containing
smooth functions with \textit{compact support}.
If one can similarly associate an element $f_{n, \phi} \in \mathfrak{D}(M^n)$
to $\mathcal{F}^{(n)}[\phi]$ for all $n, \phi$,
then we say that $\mathcal{F}$ is a \textit{regular} functional,
and we collect all such functionals into the space
$\mathfrak{F}_\mathrm{reg}(M)$.

As mentioned, for a Bastiani smooth functional $\mathcal{F} \in \mathfrak{F}(M)$,
$\mathcal{F}^{(n)}[\phi]$ will in general be
a compactly-supported distribution of $n$ variables.
We say this distribution is \textit{regular} if there exists
$f \in \mathfrak{D}(M^n)$ such that
$\forall h \in \mathfrak{E}(M^n)$
\begin{equation*}
    \left\langle 
        \mathcal{F}^{(n)}[\phi],
        h
    \right\rangle
        =
    \int_{M^n} 
        f(x_1, \ldots, x_n)
        h(x_1, \ldots, x_n)
        \mathrm{dVol}_M^n.
\end{equation*}
If $\mathcal{F}^{(n)}[\phi]$ is a regular distribution
for every $n \in \mathbb{N}$ and $\phi \in \mathfrak{E}(\mathcal{M})$,
then we say that $\mathcal{F}$ is a \textit{regular functional}.
We denote the space of regular functionals $\mathfrak{F}_\mathrm{reg}(M)$.

Regular functionals are particularly convenient to work with,
as we shall see when defining the Poisson bracket and $\star$ product later.
However, they exclude many functionals of physical interest,
such as components of the stress-energy tensor in the case of the scalar field.
Thus, we next consider the subspace of
$\mathfrak{F}(M)$ consisting of \textit{local} functionals.

Following \cite{rejznerPerturbativeAlgebraicQuantum2016}, we define a functional
$\mathcal{F} \in \mathfrak{F}(M)$ to be \textit{local} if there exists an open cover
$\bigcup_{\alpha \in \mathcal{A}} U_\alpha = \mathfrak{E}(M)$
such that, for $\phi \in U_\alpha$
\begin{equation}
    \mathcal{F}[\phi]
        =
    \int_M f_\alpha \left( j^k_x \phi \right) \mathrm{dVol}_M,
\end{equation}
where $j^k_x \phi$ is the $k^\text{th}$ jet prolongation of $\phi$ at $x$
(loosely $j^k_x \phi = (\phi(x), \nabla \phi(x), \ldots, \nabla^k \phi(x))$),
and $f_\alpha$ is some smooth, compactly-supported function on
the $k^\text{th}$ jet bundle of $M$.
We denote by $\mathfrak{F}_\mathrm{loc}(M)$ the space of local functionals on $M$,
and by $\mathfrak{F}_{\mathrm{mloc}}(M)$ the space of \textit{multilocal} functionals:
the algebraic completion of $\mathfrak{F}_\mathrm{loc}(M)$
under the pointwise product of functionals.

An important property of local functionals
\cite[Remark 3.2]{rejznerPerturbativeAlgebraicQuantum2016}
is that, for every $n \in \mathbb{N}$, $\phi \in \mathfrak{E}(M)$,
the support
\footnote{
    In the sense of distributions.
    See e.g. \cite[Definition 2.2.2]{hormanderAnalysisLinearPartial2015}.
}
of $\mathcal{F}^{(n)}[\phi]$ is contained within the thin diagonal
\begin{equation*}
    \Delta_n = \left\{ (x, \ldots, x) \in M^n \right\}_{x \in M}.
\end{equation*}
Immediately this implies that, for $n \geq 2$,
$\mathcal{F}^{(n)}[\phi]$ must either vanish or fail to be regular.
In other words, the intersection
$\mathfrak{F}_\mathrm{reg}(M) \cap \mathfrak{F}_\mathrm{loc}(M)$
comprises only the linear functionals
\begin{equation*}
    \Phi(f): \phi \mapsto \int_M f(x) \phi(x) \, \mathrm{dVol}_M,
\end{equation*}
for $f \in \mathfrak{D}(M)$.

Whilst it is possible to perform our classical and quantum operations
on local functionals, the result is typically not itself local.
As such, we need a space of functionals which is algebraically convenient,
like $\mathfrak{F}_\mathrm{reg}(M)$, but which also contains the physically important
subspace $\mathfrak{F}_\mathrm{loc}(M)$.
The space of \textit{mircocausal} functionals accomplishes this.
Unlike the previous classes of functionals however,
its definition requires more than a smooth structure on $M$.
Instead we require the structure of a \textit{spacetime},
which we define in accordance with 
\cite[\S2.1]{fewsterDynamicalLocalityCovariance2012} as follows:

\begin{definition}[Spacetime]
    \label{def:spacetime}
    A \textit{spacetime} is a tuple
    $\mathcal{M} = (M, g, \mathfrak{o}, \mathfrak{t})$
    such that $(M, g)$ is an orientable Lorentzian manifold
    of some fixed dimension $d$,
    $\mathfrak{o} \subset \Omega^d(M)$ is an equivalence class
    of nowhere-vanishing volume forms, defining an orientation,
    and $\mathfrak{t} \subset \mathfrak{X}(M)$ is an equivalence class of
    timelike vector fields, where
    $t \sim t' \Leftrightarrow g_x(t_x, t'_x) > 0 \forall x \in M$.
\end{definition}

We will typically write $\mathfrak{F}(\mathcal{M})$,
$\mathfrak{F}_{\mathrm{reg}}(\mathcal{M})$, and
$\mathfrak{F}_\mathrm{loc}(\mathcal{M})$
to refer to the respective spaces of functionals associated
to the underlying manifold of $\mathcal{M}$.

For any point $x$ in a spacetime $\mathcal{M}$,
we can define the closed past/future lightcone of the cotangent space
$\overline{V}_\pm(x) \subset T^*_x M$ as comprising covectors $k$
for which $\hat{g}_x(k, k) \geq 0$ and $\pm k(t_x) \geq 0$,
for any $t \in \mathfrak{t}$,
where $\hat{g}_x$ is the metric induced on $T^*_x M$ by $g$.
We can then define the fibre bundles
$\overline{V}_\pm \subset T^*M$ such that their fibres at $x$ are
$\overline{V}_\pm(x)$ respectively.

A functional $\mathcal{F} \in \mathfrak{F}(\mathcal{M})$
is \textit{microcausal} if it satisfies the
\textit{wavefront set spectral condition}
\begin{equation}
    \label{eq:spectral_condition}
    \mathrm{WF}(\mathcal{F}^{(n)}[\phi])
        \cap
    \left( 
        \overline{V}_+^n \cup \overline{V}_-^n
    \right)
        =
    \emptyset,
\end{equation}
For detailed definitions and properties of the wavefront set of a distribution,
see for example \cite[\S 8]{hormanderAnalysisLinearPartial2015},
as well as \cite{brouderSmoothIntroductionWavefront2014}.
Briefly put, the wavefront set is a way of characterising
the singularity structure of a distribution $T \in \mathfrak{D}'(M)$,
i.e. the precise manner in which $T$ fails to be a smooth function.
It consists of the set of \textit{non-zero} covectors
$(x, k) \in T^*M$ such that there exists no neighbourhood of
$x$ to which the restriction of $T$ is smooth,
and the Fourier transform
--
defined in an arbitrary chart, which turns out to be irrelevant
--
of $T$ fails to decay rapidly in the direction $k$.
The space of microcausal functionals is denoted
$\mathfrak{F}_{\mu c}(\mathcal{M})$,
and contains both the local and regular spaces of functionals.

The characteristic features of these spaces,
as well as the relations between them,
are summarised in the following diagram.

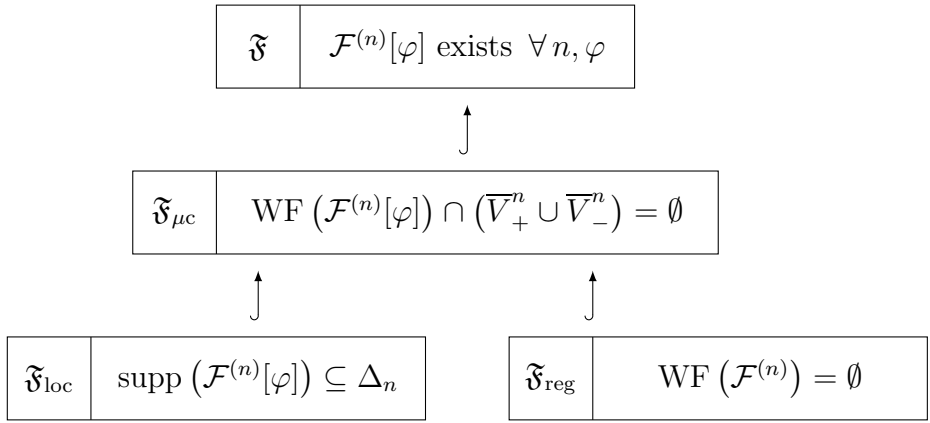
\begin{figure}[h]
    \centering
    \begin{tikzpicture}[scale=1.1]
        \draw (2.5, 2) rectangle (7.5, 1);
        \node at (3, 1.5) {$\mathfrak{F}$};
        \draw (3.5, 2) -- (3.5, 1);
        \node at (5.5, 1.5) {
            $\mathcal{F}^{(n)}[\phi]$ exists $\forall n, \phi$
            };        

        \draw (0, -2) rectangle (5, -3);
        \node at (.5, -2.5) {$\mathfrak{F}_\mathrm{loc}$};
        \draw (1, -2) -- (1, -3);
        \node at (3, -2.5) {
            $\mathrm{supp} \left( \mathcal{F}^{(n)}[\phi] \right) \subseteq \Delta_n$
            };

        \draw (6, -2) rectangle (11, -3);
        \node at (6.5, -2.5) {$\mathfrak{F}_\mathrm{reg}$};
        \draw (7, -2) -- (7, -3);
        \node at (9, -2.5) {
            $\mathrm{WF}\left( \mathcal{F}^{(n)} \right) = \emptyset$
            };

        \draw [right hook-latex] (3, -1.8) -- (3, -1.2);
        \draw [right hook-latex] (7, -1.8) -- (7, -1.2);
        \draw [right hook-latex] (5.5, 0.2) -- (5.5, 0.8);
        
        \draw (1.5, 0) rectangle (8.5, -1);
        \node at (2, -.5) {$\mathfrak{F}_\mathrm{\mu c}$};
        \draw (2.5, 0) -- (2.5, -1);
        \node at (5.5, -.5) {
            $
                \mathrm{WF}\left( \mathcal{F}^{(n)}[\phi] \right)
                    \cap
                \left( 
                    \overline{V}_+^n \cup \overline{V}_-^n
                \right)
                    =
                \emptyset
            $
        };
    \end{tikzpicture}
\end{figure}

\subsection{Classical Dynamics}
\label{sec:classical_dynamics}

There are many ways to specify the dynamics of a classical field theory.
In the present formalism it is achieved through a rigorous implementation of
the principle of critical action.
The foundational idea of this approach,
due to Peierls \cite{peierlsCommutationLawsRelativistic1952},
is the formulation of a Poisson structure in terms of
the advanced and retarded responses of a field to perturbation.
A construction of the classical algebra of observables using the Peierls bracket
was set forth in \cite{duetschMasterWardIdentity2003},
and developed in detail in
\cite{brunettiAlgebraicStructureClassical2019}.
More recent overviews may be found in,
e.g. \cite[\S 4]{rejznerPerturbativeAlgebraicQuantum2016}
or  \cite[\S 5.1]{fredenhagenPerturbativeConstructionModels2015}.

This approach has the advantage of being manifestly Poincar\'e covariant,
as will be explored further in \cref{sec:normal_ordering},
whilst still endowing our space of observables with a Poisson structure,
contrary to a common notion that a choice of Poisson structure requires one to
split a spacetime into `space' and `time'.

The issue with na\"ively written actions for common field theories,
such as the Klein-Gordon or Yang-Mills functionals,
is that their region of integration must be restricted to
a compact subset of spacetime in order to guarantee a finite value is returned.
A convenient way to achieve this is to define a map
$\mathcal{L}: \mathfrak{D}(\mathcal{M}) \to \mathfrak{F}_\mathrm{loc}(\mathcal{M})$,
where the functional $\mathcal{L}(f)$ is interpreted as the action functional with
an introduced cutoff function $f$.
Not every such map is suitable however, the necessary criteria are
outlined in the following definition
(after \cite[\S 4.1]{rejznerPerturbativeAlgebraicQuantum2016}).

\begin{definition}
    \label{def:generalised_Lagrangian}
    A map
    $
        \mathcal{L}:
        \mathfrak{D}(\mathcal{M})
            \to
        \mathfrak{F}_\mathrm{loc}(\mathcal{M})
    $
    is called a \textit{generalised Lagrangian}
    if it satisfies the following conditions:
    \begin{enumerate}
        \item   If $f, g, h \in \mathfrak{D}(\mathcal{M})$ such that
                $\mathrm{supp} \, f \cap \mathrm{supp} \, h = \emptyset$,
                then
                \begin{equation*}
                    \mathcal{L}(f + g + h) =
                    \mathcal{L}(f + g) - \mathcal{L}(g) + \mathcal{L}(g + h).
                    \hfill (\textit{Additivity})
                \end{equation*}
        \item   $\mathrm{supp} \, \mathcal{L}(f) \subseteq \mathrm{supp} \, f.$
                \hfill (\textit{Support})
        \item   If $\beta$ is an isometry of $(M, g)$
                which preserves orientation and time-orientation, then for
                $f \in \mathfrak{D}(\mathcal{M})$
                and
                $\phi \in \mathfrak{E}(\mathcal{M}),$
                \begin{equation*}
                    \mathcal{L}(f)[\beta^* \phi] = \mathcal{L}(\beta_* f)[\phi].
                    \hfill (\textit{Covariance})
                \end{equation*}
    \end{enumerate}
\end{definition}
This definition refers to the \textit{spacetime support},
which we denote $\mathrm{supp} \, \mathcal{F}$ for a functional $\mathcal{F}$.
This is the \textit{closure} of the set of points $x \in \mathcal{M}$ such that,
for all $\phi \in \mathfrak{E}(\mathcal{M})$,
there exists some perturbation localised to a neighbourhood of $x$,
say $\psi \in \mathfrak{D}(U)$ for some $U \ni x$,
which changes the output of $\mathcal{F}$,
i.e. $\mathcal{F}[\phi + \psi] \neq \mathcal{F}[\phi]$.
For example, if $x_0 \in \mathcal{M}$,
the spacetime support of the evaluation functional $(\phi \mapsto \phi(x_0))$
is just $\left\{ x_0 \right\}$.

A primary example is the generalised Lagrangian for
the Klein-Gordon field on a spacetime $\mathcal{M}$,
which is given by
\begin{equation}
    \label{eq:KG_Lagrangian}
    \mathcal{L}_\mathcal{M}(f)[\phi]
        :=
    \frac{1}{2} \int_\mathcal{M}
        f \left[ g(\nabla \phi, \nabla \phi) - m \phi^2 \right] \,
        \mathrm{dVol}_g,
\end{equation}
where $\nabla$ is the gradient operator associated to the
metric $g$ of $\mathcal{M}$ and $\mathrm{dVol}_g$ is its associated volume form.

Heuristically, one may think of the limit of $\mathcal{L}(f)$
as $f$ tends to a Dirac delta $\delta_x$ as describing the Lagrangian density at $x$
and, if $f$ instead tends to the constant function $\mathbf{1}$,
then $\mathcal{L}(f)$ becomes the action functional $S$.
However one must bear in mind that, in general,
these limits may not (and typically \textit{will} not) yield
well-defined local functionals.

Given a generalised Lagrangian $\mathcal{L}$,
we define the \textit{Euler-Lagrange derivative}
at a point $\phi \in \mathfrak{E}(\mathcal{M})$
as the distribution $S'[\phi]$ such that
\begin{equation}
    \label{eq:Euler_Lagrange}
    \left\langle \mathcal{L}(f)^{(1)}[\phi], h \right\rangle
        =:
    \left\langle S'[\phi], h \right\rangle.
\end{equation}
where,
$h \in \mathfrak{D}(\mathcal{M})$ and
$f \in \mathfrak{D}(\mathcal{M})$ is chosen such that
$f^{-1}\{1\}$ contains a neighbourhood of $\mathrm{supp} \, h$
\footnote{
    We opt for a slightly stronger condition on $f$ than is usual,
    this is ultimately insignificant, but it makes it easier to show that
    null Lagrangians (defined below) have vanishing Euler-Lagrange derivative
}.
One can use the additivity and support properties to verify that
$S'[\phi]$ is well-defined
(i.e. \eqref{eq:Euler_Lagrange} is independent of the choice of $f$).
A field configuration $\phi \in \mathfrak{E}(\mathcal{M})$
is called \textit{on-shell} if its Euler-Lagrange derivative
$S'[\phi]$ vanishes as a distribution.

Different choices of a generalised Lagrangian may
yield the same Euler-Lagrange derivative.
If a generalised Lagrangian $\mathcal{L}_0$ satisfies
$\mathrm{supp} \, \mathcal{L}_0(f) \subseteq \mathrm{supp} \, df$,
then clearly its Euler-Lagrange derivative vanishes for all
$\phi \in \mathfrak{E}(\mathcal{M})$.
In such a case, we describe $\mathcal{L}_0$ as \textit{null}.
One may add a null Lagrangian to an arbitrary generalised Lagrangian
without changing its Euler-Lagrange derivative.
Given this, we say that two generalised Lagrangians,
$\mathcal{L}$ and $\mathcal{L}'$ define the same \textit{action}
if their difference is null,
we denote this fact by $[\mathcal{L}] = [\mathcal{L}'] =: S$.

In the case where $S$ is a quadratic action,
(i.e. it may be represented by a Lagrangian $\mathcal{L}$
such that $\mathcal{L}(f)$ is a quadratic functional for all $f$)
the map
$\phi \mapsto \left\langle S'[\phi], h \right\rangle$
is linear in $\phi$.
We assume that this functional can be expressed in the form
$\phi \mapsto \left\langle P \phi, h \right\rangle$,
where $P$ is a \textit{normally hyperbolic} differential operator,
i.e. $P$ is a second order differential operator of the form
$\nabla^a \nabla_a +$ lower order terms.
A more precise definition of normally hyperbolic differential operators can be found in,
e.g.  \cite[\S 1.5]{barWaveEquationsLorentzian2007}.
As an example, given the free field Lagrangian \eqref{eq:KG_Lagrangian},
$P$ is simply the Klein-Gordon operator $-(\square + m^2)$.

For interacting theories, one must take a further functional derivative,
defining
\begin{equation}
    \left\langle \mathcal{L}(f)^{(2)}[\phi], h \otimes g \right\rangle
        =:
    \left\langle S''[\phi], h \otimes g \right\rangle,
\end{equation}
where $f$ is chosen as before.
For a broad class of physically interesting actions,
there exists a self-adjoint,
Green hyperbolic differential operator
(\cite[Definition 3.2]{baerGreenhyperbolicOperatorsGlobally2015}) $P[\phi]$ such that
\begin{equation}
    \label{eq:linearisation_hypothesis}
    \left\langle S''[\phi], h \otimes g \right\rangle
        =
    \left\langle P[\phi] g, h \right\rangle.
\end{equation}
We refer to $P[\phi] g = 0$ as the \textit{linearised equations of motion}
at the configuration $\phi$ and, if such an operator exists for every
$\phi \in \mathfrak{E}(\mathcal{M})$, we say that the action satisfies the
\textit{linearisation hypothesis}.
If $\phi$ is an on-shell configuration,
then $\mathrm{Ker} \, P[\phi]$ can be thought of as the tangent space
at $\phi$ to the manifold of on-shell configurations.
Note that for a free action, $P$ coincides with $P[\phi]$ for every
$\phi \in \mathfrak{E}(\mathcal{M})$.

Throughout this paper we assume all spacetimes to be
\textit{globally hyperbolic}.
A Lorentzian manifold
$\mathcal{M} = (M, g)$
is globally hyperbolic if there exists a diffeomorphism
$\rho: M \overset{\sim}{\to} \Sigma \times \mathbb{R}$,
such that, for every $t \in \mathbb{R}$,
$\rho^{-1}(\Sigma \times \{t\})$ is a Riemannian submanifold
(referred to as a \textit{Cauchy surface}) of $\mathcal{M}$.

The key feature of such spacetimes is the existence of
Green hyperbolic differential operators $P$,
characterised by the property that the Cauchy problem $P \phi = 0$ 
admits fundamental solutions
$E^{R/A}: \mathfrak{D}(\mathcal{M}) \to \mathfrak{E}(\mathcal{M})$
uniquely distinguished by the fact that, for any $f \in \mathfrak{D}(\mathcal{M})$
\begin{align}
    P E^{R/A} f
    &=
    E^{R/A}  P f
    =
    f,\\
    \label{eq:propagator_supports}
    \mathrm{supp} \left( E^{R/A} f \right)
    &\subseteq
    \mathscr{J}^\pm  ( \mathrm{supp}( f ) ).
\end{align}
Here $\mathscr{J}^\pm(K)$ denotes the causal future/past of $K$,
i.e. the set of all points connected to some point $x \in K$ by
a causal future/past directed curve respectively.
We call these maps the \textit{retarded/advanced propagator} respectively.
For detailed explanation and proof of the relevant existence and uniqueness theorems,
we refer the reader again to \cite{barWaveEquationsLorentzian2007}.

Each propagator is formally adjoint to the other in the sense that,
for all $f, g \in \mathfrak{D}(\mathcal{M})$
\begin{equation}\label{eq:adjoint_pair}
    \left\langle f, E^R g \right\rangle
        =
    \left\langle g, E^A f \right\rangle.
\end{equation}
Their difference $E = E^R - E^A$,
known as the \textit{Pauli-Jordan function},
defines a map
from $\mathfrak{D}(\mathcal{M})$ to the space of solutions of $P \phi = 0$,
and is vital to our construction of a covariant Poisson structure on phase space.

Note that here and in the following we are considering a free theory,
governed by the single linear equation $P \phi = 0$.
However, to generalise to the interacting case, one need only replace $P$
with the linearised operator $P[\phi]$ defined by
\eqref{eq:linearisation_hypothesis},
and note that the fundamental solutions are then defined relative to this linearised operator.

Recall that the phase space of a free field theory is simply the space
$\mathrm{Ker} \, P$ of solutions to the equations of motion.
Traditionally, we identify this with the space of Cauchy data on some fixed surface,
i.e. the strength and momentum-density of a field at some fixed time.
\cite[Proposition 3.4.7]{barWaveEquationsLorentzian2007}
states that \textit{all} solutions with
spacelike-compact support may expressed as $Ef$ for some
$f \in \mathfrak{D}(\mathcal{M})$
and also that the kernel of this map is precisely $P(\mathfrak{D}(\mathcal{M}))$.
In other words, we can identify our phase space with the quotient
$\mathfrak{D}(\mathcal{M}) / P(\mathfrak{D}(\mathcal{M}))$.
One \textit{could} then define the algebra of observables on $\mathcal{M}$
to be the space of smooth maps from this space to $\mathbb{C}$,
which can be equipped with a non-degenerate Poisson bracket
using $E$ as a bivector.
This is not, however, the approach that we shall take,
which we outline below.

Given two regular functionals
$\mathcal{F}, \mathcal{G} \in \mathfrak{F}_\mathrm{reg}(\mathcal{M})$,
we can use $E$ to define a new functional
\begin{align}
    \label{eq:peierls_bracket}
    \left\{ \mathcal{F}, \mathcal{G}\right\} [\phi]
        :=
    \left\langle \mathcal{F}^{(1)}[\phi], E \mathcal{G}^{(1)}[\phi] \right\rangle
\end{align}
called the \textit{Peierls bracket} of $\mathcal{F}$ and $\mathcal{G}$,
where we recall that $\mathcal{F}^{(1)}[\phi]$ and $\mathcal{G}^{(1)}[\phi]$
may be identified with smooth test functions when
$\mathcal{F}$ and $\mathcal{G}$ are regular.
\textit{Local} functionals also possess this property,
hence we can define the Peierls bracket of local functionals,
though $\mathfrak{F}_\mathrm{loc}(\mathcal{M})$
is \textit{not} closed under this operation.

To obtain a closed algebra, we extend the domain of the Pauli-Jordan function
to include a suitable class of distributions.
As shown in \Cref{sec:closure_proofs},
the pairing $\left\langle f, Eg \right\rangle$
is well defined if $f$ and $g$ are compactly-supported distributions
satisfying the $(n = 1)$ wavefront set spectral condition
\eqref{eq:spectral_condition}.
In particular, this means \eqref{eq:peierls_bracket} is well defined for
$\mathcal{F}, \mathcal{G} \in \mathfrak{F}_{\mu c}(\mathcal{M})$,
and one can show (see \Cref{sec:closure_proofs}) that the result
is again a microcausal functional.
Once it is established that $\left\{ \cdot, \cdot \right\}$
is also a derivation over the pointwise product of functionals,
we may conclude that
$\left( \mathfrak{F}_{\mu c}(\mathcal{M}), \cdot, \{\cdot, \cdot\} \right)$
is a Poisson algebra \cite[Theorem 4.1.4]{brunettiAlgebraicStructureClassical2019},
which we shall denote $\mathfrak{P}(\mathcal{M})$.
In the next section, it is precisely this Poisson structure we shall deform
in order to arrive at the quantum algebra.

\subsection{Going On Shell}
\label{sec:classical_BV}

We shall now explain how this formalism distinguishes between
on-shell and off-shell observables.
Recall for the free theory we claimed that on-shell observables
could be defined as maps from
$\mathfrak{D}(\mathcal{M}) / P(\mathfrak{D}(\mathcal{M}))$
or equivalently the space of on-shell configurations,
to $\mathbb{C}$.
Broadly speaking, the strategy is to identify this space of maps
as a quotient of $\mathfrak{F}_{\mu c}(\mathcal{M})$
by a suitable ideal.

A well-known result states that,
given a manifold $X$ with some closed submanifold $Y \subseteq X$,
there is an isomorphism
\begin{equation}
    \mathcal{C}^\infty(Y)
        \simeq
    \mathcal{C}^\infty(X) / \mathcal{I}(Y),
\end{equation}
where $\mathcal{I}(Y) \subseteq \mathcal{C}^\infty(X)$ is
the ideal of functions vanishing on $Y$.
The construction of the Poisson algebra of on-shell observables
may be regarded as an infinite-dimenional analogue of this isomorphism,
where $\mathcal{C}^\infty(X)$ is replaced with $\mathfrak{F}_{\mu c}(\mathcal{M})$.
We define the ideal $\mathfrak{I}_S \subseteq \mathfrak{F}_{\mu c}(\mathcal{M})$
to be the set of functionals which vanish for all on-shell configurations,
i.e.
$\forall \mathcal{F} \in \mathfrak{I}_S, P \phi = 0 \Rightarrow \mathcal{F}[\phi] = 0$.

Crucially, $\mathfrak{I}_S$ is an ideal with respect not only
to the pointwise product $\cdot$, but also with respect
to the Peierls bracket $\{\cdot, \cdot\}$.
This can be proved from \eqref{eq:peierls_bracket} because,
if $\phi$ is a solution, $\mathcal{F} \in \mathfrak{I}_S$,
and $\mathcal{G} \in \mathfrak{F}_{\mu c}(\mathcal{M})$ then
$\phi + \epsilon E \mathcal{G}^{(1)}[\phi]$ is also a solution for any $\epsilon > 0$,
hence
\begin{equation}
    \mathcal{F}[\phi + \epsilon E \mathcal{G}^{(1)}[\phi]]
        =
    0,
\end{equation}
i.e.
$
    \{
        \mathcal{F},
        \mathcal{G}
    \}[\phi]
        =
    \left\langle 
          \mathcal{F}^{(1)}[\phi],
        E \mathcal{G}^{(1)}[\phi]
    \right\rangle
        =
    0,
$
indicating that
$
    \left\{ \mathfrak{F}_{\mu c}(\mathcal{M}), \mathfrak{I}_S \right\}
        \subseteq
    \mathfrak{I}_S
$
as desired.
Therefore, we may construct the quotient Poisson algebra
$\mathfrak{P}(\mathcal{M}) / \mathfrak{I}_S$
with the Poisson bracket given by
$\{[\mathcal{F}], [\mathcal{G}]\} := \left[\{\mathcal{F}, \mathcal{G}\}\right]$,
which we call the \textit{on-shell Peierls bracket}.

Defining the on-shell algebra as a quotient of two functional spaces, emphasises the algebraic viewpoint on geometry, where a space of maps on an algebraic variety or a topological vector space is used to describe the space itself. The advantage of this viewpoint will become even more apparent after we present a convenient way of characterising $\mathfrak{I}_S$.

We have already seen variations of the form
$
    \left\langle 
        S'[\cdot], h
    \right\rangle
$,
noting that an on-shell configuration $\phi$ is precisely one
for which the above functional vanishes,
for any $h \in \mathfrak{D}'(\mathcal{M})$.
We can identify $h$ with a constant section of the tangent bundle
$
    T \mathfrak{E}(\mathcal{M})
        \simeq
    \mathfrak{E}(\mathcal{M}) \times \mathfrak{D}(\mathcal{M}),
$
which we denote $X_h$.
Allowing such sections to act on functionals via derivation
(in the obvious way), we can rewrite the above functional as
$X_h \cdot \mathcal{L}(f)$ for any
$f \in \mathfrak{D}(\mathcal{M})$ which is suitable
in the manner specified after \eqref{eq:Euler_Lagrange}.
To discuss more general variations,
we must first discuss a suitable notion of a vector field.

A complete definition of the space of microcausal vector fields
requires a few subtleties, and may be found in
\cite[\S 4.4]{rejznerPerturbativeAlgebraicQuantum2016}.
There it is also noted how such vector fields are related
to the space of microcausal observables on the \textit{shifted cotangent bundle},
$T^*[1]\mathfrak{E}(\mathcal{M})$.
Let $\mathfrak{V}_{\mu c}(\mathcal{M})$ denote the space of
microcausal vector fields.
To every functional $\mathcal{F} \in \mathfrak{F}_{\mu c}(\mathcal{M})$
we can associate a \textit{one-form} $d \mathcal{F}$,
i.e. a smooth map $\mathfrak{V}_{\mu c}(\mathcal{M}) \to \mathbb{C}$ by
$d \mathcal{F}(X) = X \cdot \mathcal{F}$.
An important characteristic of any $X \in \mathfrak{V}_{\mu c}(\mathcal{M})$
is that there exists a compact subset $K \subset \mathcal{M}$
such that $X[\phi] \in \mathfrak{D}(K)$
for all $\phi \in \mathfrak{E}(\mathcal{M})$.
This means we can define a one-form
$\delta_S(X) = d\mathcal{L}(f) (X)$,
where $f \equiv 1$ on a neighbourhood $K$.
We call $\delta_S(X)$ the \textit{variation of the action with respect to} $X$.

The principle of critical action for $\phi \in \mathfrak{E}(\mathcal{M})$
can be expressed as the condition that,
\begin{equation}
    \delta_S(X)[\phi] \equiv 0,
        \qquad
    \forall X \in \mathfrak{V}_{\mu c}(\mathcal{M}).
\end{equation}
Hence, it is clear that all functionals which arise as a
variation of the action under a vector field must vanish on-shell,
in other words,
$\delta_S(\mathfrak{V}_{\mu c}(\mathcal{M})) \subseteq \mathfrak{I}_S(\mathcal{M})$.
If the action satisfies certain regularity conditions
\cite[\S 4.4]{henneauxLecturesAntifieldBRSTFormalism1990},
it is possible to show that all functionals vanishing on-shell arise this way,
i.e. the image of $\delta_S$ is \textit{precisely} $\mathfrak{I}_S(\mathcal{M})$.

We can begin to see some of the higher structure of this formalism by extending
the differential
$\delta_S: \mathfrak{V}_{\mu c}(\mathcal{M}) \to \mathfrak{F}_{\mu c}(\mathcal{M})$
to the exterior algebra of $\mathfrak{V}_{\mu c}(\mathcal{M})$
(graded such that the degree of
$\bigwedge^k \mathfrak{V}_{\mu c}(\mathcal{M})$ is $-k$).
This yields the cochain complex
\begin{equation}
    \cdots
        \overset{\delta_S}{\longrightarrow}
    {\bigwedge}^3 \mathfrak{V}_{\mu c}(\mathcal{M})
        \overset{\delta_S}{\longrightarrow}
    {\bigwedge}^2 \mathfrak{V}_{\mu c}(\mathcal{M})
        \overset{\delta_S}{\longrightarrow}
                \mathfrak{V}_{\mu c}(\mathcal{M})
        \overset{\delta_S}{\longrightarrow}
                \mathfrak{F}_{\mu c}(\mathcal{M})
        \longrightarrow
                0,
\end{equation}
where $\delta_S$ is defined in lower degrees via the graded Leibniz rule:
for example, a homogeneous element
$X \wedge Y \in \bigwedge^2 \mathfrak{V}_{\mu c}(\mathcal{M})$
is mapped to
$\delta_S(X \wedge Y) = \delta_S(X) Y - \delta_S(Y) X$.
We call this the \textit{Koszul complex associated to} $\delta_S$,
denoted $\mathfrak{K}(\delta_S)$.

One can show that the Peierls bracket also extends to a degree zero
Poisson bracket across the entire complex,
and that $\delta_S$ is a derivation over this bracket
(i.e. the pair $\left( \mathfrak{K}(\delta_S), \left\{ \cdot, \cdot \right\} \right)$
is a dg Poisson algebra).
In particular, for a vector field $X \in \mathfrak{V}_{\mu c}(\mathcal{M})$
and a functional $\mathcal{F} \in \mathfrak{F}_{\mu c}(\mathcal{M})$,
this means that
$\delta_S \left\{ X, \mathcal{F} \right\} = \left\{ \delta_S X, \mathcal{F} \right\}$
(as $\delta_S \mathcal{F} = 0$ for any functional $\mathcal{F}$).
In turn, this establishes that $\delta_S(\mathfrak{V}_{\mu c}(\mathcal{M}))$
is an ideal of the Peierls bracket,
and hence that the cohomology of this complex in degree 0
naturally inherits a Poisson structure.
Given the fact that $\delta_S(\mathfrak{V}_{\mu c}(\mathcal{M})) = \mathfrak{I}_S$,
this cohomology is
$H^0(\mathfrak{K}(\delta_S)) = \mathfrak{F}_{\mu c}(\mathcal{M}) / \mathfrak{I}_S$,
which we thus call the
\textit{on-shell algebra of observables}.

It is, at this point, natural to ask whether or not there exists
a physical interpretation of $H^{-1}(\mathfrak{K}(\delta_S))$,
or the cohomology in yet lower degrees.
To answer the first, note that for a vector field $X$,
$\delta_S(X) = 0$ implies that the infinitesimal transformation
$\phi \mapsto \phi + \epsilon X[\phi]$
leaves the action invariant to first order in $\epsilon$.
As such, the kernel of $\delta_S$ in degree $-1$ comprises
infinitesimal generators of \textit{gauge symmetries}.
The image of $\delta_S$ in degree $-1$
contains vector fields of the form
$\delta_S(X \wedge Y) = \delta_S(X) Y - \delta_S(Y) X$.
In the physics literature these are referred to as
\textit{trivial gauge symmetries}.
They are, in a sense, less insightful
because they are defined the same way regardless of the action in question,
and also because they act trivially on shell.
As such, we can regard $H^{-1}(\mathfrak{K}(\delta_S))$ as the space of
\textit{non-trivial gauge symmetries}
\footnote{
    In principle, one can go further
    \cite[Introduction \S 3.2]{costelloFactorizationAlgebrasQuantum2016},
    interpreting elements of
    $H^{-2}(\mathfrak{K}(\delta_S))$ as ``symmetries between symmetries'',
    however, such notions are tricky to formulate precisely and are
    well beyond the scope of this article.
}.

The above discussion motivates us to consider the space
$\bigwedge^\bullet \mathfrak{V}_{\mu c}$
as the primary kinematical object of a physical theory,
with $\delta_S$ representing the choice of dynamics.
This perspective is advantageous both in describing
conformally covariant field theories
(where the generalised Lagrangian formalism proves inconvenient)
as well as in the formulation of chiral sectors of a 2-dimensional CFT,
where one may require choices of $\delta_S$ which
cannot arise from a generalised Lagrangian.

Finally, as an aside
now that we have constructed our on-shell algebra,
it is informative to make a comparison to the `equal-time' (a.k.a. canonical) bracket
defined relative to some choice of Cauchy surface $\Sigma$.

\begin{definition}[Canonical Poisson Algebra]
    Let $\Sigma \subset \mathcal{M}$ be a Cauchy surface,
    we define the associated \emph{canonical Poisson algebra} as follows:
    The underlying vector space $\mathfrak{F}_\mathrm{can}(\Sigma)$ consists of functionals
    $F: \mathcal{C}^\infty_c(\Sigma) \times \mathcal{C}^\infty_c(\Sigma) \to \mathbb{C}$ which are Bastiani smooth, the arguments of this functional represent the initial field strength and momentum on $\Sigma$ of some on-shell field configuration.
    Given a pair $F, G$ of such functionals, their canonical bracket is then defined as
    \begin{equation}
        \{
            F, G\}_\mathrm{can}[\varphi, \pi]
            :=
        \int_\mathrm{\Sigma}\left[
                \frac{\delta F[\varphi, \pi]}{\delta \varphi (x)}
                \frac{\delta G[\varphi, \pi]}{\delta \pi  (x)}
                    -
                \frac{\delta G[\varphi, \pi]}{\delta \varphi (x)}
                \frac{\delta F[\varphi, \pi]}{\delta \pi  (x)}
        \right] \mathrm{dVol}_\Sigma.
    \end{equation}
\end{definition}

It is not immediately obvious why
the Peierls bracket should be related to this canonical bracket,
other than because $E$ parametrises the space of on-shell field configurations.
Especially as the canonical bracket requires
a particular Cauchy surface to be specified,
a manifestly Lorentz non-covariant choice.
However,
by sending the initial data
$(\varphi, \pi) \in \mathfrak{E}(\Sigma) \times \mathfrak{E}(\Sigma)$,
to their corresponding solution, one can construct a map
$\mathfrak{F}_{\mu c}(\mathcal{M}) \to \mathfrak{F}_\mathrm{can}(\Sigma)$
which in turn yields a Poisson algebra homomorphism
from the on-shell Peierls bracket to the canonical
\cite[\S 3.2]{fredenhagenPerturbativeConstructionModels2015}.

\subsection{Deformation Quantisation}
\label{sec:deformation_quantisation}

Having established our Poisson structure, the next step is to 
deform it to construct our \textit{quantum} algebra of observables.
Here we take an approach that is analogous to Moyal-Weyl quantisation,
though in QFT this is made somewhat harder than in the quantum-mechanical case,
due to the infinite degrees of freedom in the configuration space.
In particular, as is common in perturbative QFT,
our deformation shall be formal, meaning that quantised products will be
formal power series in $\hbar$,
allowing us to ignore the issue of proving convergence of our formulae.

For \textit{regular} functionals
$\mathcal{F}, \mathcal{G} \in \mathfrak{F}_\mathrm{reg}(\mathcal{M})$
we can define the \textit{star product} of
$\mathcal{F}$ and $\mathcal{G}$ directly as
\begin{equation}
    (
        \mathcal{F} \star
        \mathcal{G}
    )[\phi]
        =
    \mathcal{F}[\phi]\mathcal{G}[\phi]
        +
    \sum_{n \geq 1}
        \left( \frac{i \hbar}{2} \right)^n
        \frac{1}{n!}
        \left\langle 
            E^{\otimes n},
            \mathcal{F}^{(n)}[\phi] \otimes
            \mathcal{G}^{(n)}[\phi]
        \right\rangle.
\end{equation}
We may write this formula more concisely as
\begin{equation}\label{eq:star_prod}
    \mathcal{F} \star \mathcal{G}
        :=
    m \circ e^{\tfrac{i \hbar}{2}
    \left\langle
        E, \tfrac{\delta}{\delta \phi} \otimes \tfrac{\delta}{\delta \phi} 
    \right\rangle}
    \left( \mathcal{F} \otimes \mathcal{G} \right),
\end{equation}
where $m$ is the pointwise multiplication map
\begin{equation*}
        m(\mathcal{F} \otimes \mathcal{G})[\phi]
            :=
        (\mathcal{F} \otimes \mathcal{G})
        [\phi \otimes \phi]
            =
        \mathcal{F}[\phi] \cdot \mathcal{G}[\phi].
\end{equation*}
A general result
\cite[Proposition 4.5]{hawkinsStarProductInteracting2019}
states that this exponential form guarantees $\star$ is associative.
As mentioned, this deformation is formal,
meaning we have actually defined a map
$
    \star:
    \mathfrak{F}_\mathrm{reg}(\mathcal{M})
        \otimes
    \mathfrak{F}_\mathrm{reg}(\mathcal{M})
        \to
    \mathfrak{F}_\mathrm{reg}(\mathcal{M})[[\hbar]].
$
We can then define the $\star$ product on
$\mathfrak{F}_\mathrm{reg}(\mathcal{M})[[\hbar]]$ by linearity
to obtain a closed algebra.

Writing the first few terms explicitly, we see
$
    \mathcal{F} \star \mathcal{G}
        =
    \mathcal{F} \cdot \mathcal{G}
        + \frac{i \hbar}{2} \left\{\mathcal{F}, \mathcal{G}\right\}
        + \mathcal{O}(\hbar^2).
$
Thus the classical term of $\star$ (i.e. the coefficient of $\hbar^0$)
is simply the pointwise product
and the Dirac quantisation rule also holds modulo terms of order $\hbar^2$,
hence $\star$ is a deformation of the classical product in the sense of
\cite[\S 5.1]{rejznerPerturbativeAlgebraicQuantum2016}.
However, if we wished to apply \eqref{eq:star_prod} to other local functionals,
divergences would begin to appear.
Consider for example the family of quadratic functionals,
for $f \in \mathfrak{D}(\mathcal{M})$
\begin{equation}
    \Phi^2(f)[\phi]
        :=
    \int_\mathcal{M} f(x) \phi^2(x) \, \mathrm{dVol}_x.
\end{equation}
A na\"ive computation of the star product for two such functionals would yield
\begin{align}
    \begin{split}
        \Phi^2(f) \star \Phi^2(g)
            \,\text{``}\!=\!\text{''} \,
        &\Phi^2(f) \cdot \Phi^2(g)
        + \tfrac{i \hbar}{2} \left\{ \Phi^2(f), \Phi^2(g) \right\} \\
        &- \frac{\hbar^2}{2} \int_{\mathcal{M}^2}
            f(x) E^2(x; y) g(y)
            \, \mathrm{dVol}_x \, \mathrm{dVol}_y.
    \end{split}
\end{align}
In general, the $\mathcal{O}(\hbar^2)$ term of this product is ill-defined if
$\mathrm{supp} f \cap \mathrm{supp} g \neq \emptyset$.
This is because $E$ is a distribution, as opposed to a smooth function,
and the product of two distributions cannot be defined in general.

The solution is to make use of a \textit{Hadamard distribution}.
Physically, a Hadamard distribution is
the $2$-point correlator function for some `vacuum-like' state,
i.e. $W(x_1, x_2) = \left\langle \Phi(x_1) \star \Phi(x_2) \right\rangle$.
More precisely, a complex-valued distribution
$W \in \mathfrak{D}'(\mathcal{M}^2; \mathbb{C})$ is
\textit{Hadamard} if it satisfies the following properties
\cite{rejznerPerturbativeAlgebraicQuantum2016}
\begin{itemize}
    \item[\textbf{H0}]%
            The wavefront set of $W$ satisfies
            \begin{equation}
                \label{eq:Hadamard_WFS}
                \mathrm{WF}(W)
                    =
                \left\{ 
                    (x, y; \xi, \eta) \in \mathrm{WF}(E)
                        \,|\,
                    (x; \xi) \in \overline{V}_+
                \right\}
            \end{equation}
    \item[\textbf{H1}]%
            $
                W = \tfrac{i}{2}E + H,
            $
            where $H$ is a symmetric, real distribution.
    \item[\textbf{H2}]%
        $W$ is a weak solution to $P$.
    \item[\textbf{H3}]%
        $W$ is positive semi-definite in the sense that,
        $\forall f \in \mathfrak{D}(\mathcal{M}; \mathbb{C})$
        $\left\langle W, \bar{f} \otimes f \right\rangle \geq 0$.
\end{itemize}
Importantly, property \textbf{H0} implies that $W$ satisfies
the H\"ormander criterion \cite[Theorem 8.2.10]{hormanderAnalysisLinearPartial2015},
ensuring that pointwise powers $W^n$ are well-defined.

A choice of Hadamard distribution yields a corresponding star product by
\begin{equation}
    \label{eq:star_H_prod_defn}
    \mathcal{F} \star_H \mathcal{G}
        :=
    m \circ e^{\left\langle \hbar W, \tfrac{\delta}{\delta \phi} \otimes \tfrac{\delta}{\delta \phi} \right\rangle}
        \left( \mathcal{F} \otimes \mathcal{G} \right).
\end{equation}

Note that any freedom in the choice of a Hadamard state $W$ lies
solely in the choice of its symmetric part $H$.
As such, we shall denote by $\mathrm{Had}(\mathcal{M})$
the set of bi-distributions $H$ such that
$\tfrac{i}{2}E + H$ is a Hadamard distribution as per the above definition.

The product $\star_H$ is well-defined for regular functionals
for all $H \in \mathrm{Had}(\mathcal{M})$,
where it is in fact isomorphic to $\star$:
if we define the map
$
    \alpha_H:
    \mathfrak{F}_\mathrm{reg}(\mathcal{M}) \to
    \mathfrak{F}_\mathrm{reg}(\mathcal{M})
$ by
\begin{equation}
    \label{eq:alpha_map}
    \alpha_H \mathcal{F}
        = 
    e^{\tfrac{\hbar}{2}\big\langle H, \tfrac{\delta^2}{\delta \phi^2} \big\rangle}
        \mathcal{F},
\end{equation}
then
$
    \alpha_H \left( \mathcal{F} \star \mathcal{G} \right)
        =  
    \left( \alpha_H \mathcal{F} \right) \star_H
    \left( \alpha_H \mathcal{G} \right),
$
for any 
$\mathcal{F}, \mathcal{G} \in \mathfrak{F}_\mathrm{reg}(\mathcal{M})$
and the inverse of this map is simply $\alpha_{-H}$.
Where these two products differ, however,
is that $\star_H$ can also be extended to a well defined product
on $\mathfrak{F}_{\mu c}(\mathcal{M})$.

On a generic globally hyperbolic spacetime,
it is well-known \cite{fullingSingularityStructureTwopoint1981}
that there exist infinitely many Hadamard distributions,
thus we need never fear that $\mathrm{Had}(\mathcal{M})$ is empty.
However, there is usually no natural way of selecting
\textit{which} $H \in \mathrm{Had}(\mathcal{M})$ to use.
Thus, whilst we can always construct a well defined algebra
\begin{equation}
    \label{eq:A^H_algebras}
    (\mathfrak{F}_{\mu c}(\mathcal{M})[[\hbar]], \star_H)
        =:
    \mathfrak{A}^H(\mathcal{M})
\end{equation}
for an arbitrary globally hyperbolic spacetime $\mathcal{M}$,
it would be unnatural to define \textit{the} quantum algebra
by making such an arbitrary choice.
Fortunately, the algebraic structure of
$\mathfrak{A}^H(\mathcal{M})$
is actually independent of the Hadamard distribution selected:
if $H, H' \in \mathrm{Had}(\mathcal{M})$, then
\begin{equation}\label{alpha-H-Hp-relation}
    \alpha_{H' - H} \left( \mathcal{F} \star_{H} \mathcal{G} \right)
        =  
    \left( \alpha_{H' - H} \mathcal{F} \right) \star_{H'}
    \left( \alpha_{H' - H} \mathcal{G} \right),
\end{equation}
where
$\alpha_{H' - H}: \mathfrak{A}^H(\mathcal{M}) \to \mathfrak{A}^{H'}(\mathcal{M})$
is defined just as in \eqref{eq:alpha_map}.
As one might expect, the inverse of this map is $\alpha_{H - H'}$,
hence all of our candidate algebras are in fact \textit{isomorphic} to one another.
One way in which we can define the quantum algebra without
any undue preference to a particular Hadamard distribution is as follows.

\begin{definition}
    The \textit{quantum algebra of the free field theory},
    denoted $\mathfrak{A}(\mathcal{M})$,
    is a unital, associative $*$-algebra whose elements are the indexed sets
    $\mathcal{F} = \left( \mathcal{F}_H \right)_{H \in \mathrm{Had}(\mathcal{M})}$,
    subject to the compatibility criterion
    \begin{equation}\label{eq:compatibility}
        \mathcal{F}_{H'} = \alpha_{H' - H} \mathcal{F}_H,
    \end{equation}
    with a product defined by
    \begin{equation}
        \mathcal{F}
            \star
        \mathcal{G}
            :=
        \left(
            \mathcal{F}_H \star_H \mathcal{G}_H
        \right)_{H \in \mathrm{Had}(\mathcal{M})}.
    \end{equation}
\end{definition}

It is important to bear in mind that,
whilst we have deformed the classical algebra $\mathfrak{F}_{\mu c}(\mathcal{M})$
into a quantum algebra $\mathfrak{A}(\mathcal{M})$,
we have not yet specified a \textit{quantisation map},
embedding classical observables into the quantum algebra.
We will need to establish such a map before computing commutation relations
for the quantum stress energy tensor in \cref{sec:Virasoro}.
However, before considering what a suitable choice of map may be,
it is instructive to study how the construction we have just outlined
varies as we change the underlying spacetime $\mathcal{M}$.

\subsection{Local Covariance and Normal Ordering}
\label{sec:normal_ordering}

We have deliberately said little about
potential spacetime symmetries in the construction above.
The reason being that we take the perspective that covariance
under any symmetries a particular spacetime may enjoy is just
a special case of a broader property we wish to implement:
namely \textit{local covariance}.
The concept of local covariance,
introduced in \cite{hollandsLocalWickPolynomials2001} and
\cite{brunettiGenerallyCovariantLocality2003},
unites the representation of spacetime symmetries as
automorphisms of the algebra of observables with
the principle that an observable localised to a region
$\mathcal{O} \subset \mathcal{M}$
of a spacetime should be `unaware' of the structure of the spacetime
beyond this region.

The foundational idea is that, if there exists a suitable embedding of
a spacetime $\mathcal{M}$ into a spacetime $\mathcal{N}$,
then there should be a corresponding embedding
(more precisely, a \textit{homomorphism}) of observables
$\mathfrak{A}(\mathcal{M}) \to \mathfrak{A}(\mathcal{N})$.
A spacetime symmetry is just a suitable embedding of $\mathcal{M}$ into itself
which also admits an inverse.
If the corresponding algebra homomorphism is similarly invertible,
then we would have, in particular,
an action of the isometry group of $\mathcal{M}$ on $\mathfrak{A}(\mathcal{M})$
as desired.

To formulate local covariance more precisely,
it is convenient to invoke the language of category theory.
To begin with, by specifying the suitable embeddings of spacetimes,
we endow the collection of globally hyperbolic spacetimes
with the structure of a \textit{category},
which is denoted $\mathsf{Loc}$ and defined as follows:
\begin{itemize}
    \item   An object of $\mathsf{Loc}$ is a \textit{spacetime} $\mathcal{M}$,
            as specified in \cref{def:spacetime}, of a fixed dimension $d$.
    \item   For a pair of spacetimes
            $\mathcal{M} = (M, g, \mathfrak{o}, \mathfrak{t})$ and
            $\mathcal{N} = (N, g', \mathfrak{o}', \mathfrak{t}')$,
            a morphism $\chi: \mathcal{M} \to \mathcal{N}$ is a smooth embedding
            $\chi: M \hookrightarrow N$ which is \textit{admissible} in the sense that
            $\chi^* g' = g$,
            $\mathfrak{o} = \chi^* \mathfrak{o}'$, and
            $\mathfrak{t} = \chi^* \mathfrak{t}'$.
\end{itemize}

Given an admissible embedding $\chi: \mathcal{M} \to \mathcal{N}$,
there is a natural map
$\mathfrak{F}_{\mu c}(\mathcal{M}) \to \mathfrak{F}_{\mu c}(\mathcal{N})$
defined by
$\mathcal{F} \mapsto \chi_* \mathcal{F} := \mathcal{F} \circ \chi^*$.
We show later in \Cref{sec:conformally_covariant_classical}
that even if $\chi$ preserves the metric only up to a scale,
then $\mathcal{F}\circ\chi^*$ is still microcausal
whenever $\mathcal{F}$ is,
hence in particular
$
    \chi_*\left( \mathfrak{F}_{\mu c}(\mathcal{M}) \right)
    \subset \mathfrak{F}_{\mu c}(\mathcal{N})
$
for all $\mathsf{Loc}$ morphisms $\chi: \mathcal{M} \to \mathcal{N}$.
In fact, all of the different spaces of functionals specified in
\Cref{sec:kin} are each preserved under the map $\chi_*$,
and thus may be considered functors from $\mathsf{Loc}$ to
some category of observables.

Next, we need to find a way to specify dynamics in a coherent way
across all spacetimes.
This involves extending the generalised Lagrangian framework to
the concept of a \textit{natural Lagrangian}.
In categorical language, we can define
a natural Lagrangian as a natural transformation
$\mathcal{L}: \mathfrak{D} \Rightarrow \mathfrak{F}_\mathrm{loc}$,
such that for each $\mathcal{M} \in \mathsf{Loc}$,
$\mathcal{L}_\mathcal{M}$ is a generalised Lagrangian
as per \Cref{def:generalised_Lagrangian}.
Here, $\mathfrak{D}$ is the functor assigning each spacetime its
space of compactly-supported test functions,
and to each morphism $\chi: \mathcal{M} \to \mathcal{N}$
the map $\chi_*: \mathfrak{D}(\mathcal{M}) \to \mathfrak{D}(\mathcal{N})$
defined by
\begin{equation}
    \label{eq:test_function_pushforward}
    \chi_*f(y)
        =
    \begin{cases}
        f(\chi^{-1}(y)) & \text{ if } y \in \chi(\mathcal{M}),\\
        0             & \text{else}.
    \end{cases}
\end{equation}

Spelling this out, the naturality condition reduces to the condition that,
for every morphism of spacetimes $\chi: \mathcal{M} \to \mathcal{N}$,
$f \in \mathfrak{D}(\mathcal{M})$ and $\phi \in \mathfrak{E}(\mathcal{N})$
\begin{equation}
    \label{eq:locally_covariant_classical_field}
    \mathcal{L}_\mathcal{N}(\chi_*f)[\phi] = \mathcal{L}_\mathcal{M}(f)[\chi^*\phi],
\end{equation}
which is essentially a generalisation of the covariance condition appearing in
\Cref{def:generalised_Lagrangian} and can be shown to be satisfied by
the Klein-Gordon Lagrangian \eqref{eq:KG_Lagrangian}.

From the naturality condition, one can then show that
if $\chi: \mathcal{M} \to \mathcal{N}$,
then the Euler-Lagrange derivatives of
$\mathcal{L}_\mathcal{M}$ and $\mathcal{L}_\mathcal{N}$
are related by the equation, $\forall \phi \in \mathfrak{E}(\mathcal{N})$
\begin{equation}
    \label{eq:Euler_Lagrange_covariance}
    \chi^* S_\mathcal{N}'[\phi] = S'_\mathcal{M}[\chi^* \phi]
\end{equation}
and, in the case of the free scalar field,
the causal propagators arising from
$\mathcal{L}_\mathcal{M}$ and $\mathcal{L}_\mathcal{N}$
are related by
$E_\mathcal{N}(\chi_*f, \chi_*g) = E_\mathcal{M}(f, g)$.
From here, it can be deduced that
$\chi_*: \mathfrak{F}_{\mu c}(\mathcal{M}) \to \mathfrak{F}_{\mu c}(\mathcal{N})$
is a Poisson algebra homomorphism where each space is equipped with
its respective Peierls bracket, hence the assignment
$\mathfrak{P}(\mathcal{M})$ outlined in
the above section is \textit{covariant}
(i.e. it defines a functor from $\mathsf{Loc}$ to the category of Poisson algebras).

We shall use the generic designation $\mathsf{Obs}$
to denote the category our observables (either classical or quantum) belong to.
Choices of $\mathsf{Obs}$ relevant to our discussion include
\begin{itemize}
    \item   $\mathsf{Vec}$, whose objects are vector spaces over $\mathbb{C}$,
            and whose morphisms are linear maps.
            This is the most generic space generally considered,
            and is appropriate when one wishes to treat 
            classical and quantum theories on an equal footing.
    \item   $\mathsf{Poi}$ the category of Poisson algebras and
            Poisson algebra homomorphisms.
            This is the primary category of observables for
            classical theories.
    \item   $*\text{-}\mathsf{Alg}$, the space of topological $*$-algebras.
            We choose this as the target category of quantum theories,
            as the perturbative nature of our construction requires us to consider
            unbounded operators, else we would use instead the category
            of $C^*$-algebras.
    \item   In each of the above cases, we may add a \textit{dg-structure},
            i.e. if $\mathsf{Obs}$ is any of the above categories,
            $\mathsf{Ch(Obs)}$ comprises cochain complexes which
            in each degree take values in $\mathsf{Obs}$.
            Such categories are at the heart of the BV formalism
            in both the classical and quantum case
            \cite{gwilliamRelatingNetsFactorization2019},
            \cite{costelloFactorizationAlgebrasQuantum2016}.
\end{itemize}

A \textit{locally covariant field theory} (classical or quantum)
is then defined simply as a functor from $\mathsf{Loc} \to \mathsf{Obs}$.
Already this captures a lot of important features,
such as the representation of spacetime symmetries as
automorphisms of the algebra of observables.
Whilst one can go further by imposing additional axioms for such a functor
to satisfy, this general definition will suffice for our purposes.

The BV formalism outlined in the previous section can also be made locally covariant.
Just like $\mathfrak{F}_{\mu c}$, we can easily promote
$\mathfrak{V}_{\mu c}$ to a functor $\mathsf{Loc} \to \mathsf{Vec}$.
A choice of natural Lagrangian then yields a natural transformation
between the two, $\delta_S: \mathfrak{V}_{\mu c} \Rightarrow \mathfrak{F}_{\mu c}$.
From this it follows that the construction of the Koszul complex
$\mathfrak{K}(\delta_S)$ itself defines a functor
$\mathsf{Loc} \to \mathsf{Ch(Poi)}$.

We have already sketched an explanation of how our construction of
the classical theory may be made locally covariant.
If $H_0 \in \mathrm{Had}(\mathcal{N})$,
then one can show that $\chi^* H_0 \in \mathrm{Had}(\mathcal{M})$,
thus we can define a map
$\mathfrak{A}^{(\chi^* H_0)}(\mathcal{M}) \to \mathfrak{A}^{H_0}(\mathcal{N})$
as just the canonical extension of the pushforward
$\chi_*: \mathfrak{F}_{\mu c}(\mathcal{M}) \to \mathfrak{F}_{\mu c}(\mathcal{N})$
to formal power series in $\hbar$.
This map satisfies
\begin{equation}
    \chi_*(\mathcal{F} \star_{(\chi^* H_0)} \mathcal{G})
        =
    \chi_* \mathcal{F} \star_{H_0} \chi_* \mathcal{G}
\end{equation}
thus it defines a $*$-algebra homomorphism.
The map $\mathfrak{A}\chi: \mathfrak{A}(\mathcal{M}) \to \mathfrak{A}(\mathcal{N})$
is then given by
\begin{equation}
    \left(
        \mathfrak{A}\chi \mathcal{F}
    \right)_{H_0}
        =
    \chi_* \mathcal{F}_{(\chi^* H_0)},
\end{equation}
which can be shown to satisfy the criteron \eqref{eq:compatibility},
making the map well-defined.
With these morphisms, we can then declare
$\mathfrak{A}: \mathsf{Loc} \to \mathsf{Obs}$
to be a \textit{locally covariant quantum field theory}.

Next, we turn to the topic of \textit{normal ordering}.
On a fixed spacetime $\mathcal{M}$, normal ordering is
the process of mapping (some subset of) classical observables
into the space of quantum observables.
In our case, we seek a map
$
    \nord{-}_\mathcal{M}:
    \mathfrak{F}_\mathrm{loc}(\mathcal{M})
        \to
    \mathfrak{A}(\mathcal{M}),
$
such that the $\hbar^0$ coefficient of $\nord{\mathcal{F}}_\mathcal{M}$
is $\mathcal{F}$.
Given our somewhat indirect definition of $\mathfrak{A}(\mathcal{M})$,
it is helpful to outline here the general strategy for defining
a normal ordering prescription,
before we turn our attention to any particular maps.

It is easiest to define a normal ordering prescription by a family of map
$\mathfrak{F}_\mathrm{loc}(\mathcal{M}) \to \mathfrak{A}^H(\mathcal{M})$
for every choice of $H \in \mathrm{Had}(\mathcal{M})$.
Suppose we denote each such map as $\mathcal{F} \mapsto (: \mathcal{F} :)_H$,
they collectively define a map
$\mathfrak{F}_\mathrm{loc}(\mathcal{M}) \to \mathfrak{A}(\mathcal{M})$
if they satisfy, for every $H, H' \in \mathrm{Had}(\mathcal{M})$,
$\mathcal{F} \in \mathfrak{F}_\mathrm{loc}(\mathcal{M})$
\begin{equation}
    \label{eq:normal_ordering_compatibility}
    (:\mathcal{F}:)_H = \alpha_{H - H'} (: \mathcal{F} :)_{H'}.
\end{equation}

By choosing a fixed Hadamard state $H_0 \in \mathrm{Had}(\mathcal{M})$,
we can define a quantisation map which has the physical interpretation of
normal ordering ``with respect to'' that state.
As indicated above, we first define a map
$\mathfrak{F}_{\mathrm{loc}}(\mathcal{M}) \to \mathfrak{A}^H(\mathcal{M})$ by
\begin{equation}
    \mathcal{F} \mapsto \alpha_{H - H_0} \mathcal{F} =: (\altnord{\mathcal{F}}_{H_0})_H.
\end{equation}
This clearly satisfies the criterion \eqref{eq:normal_ordering_compatibility} above,
and hence is a valid normal ordering prescription.
We may also characterise this prescription as the only consistent choice
such that the map
$\mathfrak{F}_\mathrm{loc}(\mathcal{M}) \to \mathfrak{A}^{H_0}(\mathcal{M})$
is simply the inclusion of $\mathfrak{F}_\mathrm{loc}(\mathcal{M})$
into $\mathfrak{F}_{\mu c}(\mathcal{M})[[\hbar]]$,
the underlying vector space of $\mathfrak{A}^{H_0}(\mathcal{M})$.

Similar to our definition of a natural Lagrangian,
a \textit{locally covariant ordering prescrition} is defined to be
a natural transformation from $\mathfrak{F}_\mathrm{loc}$ to $\mathfrak{A}$.
(%
    Note that we must assume that the target category of each functor is
    $\mathsf{Vec}$,
    as normal ordering is linear, but not a homomorphism.%
)
Explicitly, this naturality condition is realised by the equation,
for every admissible embedding $\chi: \mathcal{M} \to \mathcal{N}$,
\begin{equation}
    \nord{\chi_* \mathcal{F}}_\mathcal{N}
        =
    \mathfrak{A}\chi \left( \nord{\mathcal{F}}_\mathcal{M} \right).
\end{equation}

It is tempting to believe that a covariant prescription across all spacetimes can
be found by selecting a suitable Hadamard state for each spacetime.
However, it is now a well-established fact that
such a choice cannot be made consistently across all spacetimes.
(See the remarks following definition 3.2 of
\cite{hollandsLocalWickPolynomials2001} for a discussion relevant to the scalar field,
and \cite[\S6.3]{fewsterDynamicalLocalityCovariance2012} for a more general result.)

The solution is to instead define an ordering prescription dependant upon
the Hadamard \textit{parametrix} of the spacetime in question.
Before the characterisation via wavefront sets used in \eqref{eq:Hadamard_WFS},
Hadamard states were defined by the ability to express them locally
(i.e. in some neighbourhood of the thin diagonal $\Delta \subset \mathcal{M}^2$)
in what is known as \textit{local Hadamard form}.
In the case of a $2$-dimensional spacetime, the local Hadamard form of a state $W$ is
\begin{equation}\label{eq:local_Hadamard_form}
    W(x; y)
        :=
    -\frac{1}{4 \pi}
    \lim_{\epsilon \searrow 0} \left(
        \left( \sum_{n = 0}^{N}
            v_n(x; y)\sigma(x; y)^n
        \right)
        \log \left(\frac{\sigma_\epsilon(x; y)}{\lambda^2}\right)
        + w_N(x; y)
    \right),
\end{equation}
where $\sigma(x; y)$ is the world function,
defined as half the geodesic between $x$ and $y$,
$t$ is some choice of a time function
(i.e. level sets of $t$ are Cauchy surfaces),
$\sigma_{\epsilon}$ is defined by
\begin{equation}
    \sigma_\epsilon(x; y)
        :=
    \sigma(x; y) + 2i \epsilon\left( t(x) - t(y) \right) + \epsilon^2,
\end{equation}
$w_N$ is some $2N + 1$ times continuously differentiable function,
and the functions $v_n$ are all smooth, symmetric, and
determined entirely by the metric of $\mathcal{M}$
\cite{hollandsLocalWickPolynomials2001}.

The series of distributions
$
    \left(W^\mathrm{sing}_N := W - w_N
    \right)_{N \in \mathbb{N}},
$
constitute the Hadamard parametrix,
which is independent of the choice of state.
The parametrix defines a normal ordering prescription,
first as a map
$\mathfrak{F}_\mathrm{loc}(\mathcal{M}) \to \mathfrak{A}^H(\mathcal{M})$
\begin{equation}
    \label{eq:locally_covariant_ordering}
    \left(\nord{\mathcal{F}}_\mathcal{M}\right)_H
        =
    \lim_{N \to \infty} \alpha_{H - H^\mathrm{sing}_N} \mathcal{F},
\end{equation}
where $H^\mathrm{sing}_N = W^\mathrm{sing}_N - \tfrac{i}{2}E$.
This map is defined for any local functional $\mathcal{F}$ because
the order $N$ at which we must truncate the series in
\eqref{eq:local_Hadamard_form}
depends only on the \textit{order} of the functional $\mathcal{F}$.
This corresponds to the highest order derivative of a field
configuration $\phi$ which enters into the definition of $\mathcal{F}[\phi]$,
and is guaranteed to be finite
\cite[\S 6.2.2]{rejznerPerturbativeAlgebraicQuantum2016}.
For instance, if $\mathcal{F}$ has order $n$, then
$
    \alpha_{H - H^\mathrm{sing}_N} \mathcal{F}  
        =
    \alpha_{H - H^\mathrm{sing}_n} \mathcal{F}
$
for all $N \geq n$, thus this series always converges in finite time.
From now on we shall supress both the truncation of the series,
as well as the limit in \eqref{eq:locally_covariant_ordering}.
Instead we shall write
$\left(\nord{\mathcal{F}}_\mathcal{M}\right)_H = \alpha_{H - H^\mathrm{sing}} \mathcal{F}$,
where one may interpret $H^\mathrm{sing}$ as
$H^\mathrm{sing}_N$ for a sufficiently large $N$.

We can then verify that, for $H, H' \in \mathrm{Had}(\mathcal{M})$
\begin{equation}
    (\nord{\mathcal{F}}_\mathcal{M})_{H}
        =
    \alpha_{H - H^\mathrm{sing}} \mathcal{F}
        =
    \alpha_{H - H'} \circ \alpha_{H' - H^\mathrm{sing}} \mathcal{F}
        =
    \alpha_{H - H'} \left( \nord{\mathcal{F}}_\mathcal{M} \right)_{H'},
\end{equation}
i.e. the family of functionals
$
    \big(
        \left( \nord{\mathcal{F}}_\mathcal{M} \right)_H
    \big)_{
        H \in \mathrm{Had}(\mathcal{M})
    }
$
satisfies the compatibility criterion \eqref{eq:compatibility},
hence the map
$
    \nord{-}_\mathcal{M}:
    \mathfrak{F}_\mathrm{loc}(\mathcal{M}) \to
    \mathfrak{A}(\mathcal{M})
$
is well defined.

Crucially, the Hadamard parametrix is also locally covariant
\footnote{
    This is a direct consequence of the fact that
    $\chi^*: \mathrm{Had}(\mathcal{N}) \to \mathrm{Had}(\mathcal{M})$,
    and that the difference of any pair of elements in $\mathrm{Had}(\mathcal{M})$
    is smooth.
}:
If $H^\mathrm{sing}_{\mathcal{M} / \mathcal{N}}$ are the
(symmetrised) Hadamard parametrices for two spacetimes
$\mathcal{M}$, $\mathcal{N}$,
related by a $\mathsf{Loc}$ morphism $\chi: \mathcal{M} \to \mathcal{N}$,
then $\chi^* H^\mathrm{sing}_\mathcal{N} = H^\mathrm{sing}_\mathcal{M}$.
Thus, we can use the fact that
$
    \left( \chi_* \mathcal{F} \right)^{(n)} [\phi]
        =
    \left( \chi_* \right)^{\otimes n} \mathcal{F}^{(n)} [\chi^* \phi]
$, to show
\begin{equation}
    \alpha_{H - H^\mathrm{sing}_\mathcal{N}} \left( \chi_* \mathcal{F} \right)
        =
    \chi_* \left( 
                \alpha_{\chi^*(H - H^\mathrm{sing}_\mathcal{N})} \mathcal{F}
            \right).
\end{equation}
On the left hand side, we have simply
$\left(\nord{\chi_* \mathcal{F}}_\mathcal{N}\right)_H$,
whereas on the right hand side, once we note that
$
    \alpha_{\chi^*(H - H^\mathrm{sing}_\mathcal{M})} \mathcal{F}
        =
    \alpha_{\chi^* H - H^\mathrm{sing}_\mathcal{M} } \mathcal{F}
        =
    \left( \nord{\mathcal{F}}_\mathcal{M} \right)_{\chi^*H}
$,
we see that this is
$\left( \mathfrak{A}\chi \nord{\mathcal{F}}_\mathcal{M} \right)_H$ as required.

\section{The Massless Scalar Field on a Cylinder}
\label{sec:heisenberg_virasoro}

Now that we have constructed both a classical and quantum
algebra of observables, and introduced several ordering maps between them,
we may study their finer details in an explicit example.
As our ultimate goal is to understand
conformal field theory from the perspective of pAQFT,
the massless scalar field is the obvious place to begin.
Moreover, owing to its flat geometry and compact Cauchy surfaces,
the Einstein cylinder $\mathcal{M}_\mathrm{cyl}$
--
    defined as the image of $2$D Minkowski space, $\mathbb{M}_2$,
    under the identification $(t, x) \sim (t, x + 2\pi)$
--
provides a natural and convenient setting in which to explore the chiral aspects of
the massless scalar field within the pAQFT framework.

In this section, we shall see how
the quantum algebra of observables for the massless scalar field
contains a pair of Heisenberg algebras and a pair of Virasoro algebras,
one each for the left and right null-derivatives of the field.
In the construction of the Virasoro algebra,
we shall also see that the principle of local covariance
outlined in \Cref{sec:normal_ordering} is necessary to recover
the `radially-ordered' form of the Virasoro algebra.
The argument involved in this re-ordering constitutes
a mathematically rigourous form of the known trick of identifying
$1 + 2 + 3 + \cdots = \zeta(-1)$.

\subsection{Minkowski Space}

We begin by finding the causal propagator for the
massless scalar field in Minkowski space.
From this we shall later obtain the propagator for the cylinder,
and hence the Poisson algebra $\mathfrak{P}(\mathcal{M}_\mathrm{cyl})$.
Moreover we shall begin to see how the classical Poisson algebra of
the massless scalar field naturally contains two chiral subalgebras.

The equation of motion for the massless scalar field on Minkowski space is simply
\begin{equation}
    \label{eq:EoM}
    -\left( \partial_t^2 - \partial_x^2 \right) \phi = 0.
\end{equation}
This is easiest to solve if we adopt null coordinates
$u = t - x$, $v = t + x$.
The fundamental solutions $E^{R/A}$ to \eqref{eq:EoM} must then satisfy
\begin{equation}
    \label{eq:adv_ret_propagators}
    4 \frac{\partial}{\partial u} \frac{\partial}{\partial v} E^{R/A}(u, v; u', v')
        =
    - 2 \delta(u - u') \delta(v - v').
\end{equation}
By inspection one can then deduce that the distributions
\begin{equation}
    E^{R/A}(u, v; u', v')
        =
    -\frac{1}{2}\theta(\pm (u - u') ) \theta(\pm (v - v') )
\end{equation}
both satisfy \eqref{eq:adv_ret_propagators} and have the desired supports.
Taking their difference we find the Pauli-Jordan function to be
\begin{equation}
    E(u, v; u', v')
        =
    -\frac{1}{2} \left[
        \theta(u - u')\theta(v - v') - \theta(u' - u)\theta(v' - v)
    \right].
\end{equation}
We can rewrite this propagator in the form
\begin{equation}
    E(u, v; u', v')
        =
    -\frac{1}{4} \left[
        \mathrm{sgn}(u - u') +
        \mathrm{sgn}(v - v')
    \right],
\end{equation}
where $\mathrm{sgn}(x) = \theta(x) - \theta(-x)$.
In other words, we can decouple the $u$-dependent terms from the $v$-dependent,
defining the summands
\begin{equation}
    E = E^\ell + E^\arr,
\end{equation}
such that $E^\ell$ does not depend on $v$ and \textit{vice-versa}.

This split is significant for functionals which depend on the field configuration $\phi$
only through its left/right null derivative.
If we indicate the action of the differential operator
$\partial_u$ on a functional $\mathcal{F}$ by
$(\partial_u^* \mathcal{F})[\phi] := \mathcal{F}[\partial_u \phi]$,
then the functional derivative of $\partial_u^* \mathcal{F}$ is given by
\begin{equation}
    (\partial_u^* \mathcal{F})^{(1)}[\phi]
        =
    - \partial_u \mathcal{F}^{(1)}[\partial_u \phi].
\end{equation}
Consequently, the Peierls bracket of two such functionals is
\begin{equation}
    \label{eq:derivative_peierls}
    \{ \partial_u^* \mathcal{F}, \partial_u^* \mathcal{G} \}[\phi]
        =
    \left\langle
        (\partial_u \otimes \partial_u)E,
        \mathcal{F}^{(1)}[\partial_u \phi] \otimes
        \mathcal{G}^{(1)}[\partial_u \phi]
    \right\rangle.
\end{equation}
This equality motivates the construction of a new Poisson algebra,
outlined in the following proposition:

\begin{prop}
    \label{prop:minkowski_chiral_bracket}
    The space $\mathfrak{F}_{\mu c}(\mathbb{M}_2)$, equipped with
    the pointwise product $\cdot$, and the bracket
    \begin{equation}
        \left\{ 
            \mathcal{F},
            \mathcal{G}
        \right\}_\ell
        [\phi]
            :=
        \left\langle 
            \left( \partial_u \otimes \partial_u \right) E,
            \mathcal{F}^{(1)}[\phi] \otimes
            \mathcal{G}^{(1)}[\phi]
        \right\rangle
    \end{equation}
    is a Poisson algebra, which we denote $\mathfrak{P}_\ell(\mathbb{M}_2)$.
    Furthermore, the map
    $
        \partial_u^*:
            \mathfrak{F}_{\mu c}(\mathbb{M}_2) \to
            \mathfrak{F}_{\mu c}(\mathbb{M}_2)
    $
    yields a Poisson algebra homomorphism
    $\mathfrak{P}_\ell(\mathbb{M}_2) \to \mathfrak{P}(\mathbb{M}_2)$.
\end{prop}

\begin{proof}
    Because $\mathrm{WF}((\partial_u \otimes \partial_u) E) \subseteq \mathrm{WF}(E)$,
    we see that all the estimates of
    $
        \mathrm{WF}\big(
            \left\{ \mathcal{F}, \mathcal{G}\right\}^{(n)}
        \big)
    $
    given in the proof of \Cref{prop:peierls_closure} also hold for
    $
        \mathrm{WF}\big(
            \left\{ \mathcal{F}, \mathcal{G}\right\}_\ell^{(n)}
        \big).
    $
    Thus, the microcausality of $\left\{ \mathcal{F}, \mathcal{G}\right\}$
    implies that of $\left\{ \mathcal{F}, \mathcal{G}\right\}_\ell$.

    Next, we must show that $\left\{ \cdot, \cdot\right\}_\ell$ satisfies
    the Jacobi identity.
    This we can achieve using \eqref{eq:derivative_peierls} alongside
    the observation that $\partial_u^*$ is injective
    (which follows from the fact that $\partial_u$ is surjective).
    Let $\mathcal{F}, \mathcal{G},$ and $\mathcal{H}$ all be microcausal functionals.
    Consider
    \begin{equation*}
        \partial_u^*\left( 
            \big\{
                \mathcal{F},
                \big\{
                    \mathcal{G},
                    \mathcal{H}
                \big\}_\ell
            \big\}_\ell +
            \cdots
        \right)
            = 
        \big\{
            \partial_u^* \mathcal{F},
            \big\{
                \partial_u^* \mathcal{G},
                \partial_u^* \mathcal{H}
            \big\}
        \big\} +
        \cdots,
    \end{equation*}
    where $\cdots$ includes both remaining even permutations of
    $\mathcal{F}, \mathcal{G}$, and $\mathcal{H}$.
    The right-hand side of this vanishes as
    the Peierls bracket satisfies the Jacobi identity
    hence, by injectivity, we see that
    $\{\mathcal{F}, \{\mathcal{G}, \mathcal{H}\}_\ell\}_\ell + \cdots$ also vanishes.

    Finally, we note that
    $
        \mathrm{WF}((\partial_u^* \mathcal{F})^{(n)}[\phi])
            =
        \mathrm{WF}(
            (-1)^n{\partial_u}^{\otimes n}
            \mathcal{F}^{(n)}[\partial_u \phi]
        )
            \subseteq
        \mathrm{WF}(\mathcal{F}^{(n)}[\partial_u \phi]),
    $
    confirming that $\partial_u^*$ indeed defines a linear endomorphism on
    $\mathfrak{F}_{\mu c}(\mathbb{M}_2)$ and hence,
    by \eqref{eq:derivative_peierls}, a Poisson algebra homomorphism.
\end{proof}

Note that, $(\partial_u \otimes \partial_u)E^\arr = 0$,
hence the integral kernel of the differentiated propagator is
\begin{equation}
    \partial_u \partial_{u'}   E   (u, v; u', v')
        =
    \partial_u \partial_{u'} E^\ell(u, v; u', v')
        =
    \tfrac{1}{2} \delta'(u - u').
\end{equation}
This form of the commutator can be seen as an example of the
\textit{mutual locality} of chiral fields,
\cite[Definition 2.3]{kacVertexAlgebrasBeginners1998},
a concept central to many theorems in the VOA framework
We shall henceforth refer to $\{\cdot, \cdot\}_\ell$ as
the \textit{chiral} bracket,
and the analogously defined $ \{\cdot, \cdot\}_\arr$ as
the \textit{anti-chiral} bracket.

It turns out that the chiral and anti-chiral brackets can be defined on
a space of functionals larger than
$\mathfrak{F}_{\mu c}(\mathbb{M}_2)$.
In an upcoming paper, we shall explore what a suitable enlargement is,
and how this relates to the concept
that chiral fields are defined over a single light-ray.

\subsection{The Heisenberg Algebra}

We shall now find the advanced and retarded propagators for
the Einstein cylinder $\mathcal{M}_{\mathrm{cyl}}$.
If $(u, v)$ denotes the null coordinates of a point in $\mathbb{M}_2$,
then we define an equivalence relation on $\mathbb{M}_2$ by
$(u, v) \sim (u + 2\pi, v - 2\pi)$.
The Einstein cylinder is then defined as the quotient space
$\mathcal{M}_{\mathrm{cyl}} = \mathbb{M}_2 / \sim$,
with the unique metric such that the covering map
$\pi: \mathbb{M}_2 \to \mathcal{M}_\mathrm{cyl}$
is a local isometry.
We will write points in $\mathcal{M}_{\mathrm{cyl}}$
as equivalence classes
$[u, v] \subset \mathbb{M}_2$,
where $[u, v] = [u + 2 \pi, v - 2 \pi]$.

The causal propagator for the cylinder may be obtained from
the advanced and retarded propagators of Minkowski spacetime using the method of images.
Firstly, note there is an isomorphism between
$
    \mathfrak{E}(\mathbb{M}_2)^{\mathbb{Z}}
        =
    \{
        f \in \mathfrak{E}(\mathbb{M}_2)
            \, | \,
        f \circ T_n \equiv f, \forall n \in \mathbb{Z}
    \}
$
and
$\mathfrak{E}(\mathcal{M}_{\mathrm{cyl}})$.
Going from $\mathcal{M}_{\mathrm{cyl}}$ to $\mathbb{M}_2$,
this map is simply the corestriction of $\pi^*$
to the space of $\mathbb{Z}$ invariants.
If we denote the inverse of this isomorphism by $\pi_*$,
then we claim the retarded and advanced propagators on the cylinder are given by
\begin{equation}
    \label{eq:loose_cyl_prop_def}
    E^{R/A}_\mathrm{cyl}
        =
    \pi_* E^{R/A} \pi^*.
\end{equation}
For this map to be well defined, amongst other details,
we must show that the domain of $E^{R/A}$ can be extended to the image
$\pi^* \big( \mathfrak{D}(\mathcal{M}_{\mathrm{cyl}}) \big)$,
and that the output of $E^{R/A} \pi^*$ contains only $\mathbb{Z}$ invariants.
Proof of which can be found in \Cref{sec:images}.

That these maps are then the desired propagators follows from
the relationship between the equations of motion on the cylinder and Minkowski.
Let $U \subseteq \mathbb{M}_2$ be a sub-spacetime of $\mathbb{M}_2$ and let
$\iota_U: U \hookrightarrow \mathbb{M}_2$ be its inclusion into $\mathbb{M}_2$.
If $U$ is small enough that
$\pi \circ \iota_U: U \to \mathcal{M}_{\mathrm{cyl}}$ is an embedding,
then we can show from \eqref{eq:Euler_Lagrange_covariance} that
\begin{equation}
    (\pi \circ \iota_U)^* P_{\mathcal{M}_{\mathrm{cyl}}}
        =
    P_U (\pi \circ \iota_U)^*.
\end{equation}
Furthermore, $\iota_U$ is itself an isometric embedding, hence
\begin{equation}
    \iota_U^* P_{\mathbb{M}_2} = P_U \iota_U^*.
\end{equation}
Combining these equations, we find
\begin{equation}
    \label{eq:local_intertwine}
    \iota_U^* \pi^* P_{\mathcal{M}_{\mathrm{cyl}}}
        =
    \iota_U^* P_{\mathbb{M}_2} \pi^*.
\end{equation}
One can then show that $\mathbb{M}_2$ is covered by open sets $U$
for which \eqref{eq:local_intertwine} holds,
and thence that $\pi^* P_{\mathcal{M}_{\mathrm{cyl}}} = P_{\mathbb{M}_2} \pi^*$.
By acting on the left-hand side of \eqref{eq:loose_cyl_prop_def} with
$\pi^* P_{\mathcal{M}_{\mathrm{cyl}}}$ and the right-hand side with
$P_{\mathbb{M}_2} \pi^*$, we are then able to see why these maps are fundamental solutions
to $P_{\mathcal{M}_{\mathrm{cyl}}}$.

Throughout this section we shall use the following coordinates for
$\mathcal{M}_\mathrm{cyl}$.
Let $U = (0,2 \pi) \times \mathbb{R} \subset \mathbb{R}^2$, then
\begin{align}
    \rho:
    U
    &\longrightarrow
    \mathcal{M}_\mathrm{cyl},
    \nonumber \\
    \label{eq:cyl_chart}
        (u, v)
            &\longmapsto
        [u, v].
\end{align}
And, by a standard abuse of notation,
for $\phi \in \mathfrak{E}(\mathcal{M}_{\mathrm{cyl}})$,
we shall write $(\phi \circ \rho) (u, v)$ as simply $\phi(u, v)$.
As the $(u, v)$ coordinates parametrise
$\mathcal{M}_{\mathrm{cyl}}$ up to a set of measure zero,
they are sufficient to define integration on $\mathcal{M}_{\mathrm{cyl}}$.
In turn, this allows us to define an integral kernel for $E_\mathrm{cyl}$ by
\begin{equation}
    (E_\mathrm{cyl} \phi)(u, v)
        =:
    \int_U
        E_\mathrm{cyl}(u, v; u', v') \phi(u', v')
    \, \mathrm{d}u' \, \mathrm{d}v',
\end{equation}
which we may then write in terms of the integral kernel of $E$ as
\begin{align}
    \label{eq:causal_prop_cyl}
    E_\mathrm{cyl}(u, v; u', v')
   &    =
    \sum_{k \in \mathbb{Z}}
    E(u, v; u' + 2 \pi k, v' - 2 \pi k), \nonumber \\
   &    =
    - \frac{1}{2}
    \left(
        \left\lfloor \frac{u - u'}{2 \pi} \right\rfloor +
        \left\lfloor \frac{v - v'}{2 \pi} \right\rfloor +
        1
    \right).
\end{align}
Once again, we see the characteristic splitting
of the $u$-dependent and $v$-dependent terms of $E_\mathrm{cyl}$,
which we write
$E_\mathrm{cyl} = E_\mathrm{cyl}^\ell + E_\mathrm{cyl}^\arr$,
just as before.

Just as with \Cref{prop:minkowski_chiral_bracket},
we can define a chiral bracket $\{\cdot, \cdot\}_\ell$ on
$\mathfrak{F}_{\mu c}(\mathcal{M}_{\mathrm{cyl}})$
using $\left( \partial_u \otimes \partial_u \right)E_\mathrm{cyl}$
instead of $E_\mathrm{cyl}$,
yielding the chiral Poisson algebra
$\mathfrak{P}_\ell(\mathcal{M}_{\mathrm{cyl}})$.
The proof that $\mathfrak{P}_\ell(\mathcal{M}_{\mathrm{cyl}})$
is a Poisson algebra and that
$
    \partial_u^*:
        \mathfrak{F}_{\mu c}(\mathcal{M}_{\mathrm{cyl}}) \to
        \mathfrak{F}_{\mu c}(\mathcal{M}_{\mathrm{cyl}})
$
is a Poisson algebra homomorphism
carries over essentially unchanged from $\mathbb{M}_2$.
For our choice of chart, we always have that
$-2 \pi < u - u' < 2\pi$,
thus the integral kernel for the chiral bracket can be written
\begin{equation}
    \left( \partial_u \otimes \partial_u \right)
    E_\mathrm{cyl}(u, v; u', v')
        =
    \left( \partial_u \otimes \partial_u \right)
    E^\ell(u, v; u', v')
        =
    \frac{1}{2} \delta'(u - u').
\end{equation}

We shall perform our next set of calculations using
$\{\cdot, \cdot\}_\ell$.
In an effort to avoid confusion,
when we are working in $\mathfrak{P}_\ell(\mathcal{M}_{\mathrm{cyl}})$,
we shall denote the field configuration input to the functional by $\psi$.
We think of $\psi$ as $\partial_u \phi$
which is realised when we apply the algebra homomorphism
$
    (\partial_u^* \mathcal{F})[\phi]
        =
    \mathcal{F}[\partial_u \phi]
        =
    \mathcal{F}[\psi].
$

We first define the family of functionals
$
    \{A_n\}_{n \in \mathbb{Z}}
        \subset
    \mathfrak{F}(\mathcal{M}_{\mathrm{cyl}})
$
by
\begin{equation}
    \label{eq:a_n_defn}
    A_n[\psi]
        :=
    \frac{1}{\sqrt{\pi}}
    \int_{u = 0}^{2 \pi}
        e^{inu} \psi(u, -u)
    \, \mathrm{d}u.
\end{equation}
Their derivatives are given by
\begin{equation}
    \left\langle 
        A_n^{(1)}[\psi],
        h
    \right\rangle
        =
    \frac{1}{\sqrt{\pi}}
    \int_{u = 0}^{2 \pi} 
        e^{inu} h(u, -u)
    \, \mathrm{d}u,
\end{equation}
for $h \in \mathfrak{D}(\mathcal{M}_{\mathrm{cyl}})$.

For $\psi \in \mathfrak{E}(\mathcal{M}_\mathrm{cyl})$, $A_n[\psi]$ is simply the
$n^\text{th}$ Fourier mode of $\psi$ restricted to the $t=0$ Cauchy surface
$\Sigma_0$ if we wind around the surface \textit{clockwise}.
These functionals are neither microcausal nor local because,
by \cite[Theorem 8.2.5]{hormanderAnalysisLinearPartial2015},
one can show the wavefront set of $A_n^{(1)}[\psi]$
is the conormal bundle to $\Sigma_0$.
However, we shall see that they still posess a well-defined chiral bracket,
and generate a closed algebra with respect to it.

A direct computation of the chiral bracket yields
\begin{align}
    \big\{ A_n, A_m\big\}_\ell [\psi]
        &=
    \frac{1}{\pi} \int_{u=0}^{2\pi} \int_{u'=0}^{2\pi}
        e^{i(nu + mu')}
        (\partial_u \partial_{u'} E_\mathrm{cyl}^\ell)
            (u, - u; u', - u')
    \, \mathrm{d}u \mathrm{d}u'
        \nonumber \\
        &=
    \frac{1}{2 \pi}
    \int_{u = 0}^{2 \pi} \int_{u' = 0}^{2 \pi}
        e^{i(nu + mu')}
        \delta'(u - u')
    \, \mathrm{d}u' \mathrm{d}u'
        \nonumber \\
        &=
    -i n \delta_{n+m, 0},
\end{align}
hence
\begin{equation}
    \big\{ A_n, A_m\big\}_\ell
        =
    -i n \delta_{n+m, 0},
\end{equation}
where we suppress the constant functional for convenience.

This demonstrates that the Lie algebra generated by the $A_n$
with the Lie bracket $\{ \cdot, \cdot\}_\ell$
is isomorphic to the Heisenberg algebra.
Moreover, as $\partial_u^*$ is a Poisson algebra homomorphism,
we see that the algebra generated by
$\mathcal{A}_n := \partial_u^*A_n$
with the Peierls bracket is also isomorphic to the Heisenberg algebra.

Quantising this family of functionals is relatively simple.
Let $H \in \mathrm{Had}(\mathcal{M}_{\mathrm{cyl}})$
be some Hadamard distribution.
As the functionals $\mathcal{A}_n$ are linear,
the definition of the $\star_H$ product implies
the familiar Dirac quantisation rule is valid:
\begin{equation}
    \big[
        \mathcal{A}_n,
        \mathcal{A}_m
    \big]_{\star_H}
        =
    i \hbar
        \big\{
            \mathcal{A}_n,
            \mathcal{A}_m
        \big\}
        =
      \hbar n \delta_{n+m, 0}.
\end{equation}
Furthermore, $\alpha_{H' - H}$ acts by identity on linear functionals,
hence this result is independent of our choice of a Hadamard state $H$.

Of course, there is nothing particularly special about the choice of
$\Sigma_0$ as the Cauchy surface.
From the covariance of the Peierls bracket we already know that,
for any isometry $\chi \in \mathrm{Aut}(\mathcal{M}_{\mathrm{cyl}})$,
the family of functionals $\{ \chi_* \mathcal{A}_n \}_{n \in \mathbb{N}}$
has the same commutation relations as $\{ \mathcal{A}_n \}_{n \in \mathbb{N}}$.
Moreover, we can see in these functionals the beginnings of conformal covariance,
which will be explored further in \Cref{sec:conformal_covariance}.
In null coordinates, we can define a conformal transformation of the cylinder as
$\chi[u, v] = [\mu(u), \nu(v)]$
where the pair of functions $\mu, \nu \in \mathrm{Diff}_+(\mathbb{R})$
satisfy $\mu(u + 2 \pi) = \mu(u) + 2 \pi$ and
$\nu(v + 2 \pi) = \nu(v) + 2 \pi$.
One can then show that the family $\{ \chi_* \mathcal{A}_n \}_{n \in \mathbb{N}}$
still has the same commutation relations as before in the case
that $\chi$ is conformal.

We can define a family of functionals akin to $\mathcal{A}_n$:
\begin{equation}
    \mathcal{A}_n^\gamma[\psi] :=
    \int_{S^1}
        e^{in \tau}
        \gamma^* \left( \frac{\partial \phi}{\partial u} \mathrm{d}u \right),
\end{equation}
where $\gamma: S^1 \to \mathcal{M}_{\mathrm{cyl}}$ is any
spacelike loop around $\mathcal{M}_{\mathrm{cyl}}$.
The original $\mathcal{A}_n$ correspond to the choice of loop
$\gamma_0(\tau) = [\tau, -\tau]$, and one can show that,
if $\gamma = \chi \circ \gamma_0$ for some conformal transformation $\chi$,
then $\chi_* \mathcal{A}_n = \mathcal{A}_n^\gamma$.

In fact, for any other Cauchy surface $\Sigma$ of $\mathcal{M}_{\mathrm{cyl}}$,
it is possible to find a conformal transformation $\chi$
such that $\gamma =\chi \circ \gamma_0$ is a parametrisation of $\Sigma$,
hence $\mathcal{A}_n^\gamma$ is a copy of the Heisenberg algebra associated with
the surface $\Sigma$.
As a sketch: $\chi$ is obtained by taking a right-moving null ray
passing through a point $[u, -u] \in \Sigma_0$, and finding the
unique point $[u, v] \in \Sigma$ lying on the same ray.
This defines the map $\nu$ such that $\nu(-u) = v$,
which one can show is an element of $\mathrm{Diff}_+(S^1)$,
then any choice of $\mu \in \mathrm{Diff}_+(S^1)$ completes
the definition of $\chi$, for example just the identity function.

These $\mathcal{A}_n^\gamma$ will not be needed in this paper.
However, functionals of this form prove vital for defining truly
chiral (i.e. 1-dimensional) algebras as emerging from locally covariant field theory.
We shall explore this further in a future paper.

\subsection{The Virasoro Algebra}
\label{sec:Virasoro}

As the Virasoro algebra arises from quadratic functionals,
the ordering ambiguities we could previously disregard become relevant,
and we cannot so easily carry computations from
Minkowski space over to the cylinder.
To start, the classical functionals are defined analogously to the $A_n$ functionals.
Again, we begin by defining a family
$
    \{B_n\}_{n \in \mathbb{Z}}
        \subset
    \mathfrak{F}(\mathcal{M}_{\mathrm{cyl}}),
$
by
\begin{equation*}
    B_n [\psi]
        :=
    \int_{u = 0}^{2 \pi}
        e^{inu} \psi^2(u, -u)
    \, \mathrm{d}u.
\end{equation*}
As before, we shall compute the chiral bracket of $B_n$ with $B_m$
in order to obtain the Peierls bracket for the functionals
$\mathcal{B}_n := \partial_u^*B_n$.

For future reference,
the functional derivatives of $B_n$ are
\begin{subequations}
    \begin{align}
        \label{eq:B_n_first_deriv}
        \left\langle 
            B_n^{(1)}[\psi],
            g
        \right\rangle
            &=
        2 \int_{u = 0}^{2 \pi}
            e^{in u} \psi(u, -u) g(u, -u)
        \, \mathrm{d}u,
            \\
        \left\langle 
            B_n^{(2)}[\psi],
            g \otimes h
        \right\rangle
            &=
        2 \int_{u = 0}^{2 \pi}
            e^{inu} g(u, -u) h(u, -u)
        \, \mathrm{d}u.
    \end{align}
\end{subequations}
Here again, the wavefront set of $B_n^{(1)}[\psi]$ is contained within
the conormal bundle of $\Sigma_0$ and hence $B_n$
is not microcausal.
Moreover, we see that, like $A_n$, these functionals are additive,
which means that the support of $B_n^{(2)}$,
and hence that of $\mathcal{B}_n^{(2)}$,
is contained within the thin diagonal
$\Delta_2 \subset \mathcal{M}_{\mathrm{cyl}}^2$.
This will be vital when we later apply
the locally covariant Wick ordering prescription 
outlined in \Cref{sec:normal_ordering} to these functionals.

The chiral bracket of $B_n$ with $B_m$ is given by
\begin{align}
    \{B_n, B_m\}_\ell [\psi]
   &    =
    2 \int_{u = 0}^{2 \pi} \int_{u' = 0}^{2 \pi}
        \delta'(u - u')
        e^{inu } \psi(u , -u ) \cdot
        e^{imu'} \psi(u', -u')
    \, \mathrm{d}u \, \mathrm{d}u'
    \nonumber \\
    &    =
    -2 \int_{u = 0}^{2 \pi}
        \big[
            in \psi(u, -u) +
            (\partial_u \psi)(u, -u)
        \big]
        \psi(u, -u)
        e^{i(n+m)u}
    \, \mathrm{d}u,
    \nonumber \\
   &    =
    -i(n - m)
    \int_{u = 0}^{2\pi}
        e^{i(n + m)u}
        \psi^2(u, -u)
    \, \mathrm{d}u
    \nonumber \\
    \label{eq:B_n_witt_relations}
   &    =
    -i(n-m)B_{n+m}[\psi],
\end{align}
where the move from the second to the third line can be made by exploiting
the skew-symmetry of the equation under the interchange of $n$ with $m$.
Hence, we can already see that the $B_n$ under the chiral bracket generate
a copy of the Witt algebra.

Next, we shall quantise the $\mathcal{B}_n$ observables.
Using \eqref{eq:B_n_witt_relations},
we can immediately note that the $\mathcal{O}(\hbar)$ term of
$\left[ \nord{\mathcal{B}_n}, \nord{\mathcal{B}_m} \right]$
must be ${\hbar(n - m)\nord{\mathcal{B}_{n + m}}}$,
regardless of the quantisation map used.
In order to determine the $\mathcal{O}(\hbar^2)$ term though,
we must decide on a particular choice of prescription.

As explained in \Cref{sec:deformation_quantisation},
it is inconvenient to work directly with $\mathfrak{A}(\mathcal{M}_\mathrm{cyl})$.
Instead, we perform our computations in $\mathfrak{A}^H(\mathcal{M}_\mathrm{cyl})$
for some suitable choice of Hadamard distribution $H$.
The simplest choice is to take $H = W_\mathrm{cyl} - \tfrac{i}{2}E_\mathrm{cyl}$,
where $W_\mathrm{cyl}$ is the ultrastatic vacuum for the cylinder,
uniquely distinguished by the fact that it is
invariant under time-translations.
The integral kernel of $W_\mathrm{cyl}$ may be written
\begin{equation}
    \label{eq:cylinder_2_point}
    W_\mathrm{cyl}(u, v; u', v')
        =
    \frac{1}{4 \pi}
    \sum_{k \in \mathbb{Z}^*}
        \frac{1}{k}
        \left( 
            e^{-ik(u - u')} + e^{-ik(v - v')}
        \right).
\end{equation}
Unlike for the massive scalar field,
time-translation is not enough to fix the kernel of $W_\mathrm{cyl}$ uniquely,
owing to the presence of zero mode solutions to the massless Klein-Gordon equation.
However, this is no issue in the algebraic approach to QFT,
as the construction of our algebra of observables is
independent of any choice of ground state and, hence,
of any way in which we may choose to handle the problem of zero modes.

Moreover, we are concerned with the $\star$ products
of functionals which depend on the field configuration $\phi$
only through one of its null derivatives.
In effect, this means we only depend on $W_\mathrm{cyl}$ to define
the 2-point function for the derivative field
\begin{equation}
    (\partial_u \otimes \partial_u)W_\mathrm{cyl}(\mathbf{x}; \mathbf{y})
        =
    \big\langle
        (\partial_u \phi)(\mathbf{x})
        (\partial_u \phi)(\mathbf{y})
    \big\rangle_{\omega}.
\end{equation}
Taking this derivative annihilates any zero-modes,
thus there is no ambiguity in defining the integral kernel
of $(\partial_u \otimes \partial_u)W_\mathrm{cyl}$.

If we consider the $\star_H$ product of two functionals of the form
$\partial_u^*\mathcal{F}$, we find
\begin{equation}
    \label{eq:star_prod_intertwine}
    \left(
        \left( \partial_u^* \mathcal{F} \right)
            \star_H
        \left( \partial_u^* \mathcal{G} \right)
    \right)[\phi]
        =
    \sum_{n = 0}^\infty \frac{\hbar^n}{n!}
    \left\langle 
        \left[ 
            \left( \partial_u \otimes \partial_u \right) W_\mathrm{cyl}
        \right]^{\otimes n},
        \mathcal{F}^{(n)}[\partial_u \phi] \otimes
        \mathcal{G}^{(n)}[\partial_u \phi]
    \right\rangle.
\end{equation}
Analogously to \Cref{prop:minkowski_chiral_bracket},
we can hence define a chiral subalgebra of $\star_H$ via the following:
\begin{prop}\label{prop:du-star-alg}
    The space $\mathfrak{F}_{\mu c}(\mathcal{M}_{\mathrm{cyl}})[[\hbar]]$,
    equipped with the associative product $\star_{H, \ell}$ defined by
    \begin{equation}
        (
            \mathcal{F} \star_{H, \ell}
            \mathcal{G}
        )[\phi]
            :=
        \sum_{n \in \mathbb{N}}
            \frac{\hbar^n}{n!}
            \left\langle 
                [(\partial_u \otimes \partial_u) W_\mathrm{cyl}]^{\otimes n},
                \mathcal{F}^{(n)}[\phi] \otimes \mathcal{G}^{(n)}[\phi]
            \right\rangle,
    \end{equation}
    is a $*$-algebra, which we denote by
    $\mathfrak{A}^H_\ell(\mathcal{M}_{\mathrm{cyl}})$.
    Moreover, the linear extension of $\partial_u^*$
    --
        defined in \Cref{prop:minkowski_chiral_bracket}
    --
    to $\mathfrak{F}_{\mu c}(\mathcal{M}_{\mathrm{cyl}})[[\hbar]]$
    yields a $*$-algebra homomorphism
    $
        \mathfrak{A}^H_\ell(\mathcal{M}_{\mathrm{cyl}}) \to
        \mathfrak{A}^H     (\mathcal{M}_{\mathrm{cyl}})
    $.
\end{prop}

\begin{proof}
    Just as in the classical case,
    because
    $
        \mathrm{WF}(\left( \partial_u \otimes \partial_u \right) W)
            \subseteq
        \mathrm{WF}(W)
    $,
    the closure of
    $\mathfrak{F}_{\mu c}(\mathcal{M}_{\mathrm{cyl}})[[\hbar]]$ under
    $\star_{H, \ell}$ is proved in exactly the same way as for $\star_H$,
    as spelled out in \Cref{prop:quantum_closure}.
    That $\partial_u^*$ intertwines $\star_{H, \ell}$ with $\star_H$ is verified by
    \eqref{eq:star_prod_intertwine}.
    And associativity follows from injectivity of $\partial_u^*$.
\end{proof}

we may now compute the product
$B_n \star_{H_\mathrm{cyl}, \ell} B_m$.
In the abstract algebra, this amounts to computing
$
    \altnord{\mathcal{B}_n}_{H_\mathrm{cyl}}
        \star
    \altnord{\mathcal{B}_m}_{H_\mathrm{cyl}}.
$
Later, we shall compare this to the product of the
\textit{covariantly} ordered $\mathcal{B}_n$.

As the $B_n$ functionals are quadratic, the power series for
their star product truncates at $\mathcal{O}(\hbar^2)$.
Thus, it may be written in full as
\begin{align}
    \begin{split}
    B_n \star_{H_\mathrm{cyl}, \ell} B_m
        =
    B_n \cdot B_m
        &+
    \hbar
    \left\langle 
        \left[
            \left( \partial_u \otimes \partial_u \right) W_\mathrm{cyl}
        \right]
            ,
        B_n^{(1)}[\psi] \otimes
        B_m^{(1)}[\psi]
    \right\rangle\\
        &+
    \frac{\hbar^2}{2}
    \left\langle 
        \left[
            \left( \partial_u \otimes \partial_u \right) W_\mathrm{cyl}
        \right]^{\otimes 2}
            ,
        B_n^{(2)}[\psi] \otimes
        B_m^{(2)}[\psi]
    \right\rangle.
    \end{split}
\end{align}
First, let us consider the $\mathcal{O}(\hbar)$ term
\begin{align}
    \begin{split}
        \left\langle 
            \left[ \left( \partial_u \otimes \partial_u \right) W_\mathrm{cyl} \right]
                ,
            B_n^{(1)}[\psi] \otimes
            B_m^{(1)}[\psi]
        \right\rangle
        &   = \\
        \sum_{k \in \mathbb{N}}
        \frac{1}{\pi}
        \int_{u  = 0}^{2 \pi}
        \int_{u' = 0}^{2 \pi}
            k e^{-ik(u - u')} \cdot
        &   e^{inu } \psi(u , -u ) \cdot
            e^{imu'} \psi(u', -u')
        \, \mathrm{d}u \mathrm{d}u'.
    \end{split}
\end{align}
We can simplify this slightly by reintroducing the $A_n$ functionals.
Upon doing so, we find
\begin{equation}
    \label{eq:virasoro_o_hbar}
    \left\langle 
        \left[ \left( \partial_u \otimes \partial_u \right) W \right]
            ,
        B_n^{(1)}[\psi] \otimes
        B_m^{(1)}[\psi]
    \right\rangle
        =
    \sum_{k = 1}^\infty k A_{n - k}[\psi] A_{m + k}[\psi].
\end{equation}
(Note that for any function $\psi$ the above series is absolutely convergent as
the smoothness of $\psi$ guarantees $|A_n[\psi]|$ decays rapidly in $n$.)

For the commutator, we need only the anti-symmetric part of \eqref{eq:virasoro_o_hbar},
which is markedly simpler.
For now, however, we proceed to compute the $\mathcal{O}(\hbar^2)$ term.
To do this, we need the following form of the squared propagator:
\begin{equation}
    \label{eq:squared_propagator}
    \left\langle
        \left[
            \left( \partial_u \otimes \partial_u \right)
            W_\mathrm{cyl}
        \right]^2,
        f
    \right\rangle
        =
    \frac{1}{16 \pi^2}
    \sum_{k = 0}^\infty \sum_{l = 0}^{k}
        l(k - l) \int_{\mathcal{M}_\mathrm{cyl}^2} e^{-ik(u - u')} f(u, v, u', v')
    \, \mathrm{dVol}^2.
\end{equation}
This can be obtained na\"ively by just squaring \eqref{eq:cylinder_2_point}
and applying the Cauchy product formula.
For a proof that this indeed converges to the correct distribution,
see \Cref{sec:cauchy_prod_prop}.
We then find
\begin{align}
    \label{eq:virasoro_o_hbar_2}
    \begin{split}
        \frac{1}{2}
        \big\langle 
            \left[
                \left( \partial_u \otimes \partial_u \right) W_\mathrm{cyl}
            \right]^{\otimes 2}
    &           , \,
            B_n^{(2)}[\psi] \otimes
            B_m^{(2)}[\psi]
        \big\rangle \\
    &       =
        \frac{1}{8 \pi^2}
        \sum_{k \in \mathbb{N}} \sum_{l = 0}^k
        l (k - l) 
        \int_{u  = 0}^{2 \pi}
        \int_{u' = 0}^{2 \pi}
            e^{-ik(u - u')}
            e^{in u }
            e^{im u'}
        \, \mathrm{d}u \mathrm{d}u', \\
    &       =
        \frac{1}{2}
        \sum_{k \in \mathbb{N}} \sum_{l = 0}^k
            l (k - l)
        \delta_{n - k, 0} \delta_{m + k, 0}, \\
    &       =
        \frac{n(n^2 - 1)}{12}
        \theta(n)
        \delta_{n+m, 0}.
    \end{split}
\end{align}
Hence, altogether we have
\begin{equation}
    \left(
        B_n
            \star_{H_\mathrm{cyl}, \ell}
        B_m
    \right)
        =
    B_n \cdot B_m
        +
    \hbar
    \sum_{k = 1}^\infty k
       A_{n - k} \cdot A_{m + k}
            +
    \frac{\hbar^2}{12} n^2(n - 1) \theta(n) \delta_{n + m, 0}.
\end{equation}
Next, we compute the commutator
$\left[ B_n, B_m \right]_{\star_{H_\mathrm{cyl}, \ell}}$
Taking the anti-symmetric part of the $\mathcal{O}(\hbar^2)$ term is straighforward:
simply drop the $\theta(n)$.
For \eqref{eq:virasoro_o_hbar}, note that we can write
\begin{equation}
    \sum_{k = 1} k A_{n - k} A_{m + k}
        =
    \frac{1}{2}
    \left(
        \sum_{k \in \mathbb{Z}}  k  A_{n - k} A_{m + k}
            +
        \sum_{k \in \mathbb{Z}} |k| A_{n - k} A_{m + k}
    \right).
\end{equation}
The first series is anti-symmetric under an interchange of $n$ and $m$,
whereas the latter is symmetric and can thus be disregarded.
Next, we take two copies of the anti-symmetric series,
for the first copy we make the change of variables $k \mapsto (n - k)$,
and for the second we choose $k \mapsto (k - m)$.
Recombining these two copies we find
\begin{equation}
    \label{eq:A_n_series_for_B_n_plus_m}
    \sum_{k \in \mathbb{Z}}
        k A_{n - k} A_{m + k}
            =
    \frac{1}{2} (n - m)
    \sum_{k \in \mathbb{Z}}
        A_k A_{n + m - k}.
\end{equation}
By the second convolution theorem,
this final series converges (up to a constant factor) to the
$(n + m)$\textsuperscript{th} Fourier mode of $\psi^2$.
Thus, \eqref{eq:A_n_series_for_B_n_plus_m} is equal to $(n - m)B_{n + m}$,
agreeing with our earlier calculation using the chiral bracket $\{\cdot, \cdot\}_\ell$.
Combining this with the $\mathcal{O}(\hbar^2)$ term \eqref{eq:virasoro_o_hbar_2},
we arrive at the Virasoro relations
\begin{equation}\label{eq:standard_virasoro}
    [B_n, B_m]_{\star_{H_\mathrm{cyl}, \ell}}
        =
    \hbar(n - m)B_{n + m}
        +
    \frac{\hbar^2}{12} n (n^2 - 1) \delta_{n+m, 0}.
\end{equation}
Using the $*$-algebra homomorphism $\partial_u^*$ from Proposition \ref{prop:du-star-alg}, we can then conclude that
\begin{equation}
    [
        \mathcal{B}_n,
        \mathcal{B}_m
    ]_{\star_{H_\mathrm{cyl}}}
        =
    \hbar(n - m) \mathcal{B}_{n + m}
        +
    \frac{\hbar^2}{12} n (n^2 - 1) \delta_{n+m, 0}.
\end{equation}
Finally, applying $\alpha_{H - H_{\mathrm{cyl}}}$ and using the identity \eqref{alpha-H-Hp-relation} we obtain the commutation relation
\begin{equation}
    [
        \altnord{\mathcal{B}_n}_{H_\mathrm{cyl}},
        \altnord{\mathcal{B}_m}_{H_\mathrm{cyl}}
    ]
        =
    \hbar(n - m) \altnord{\mathcal{B}_{n + m}}_{H_\mathrm{cyl}}
        +
    \frac{\hbar^2}{12} n (n^2 - 1) \delta_{n+m, 0}
\end{equation}
in $\mathfrak{A}(\mathcal{M}_{\mathrm{cyl}})$, recalling that $\big( \altnord{\mathcal{B}_n}_{H_\mathrm{cyl}} \big)_H = \alpha_{H - H_{\mathrm{cyl}}} \mathcal{B}_n$.

It is curious that at this stage we have commutators
recognisable as what one might call the `planar' Virasoro relations
(for example \cite[(2.6.6)]{kacVertexAlgebrasBeginners1998})
for a central charge $c = 1$,
despite the fact that all the functionals in question belong on the cylinder.
We will now compute the correction to these relations which occurs when adopting
the locally covariant Wick ordering prescription.
In doing so, we shall see the result is the `radially ordered' Virasoro relations.

Recall from \cref{sec:normal_ordering} that,
heuristically, locally covariant Wick ordering is normal ordering
with respect to the Hadamard parametrix.
In the case of the Minkowski cylinder,
the Hadamard parametrix \eqref{eq:local_Hadamard_form} is particularly simple.
Locally the cylinder is isometric to Minkowski space,
hence the parametrix of the cylinder coincides with that of Minkowski.
For an arbitrary choice of length scale $\lambda$,
the singular part of a Hadamard distribution for the \textit{undifferentiated} field $\phi$ is
\begin{equation}
    W_\mathrm{sing}(u, v; u', v')
        =
    - \frac{1}{4 \pi} \log \left( \frac{(u - u')(v - v')}{\lambda^2} \right).
\end{equation}
Here it is clear that the parametrix exists only locally,
as $W_\mathrm{sing}$ is not spacelike periodic.
Passing over to the differentiated field $\psi$,
the singular term becomes 
\begin{equation}\label{eq:diff_sing}
    \partial_u \partial_{u'} W_\mathrm{sing}(u; u')
        =
    - \frac{1}{4 \pi}
      \frac{1}{(u - u')^2}.
\end{equation}
For the cylindrical vacuum, we have
\begin{equation}
    \partial_u \partial_u' W_\mathrm{cyl}(u; u')
        =
    \frac{1}{4 \pi}
    \sum_{k \in \mathbb{N}}
        k e^{-ik(u - u')}.
\end{equation}
We can think of the above series formally as the derivative of a geometric series.
Replacing $u - u'$ with $z_\epsilon = u - u' - i \epsilon$ makes this series absolutely convergent for $\epsilon > 0$,
thus we can write the $2$-point function as 
\begin{align}
    \partial_u \partial_{u'}
    W_\mathrm{cyl}(u; u') 
        = 
    \frac{1}{4 \pi}
    \lim_{\epsilon \searrow 0} \frac{
        e^{iz_\epsilon}
    }{
        (1 - e^{iz_\epsilon})^2
    }.
\end{align}
Performing an asymptotic expansion of this function near the coincidence limit $u - u' = 0$, we find
\begin{equation}
    \partial_u \partial_{u'} W_\mathrm{cyl}(u; u')
        \approx
    - \frac{1}{4 \pi}
      \frac{1}{(u - u')^2}
        -
      \frac{1}{4 \pi}
      \frac{1}{12}
        +
    \mathcal{O}\left( (u - u')^2 \right).
\end{equation}
Which provides an explicit verification that the vacuum state differs from the parametrix
only by the addition of a smooth, symmetric function.
Moreover, this allows us to calculate
$\nord{\mathcal{B}_n}_{\mathcal{M}_{\mathrm{cyl}}}$.
As we are working in
$\mathfrak{A}^{H_\mathrm{cyl}}(\mathcal{M}_{\rm cyl})$,
we need only compute the functional
$
    \left(
        \nord{\mathcal{B}_n}_{\mathcal{M}_{\mathrm{cyl}}}
    \right)_{H_\mathrm{cyl}}
$,
which is given by
\begin{align}
    \left(
        \nord{\mathcal{B}_n}_{\mathcal{M}_{\mathrm{cyl}}}
    \right)_{H_\mathrm{cyl}}
        &=
    \alpha_{H_\mathrm{cyl} - H_\mathrm{sing}}
    \mathcal{B}_n
    \nonumber \\
        &= 
    \mathcal{B}_n
        +
    \frac{\hbar}{2}
    \left\langle
        H_\mathrm{cyl} - H_\mathrm{sing},
        \mathcal{B}_n^{(2)}
    \right\rangle
    \nonumber \\
        &=
    \mathcal{B}_n
        +
    \frac{\hbar}{2}
    \left\langle
        \left[
            (\partial_u \otimes \partial_u)
            (H_\mathrm{cyl} - H_\mathrm{sing})
        \right],
        B_n^{(2)}
    \right\rangle
    \nonumber \\
        &=
    \mathcal{B}_n
        +
    \hbar \int_{u = 0}^{2 \pi}
        e^{in u}
        \left[ 
            ( \partial_u \partial_{u'} H_\mathrm{cyl} ) -
            ( \partial_u \partial_{u'} H_\mathrm{sing} )
        \right](u, -u; u, -u)
    \, \mathrm{d}u
    \nonumber \\
        &=
    \mathcal{B}_n - \frac{\hbar}{24}\delta_{n, 0}.
\end{align}
For a generic Hadamard state
$H\in \mathrm{Had}(\mathcal{M}_{\mathrm{cyl}})$ we then have
\begin{align}
    \left(
        \nord{\mathcal{B}_n}_{\mathcal{M}_{\mathrm{cyl}}}
    \right)_H
        &=
    \alpha_{H - H_\mathrm{sing}}
    \mathcal{B}_n = \alpha_{H - H_\mathrm{cyl}} \big( \alpha_{H_\mathrm{cyl} - H_\mathrm{sing}}
    \mathcal{B}_n \big)
    \nonumber \\
        &= \alpha_{H - H_\mathrm{cyl}} \mathcal{B}_n - \frac{\hbar}{24}\delta_{n, 0}
    = \big( \altnord{\mathcal{B}_n}_{H_\mathrm{cyl}} \big)_H - \frac{\hbar}{24}\delta_{n, 0}.
\end{align}
In other words, the quantum observables $\nord{\mathcal{B}_n}_{\mathcal{M}_{\mathrm{cyl}}}$ and $\altnord{\mathcal{B}_n}_{H_\mathrm{cyl}}$ in $\mathfrak{A}(\mathcal{M}_{\mathrm{cyl}})$ defined, respectively, as the locally covariant Wick ordering and the normal ordering with respect to the vacuum $H_{\mathrm{cyl}}$ of the classical functionals $\mathcal{B}_n$, are related by a shift
\begin{equation} \label{eq:zero_mode_shift}
    \nord{\mathcal{B}_n}_{\mathcal{M}_{\mathrm{cyl}}}
        =
    \altnord{\mathcal{B}_n}_{H_\mathrm{cyl}} - \frac{\hbar}{24}\delta_{n, 0}.
\end{equation}
With this shift we find, as expected, that the commutation relations of $\nord{\mathcal{B}_n}_{\mathcal{M}_{\mathrm{cyl}}}$ are
\begin{equation}
    \label{eq:shifted_virasoro}
    \left[
            \nord{\mathcal{B}_n}_{\mathcal{M}_\mathrm{cyl}},
            \nord{\mathcal{B}_m}_{\mathcal{M}_\mathrm{cyl}} 
    \right]
        =
    \hbar (n - m) \nord{\mathcal{B}_{n + m}}_{\mathcal{M}_\mathrm{cyl}} 
    + \frac{\hbar^2}{12} n^3 \delta_{n + m, 0}.
\end{equation}

Recall that $\altnord{-}_{H_\mathrm{cyl}}$
can be interpreted as normal ordering with respect to the vacuum $H_\mathrm{cyl}$.
Moreover, we established the Hadamard parametrix $H_\mathrm{sing}$
of the cylinder is effectively the 2-point function of the Minkowski vacuum,
embedded into some suitable neighbourhood of
$\Delta \subset \mathcal{M}_\mathrm{cyl}^2$.
Accordingly, \eqref{eq:standard_virasoro} computes the commutation relations
for Fourier modes of the stress-energy tensor normally ordered with respect to
$H_\mathrm{cyl}$, and \eqref{eq:shifted_virasoro} the same but
ordered with respect to the Minkowski vacuum.

In the standard approach to CFT in two dimensions,
one typically imposes \eqref{eq:standard_virasoro} as the
standard commutation relations for Laurent modes of the stress energy tensor,
here understood as a field over the complex plane in a particular sense.
Then, mapping the plane to the `cylinder' via the map $z \mapsto e^{iz}$,
one may obtain the radially ordered commutation relations,
concordant with \eqref{eq:shifted_virasoro}.
However, in our framework, it does not make much sense
to speak of a Virasoro algebra for the plane,
as there is no suitable notion of mode expansion for the stress-energy tensor.
In fact, arguably the most significant differences between our approach
and the VOA framework is that the latter relies on mode decomposition
in order to analyse the singularity structure of quantum fields,
whereas we instead use tools from microlocal analysis.
\todo[color=todogreen, inline]{
    Changed it. Do we like this as a closing remark?
}
\todo[color=todoBenoit, inline]{
    Yes, looks good now.
}

\subsection{Connection to Zeta Regularisation}

There is a well known trick in the physicists' literature to explain \eqref{eq:zero_mode_shift}.
Firstly, recall that we can write a given $\mathcal{B}_n$ functional as an
infinite series over $\mathcal{A}_m$ functionals
(which is point-wise convergent) as:
\begin{equation}
    \label{eq:B_in_terms_of_A}
    \mathcal{B}_n
        =
    \frac{1}{2}
    \sum_{k \in \mathbb{Z}}
        \mathcal{A}_k \cdot \mathcal{A}_{n - k}.
\end{equation}
The $\star_{H_\mathrm{cyl}}$ product of two such functionals is
\begin{equation}
    \mathcal{A}_k
        \star_{H_\mathrm{cyl}}
    \mathcal{A}_{n - k}
        =
    \mathcal{A}_k \cdot \mathcal{A}_{n - k}
        +
    \hbar k \theta(k) \delta_{n, 0}.
\end{equation}
In particular, for $n \neq 0$ this means that
$
    \mathcal{A}_k
        \cdot
    \mathcal{A}_{n - k}
        =
    \mathcal{A}_{k}
        \star_{H_\mathrm{cyl}}
    \mathcal{A}_{n - k}.
$
Hence, we can define a family of observables
$
    \{(\mathcal{L}_n)\}_{n \in \mathbb{Z}^*}
        \subset
    \mathfrak{A}(\mathcal{M}_{\mathrm{cyl}})
$
by replacing the classical pointwise product $\cdot$ in
\eqref{eq:B_in_terms_of_A} with $\star$.
This family would then coincide with
$
    \{\altnord{\mathcal{B}_n}_{H_\mathrm{cyl}}\}_{n \in \mathbb{Z}^*}.
$
For $n = 0$, we may still replace the pointwise product with
$\star_{H_\mathrm{cyl}}$, but the ordering of the functionals is now significant. Na\"ively replacing the classical pointwise product $\cdot$ in
\eqref{eq:B_in_terms_of_A} for $n=0$ by the $\star$ product yields
the quantum observable
\begin{equation}
    \mathcal{B}_0
        =
    \frac{1}{2}
    \sum_{k = 1}^\infty
        \mathcal{A}_{- k}
            \star_{H_\mathrm{cyl}}
        \mathcal{A}_{k}
            +
    \frac{1}{2}
    \sum_{k = - \infty}^0
        \mathcal{A}_{-k}
            \star_{H_\mathrm{cyl}}
        \mathcal{A}_{k}.     
\end{equation}
Casting rigour aside,
we could then `reorder' $\mathcal{B}_0$ by moving every
$\mathcal{A}_k$ in the second series to the left hand side of the product,
which would produce the infamous divergent series
\begin{equation}
    \mathcal{B}_0
        =
    \frac{1}{2} \mathcal{A}_0
            \star_{H_\mathrm{cyl}}
        \mathcal{A}_0 +
    \sum_{k=1}^\infty
        \mathcal{A}_{-k}
            \star_{H_\mathrm{cyl}}
        \mathcal{A}_k
        +
    \frac{\hbar}{2}
    \sum_{k \in \mathbb{N}} k.
\end{equation}

The rigourous and covariant way of reordering $\mathcal{B}_0$,
as we saw in the previous section, is to apply the map
$\alpha_{H_\mathrm{cyl} - H_\mathrm{sing}}$.
If we define
$
    w_\mathrm{cyl}(u)
        :=
    (\partial_u \otimes \partial_u)
    \big[
        H_\mathrm{cyl}(u; 0)
            -
        H_\mathrm{sing}(u; 0)
    \big]
$,
where we exploit translation invariance to write
$w_\mathrm{cyl}$ as a function of a single variable,
then we can write the normally ordered form of $\mathcal{B}_0$ as
\begin{equation}
    \alpha_{H_\mathrm{cyl} - H_\mathrm{sing}}
        \mathcal{B}_0
            =
    \mathcal{B}_0 + \frac{\hbar}{2} \lim_{u \to 0} w_\mathrm{cyl}(u).
\end{equation}
By approximating both $H_\mathrm{cyl}$ and $H_\mathrm{sing}$
by smooth functions, we can write 
\begin{align}
    w_\mathrm{cyl}(u)
        &=
    \lim_{\epsilon \searrow 0} \left[
        \sum_{n = 0}^\infty n e^{-inu} e^{-n\epsilon}
            -
        \int_{p = 0}^\infty p e^{-ipu} e^{-np} \, \mathrm{d}p
    \right] \\
        &=
    \lim_{z \to -iu} \frac{d}{dz} \left[
        \frac{1}{1 - e^z} + \frac{1}{z}
    \right] \\
        &=
    \lim_{z \to -iu} \frac{d}{dz} \left[
        - \sum_{k = 0}^\infty \frac{B_k}{k!} z^{k-1} + z^{-1}
    \right] \\
        &=
    \lim_{z \to -iu} \frac{d}{dz} \left[
        \sum_{k=0}^\infty \frac{\zeta(-k)}{k!} z^k
    \right],
\end{align}
where here $B_k$ denotes the $k^\text{th}$ Bernoulli number.
This explains the appearance of $\zeta(-1)$ in the normal ordering of
$\mathcal{B}_0$ without any recourse to intermediate divergent series.

To close out this section,
we make a brief remark about how our notion of normal ordering
corresponds to the procedure of shuffling creation operators past annihilators,
or similarly the normally ordered products of chiral fields
\cite[(2.3.5)]{kacVertexAlgebrasBeginners1998}.

Consider the classical product of a collection of $\mathcal{A}_{m_i}$,
the functional derivative of this may be written
$
    (\mathcal{A}_{m_1} \cdots \mathcal{A}_{m_k})^{(1)}
        =
    \sum_{i = 1}^k
        (\mathcal{A}_{m_1} \cdots
        \widehat{\mathcal{A}_{m_i}} \cdots
        \mathcal{A}_{m_k})
        \mathcal{A}_{m_i}^{(1)},
$
where $\widehat{-}$ indicates ommission.
From this we may compute that
\begin{equation}
    \left( \mathcal{A}_{m_1} \cdots \mathcal{A}_{m_k} \right)
        \star_{H_\mathrm{cyl}}
    \mathcal{A}_{n}
        =
    \mathcal{A}_{m_1} \cdots \mathcal{A}_{m_k} \cdot \mathcal{A}_n
        +
    \hbar \sum_{i = 1}^k
        (\mathcal{A}_{m_1} \cdots
        \widehat{\mathcal{A}_{m_i}} \cdots
        \mathcal{A}_{m_k})
        m_i \theta(-m_i) \delta_{m_i + n, 0}.
\end{equation}
Note that the $i^\text{th}$ term in the sum vanishes if $n \leq m_i$.
If we have $n \leq m_i$ for every $i \in \{1, \ldots k\}$,
then we are only left with the $\hbar^0$ term on the right hand side.
Moving to the abstract algebra $\mathfrak{A}(\mathcal{M}_{\mathrm{cyl}})$
by applying the formal map $\alpha_{H_\mathrm{cyl}}^{-1}$,
we then have
\begin{equation}
    \altnord{\mathcal{A}_{m_1} \cdots \mathcal{A}_{m_k}
        \cdot
    \mathcal{A}_n}_{H_\mathrm{cyl}}
    =
    \altnord{\mathcal{A}_{m_1} \cdots \mathcal{A}_{m_k}}_{H_\mathrm{cyl}}
        \star
    \, \mathcal{A}_{n},
\end{equation}
where we make use of the fact that we can canonically identify
linear classical observables with their quantum counterparts.
Applying this procedure iteratively,
if we assume that the sequence $i \mapsto m_i$ is monotonically decreasing,
then we can write
\begin{equation}
    \altnord{\mathcal{A}_{m_1} \cdots \mathcal{A}_{m_k}}_{H_\mathrm{cyl}}
        =
    \mathcal{A}_{m_1} \star \cdots \star \mathcal{A}_{m_k},
        \qquad
    m_i \leq m_{i+1}.
\end{equation}
Given that $[\mathcal{A}_m, \mathcal{A}_n] = 0$
whenever $m$ and $n$ are either both negative or both positive,
we have recovered the familiar result that normal ordering
moves $\mathcal{A}_m$ ``to the right'' if $m \leq 0$
and ``to the left'' if $m > 0$.

\section{Conformal Covariance}
\label{sec:conformal_covariance}

So far, our classical and quantum algebras of observables are insensitive to any
conformal symmetries a given theory may possess.
This is because the morphisms in $\mathsf{Loc}$ are isometric embeddings,
required to preserve the metric exactly.
To study the contidions for and consequences of conformal covariance,
we must relax this condition to allow \textit{conformally} admissible embeddings.
\begin{definition}[Conformally admissible embedding]
    Let $ \mathcal{M} = (M, g, \mathfrak{o}, \mathfrak{t})$
    and $ \mathcal{N} = (N, g', \mathfrak{o}', \mathfrak{t}')$
    be a pair of spacetimes
    (i.e. objects of $\mathsf{Loc}$).
    A smooth embedding
        $\chi: M \hookrightarrow N$
    is \textit{conformally admissible} if
        $\chi^*\mathfrak{o}' = \mathfrak{o}$,
        $\chi^*\mathfrak{t}' = \mathfrak{t}$, and
        $\chi^* g' = \Omega^2 g$,
    where $\Omega \in \mathfrak{E}(\mathcal{M})$ is some nowhere-vanishing function
    known as the \textit{conformal factor}.
\end{definition}
The category $\mathsf{CLoc}$
-- first introduced by Pinamonti in \cite{pinamontiConformalGenerallyCovariant2009} --
is the natural setting for the study of conformal field theories.
It comprises the same objects as $\mathsf{Loc}$,
but enlarges the collection of morphisms to conformally admissible embeddings.
As one might expect, we upgrade the concept of
locally covariant field theory to
locally \textit{conformally} covariant field theory
simply by replacing the category $\mathsf{Loc}$ with $\mathsf{CLoc}$.
In the next section, we show explicitly how this may be done
for a large class of classical theories,
and for the conformally coupled scalar field in the quantum case.

\subsection{Conformally Covariant Field Theory}

\subsubsection{Conformal Lagrangians}

In this section we shall outline the language necessary
to identify a particular Lagrangian (more precisely, its corresponding action)
as being conformally covariant.
In order to do so we must first introduce some notation.

\begin{definition}[Weighted Pushforward/Pullback]
    Let $\chi: \mathcal{M} \hookrightarrow \mathcal{N}$ be
    a conformally admissible embedding with conformal factor $\Omega^2$.
    Given $\Delta \in \mathbb{R}$, the \textit{weighted pushforward}
    with respect to $\Delta$ is defined by
    \begin{align}
        \label{eq:weighted_pushforward}
        \chi_*^{(\Delta)}:
            \mathfrak{D}(\mathcal{M})
                &\to
            \mathfrak{D}(\mathcal{N}),
                \nonumber \\
            f
                &\mapsto
            \chi_*\left( \Omega^{-\Delta} f \right),
    \end{align}
    where $\chi_*$ denotes the standard pushforward of test functions
    \eqref{eq:test_function_pushforward}.
    Similarly, we define the \textit{weighted pullback} with respect to $\Delta$ by
    \begin{align}
        \chi^*_{(\Delta)}:
            \mathfrak{E}(\mathcal{N})
                &\to
            \mathfrak{E}(\mathcal{M}),
                \nonumber \\
            \phi
                &\mapsto
            \Omega^{\Delta} \chi^* \phi.
    \end{align}
\end{definition}

In the following proposition, we collect some useful properties of these maps.

\begin{prop}
    \label{prop:weighted_map_facts}
    Let $\chi \in \mathrm{Hom}_{\mathsf{CLoc}}(\mathcal{M}; \mathcal{N})$,
    and $\rho \in \mathrm{Hom}_{\mathsf{CLoc}}(\mathcal{N}; \mathcal{O})$.
    Then, 
    \begin{enumerate}
        \item   \label{itm:weighted_pushforward_composition}
                $
                    \rho_*^{(\Delta)} \circ
                    \chi_*^{(\Delta)}
                        =
                    \left( \rho \circ \chi \right)_*^{(\Delta)}
                $
        \item   \label{itm:weighted_pullback_composition}
                $
                    \chi^*_{(\Delta)}
                        \circ
                    \rho^*_{(\Delta)}
                    =
                    \left( \rho \circ \chi \right)^*_{(\Delta)}
                $
        \item   \label{itm:weighted_maps_duality}
                For
                $\phi \in \mathfrak{E}(\mathcal{N})$,
                $f    \in \mathfrak{D}(\mathcal{M})$
                \begin{equation*}
                    \int_{\mathcal{N}}
                        \phi \left( \chi_*^{(\Delta)} f \right)
                    \mathrm{dVol}_\mathcal{N}
                            =
                    \int_{\mathcal{M}}
                        \left( \chi^*_{(d - \Delta)} \phi \right) f
                    \mathrm{dVol}_\mathcal{M},
                \end{equation*}
                where $d = \mathrm{Dim}(\mathcal{M}) = \mathrm{Dim}(\mathcal{N})$.
    \end{enumerate}
\end{prop}

\begin{proof}
    The first of these results is easiest to see as a consequence of
    the other two, thus we defer its proof until the end.

    Result \ref{itm:weighted_pullback_composition} can be obtained by
    a direct computation.
    Firstly, note that if
    $\chi^* g_\mathcal{N} = \Omega_\chi^2 g_\mathcal{M}$, and
    $\rho^* g_\mathcal{O} = \Omega_\rho^2 g_\mathcal{N}$,
    then the conformal factor for $\rho \circ \chi$ is given by
    $
        (\rho \circ \chi)^* g_\mathcal{O}
            =
        (\Omega_\chi \cdot \chi^* \Omega_\rho)^2
        g_\mathcal{M}.
    $
    If we select some arbitrary $\phi \in \mathfrak{E}(\mathcal{O})$, then
    \begin{align*}
        \chi^*_{(\Delta)} \left( \rho^*_{(\Delta)} \phi \right)
       &    =
        \chi^*_{(\Delta)} \left( \Omega_\rho^\Delta \rho^* \phi \right)
            \\
       &    =
        \left(\Omega_\chi \cdot ( \chi^* \Omega_\rho ) \right)^\Delta
        \left( \chi^* \rho^* \phi \right)
            \\
       &    =
        \left( \rho \circ \chi \right)^*_{(\Delta)} \phi.
    \end{align*}

    To prove \ref{itm:weighted_maps_duality}, first note that, because
    $\mathrm{supp} \, \big( \chi_*^{(\Delta)} f \big) \subseteq \chi(\mathcal{M})$,
    we may restrict the first integral to $\chi(\mathcal{M})$,
    where we may consider $\chi$ to be a diffeomorphism.
    Next, recall that a standard result for conformal transformations states
    $\chi^* (\mathrm{dVol}_\mathcal{N}) = \Omega^d \mathrm{dVol}_\mathcal{M}$.
    From this we find 
    \begin{align*}
        \chi^*  \left(
                    \phi \cdot (\chi_*^{(\Delta)} f) \cdot \mathrm{dVol}_\mathcal{N}
                \right)
       &    =
        \left( \chi^* \phi \right) \cdot
        \left( \Omega^{-\Delta} f \right) \cdot
        \left( \Omega^d \mathrm{dVol}_\mathcal{M} \right)
            \\
       &    =
        \left( \chi^*_{(d - \Delta)} \phi \right) f \, \mathrm{dVol}_\mathcal{M}.
    \end{align*}

    Finally, to prove \ref{itm:weighted_pushforward_composition},
    let $f \in \mathfrak{D}(\mathcal{M})$ and take
    some arbitrary test function $h \in \mathfrak{D}(\mathcal{O})$.
    Then, consider
    $
        \int_{\mathcal{O}}
            h
            \left( \rho_*^{(\Delta)} \chi_*^{(\Delta)} f \right)
        \, \mathrm{dVol}_\mathcal{O}.
    $
    Using the two results we have just established, we see that
    \begin{align*}
        \int_{\mathcal{O}}
            h
            \left( \rho_*^{(\Delta)} \chi_*^{(\Delta)} f \right)
        \, \mathrm{dVol}_\mathcal{O}
       &    =
        \int_\mathcal{M}
            \left( \chi^*_{(d - \Delta)} \rho^*_{(d - \Delta)} h \right)
            f
        \, \mathrm{dVol}_\mathcal{M}
            \\
       &    =
        \int_\mathcal{M}
            \left( (\rho \circ \chi)^*_{(d - \Delta)} h \right)
            f
        \, \mathrm{dVol}_\mathcal{M}
            \\
       &    =
        \int_{\mathcal{O}}
            h
            \left( (\rho_ \circ \chi)_*^{(\Delta)} f \right)
        \, \mathrm{dVol}_\mathcal{O}.
    \end{align*}
    Thus, as this holds for every choice of $h \in \mathfrak{D}(\mathcal{O})$,
    we can conclude that
    $
        \rho_*^{(\Delta)} \chi_*^{(\Delta)} f
            =
        (\rho_ \circ \chi)_*^{(\Delta)} f.
    $
\end{proof}

Using these definitions,
we can then state the condition required for the theory arising from
a natural Lagrangian $\mathcal{L}$ to be conformally covariant.

\begin{definition}[Conformal Natural Lagrangian]
    Let $\mathcal{L}: \mathfrak{D} \Rightarrow \mathfrak{F}_\mathrm{loc}$ be
    a natural Lagrangian as per \Cref{sec:normal_ordering}.
    Suppose there exists $\Delta \in \mathbb{R}$ such that,
    for every conformally admissible embedding
    $\chi \in \mathrm{Hom}_{\mathsf{CLoc}}(\mathcal{M}; \mathcal{N})$,
    every $\phi \in \mathfrak{E}(\mathcal{N})$,
    and every $f \in \mathfrak{D}(\mathcal{M})$
    \begin{equation}
        \label{eq:conformal_lagrangian_cond}
        \left\langle 
            S_\mathcal{M}'[\chi^*_{(\Delta)} \phi],
            f
        \right\rangle
            =
        \left\langle 
            S_\mathcal{N}'[\phi],
            \chi_*^{(\Delta)} f
        \right\rangle,
    \end{equation}
    where $S'_\mathcal{M}$ is the Euler-Lagrange derivative of
    $\mathcal{L}_\mathcal{M}$ as defined in
    \eqref{eq:Euler_Lagrange}.
    In this case, we call $\mathcal{L}$ a \textit{conformal natural Lagrangian}.
\end{definition}
 
We can state this condition more elegantly by
once again taking the BV perspective where,
instead of focussing on the natural Lagrangian $\mathcal{L}$,
we use its associated differential
$\delta_S: \mathfrak{V}_{\mu c} \Rightarrow \mathfrak{F}_{\mu c}$.

Firstly, we can use the weighted pullback to define a modification of
the functor assigning a spacetime its classical observables,
$\mathfrak{F}_{\mu c}$.
For $\Delta \in \mathbb{R}$, let $\mathfrak{F}_\mathrm{\mu c}^{(\Delta)}$ be
a functor $\mathsf{CLoc} \to \mathsf{Vec}$ which assigns
to each spacetime $\mathcal{M}$ its microcausal observables
$\mathfrak{F}_\mathrm{\mu c}(\mathcal{M})$ as usual,
but assigns to
$\chi \in \mathrm{Hom}_{\mathsf{CLoc}}(\mathcal{M}; \mathcal{N})$
the morphism
\begin{equation}
    \label{eq:F_mc_delta_morphisms}
    (\mathfrak{F}_\mathrm{\mu c}^{(\Delta)}\chi \mathcal{F})[\phi]
        :=
    \mathcal{F}[\chi^*_{(\Delta)} \phi].
\end{equation}
\Cref{prop:weighted_map_facts} assures us these morphisms compose as they should.
Moreover, by using
\begin{equation}
    \label{eq:F_mc_delta_derivatives}
    \left( \mathfrak{F}_{\mu c}^{(\Delta)}\chi \mathcal{F} \right)^{(n)}[\phi]
        =
    \left( \chi_*^{(d - \Delta)} \right)^{\otimes n}
    \mathcal{F}^{(n)}[\chi^*_{(\Delta)} \phi],
\end{equation}
we can see that the wavefront sets of
functional derivatives are independent of the choice of $\Delta$.
Then, by noting that the joint future/past lightcones
$\overline{V}_\pm^n$ are preserved under pullback by $\chi$
are, and both preserved under pushforward by a conformal embedding,
the wavefront set spectral condition
\eqref{eq:spectral_condition} is also preserved.
Hence
$
    \mathfrak{F}_{\mu c}^{(\Delta)}\chi:
    \mathfrak{F}_{\mu c}(\mathcal{M})
        \to
    \mathfrak{F}_{\mu c}(\mathcal{N})
$
as desired.

Similarly to $\mathfrak{F}_\mathrm{\mu c}$, for any choice of weight $\Delta$,
we can define an extention
$\mathfrak{V}^{(\Delta)}_{\mu c}: \mathsf{CLoc} \to \mathsf{Vec}$ by
\begin{equation*}
    \left( 
        \mathfrak{V}_{\mu c} \chi X
    \right)[\phi]
        =
    \chi_*^{(\Delta)} (X[\chi^*_{(\Delta)} \phi]),
\end{equation*}
where $\chi_*^{(\Delta)}$ is again the weighted pushforward of test functions.
Recall that we defined local covariance in the BV formalism as
the condition that $\delta_S$ is a natural transformation
$\mathfrak{V}_{\mu c} \Rightarrow \mathfrak{F}_{\mu c}$,
where each is a functor $\mathsf{Loc} \to \mathsf{Vec}$.
Similarly, \eqref{eq:conformal_lagrangian_cond} simply states that
such a theory is conformally covariant if
the same collection of maps comprising $\delta_S$
also define a natural transformation
$
    \delta_S:
    \mathfrak{V}^{(\Delta)}_{\mu c} \Rightarrow
    \mathfrak{F}^{(\Delta)}_\mathrm{\mu c},
$
where each is now a functor $\mathsf{CLoc} \to \mathsf{Vec}$.

\subsubsection{Conformally Covariant Classical Field Theory}
\label{sec:conformally_covariant_classical}

We can now see how the criterion for conformal covariance
that has just been outlined gives rise to classical dynamical structures
which vary as one would expect under conformal transformations.
The first result compares the linearised equations of motion
on two spacetimes related by a conformally admissible embedding.

\begin{prop}
    Let $\mathcal{L}$ be a conformal natural Lagrangian which satisfies
    the linearisation hypothesis \eqref{eq:linearisation_hypothesis}.
    If $\chi \in \mathrm{Hom}_{\mathsf{CLoc}}(\mathcal{M}; \mathcal{N})$ and
    $\phi \in \mathfrak{E}(\mathcal{N})$, then
    \begin{equation}
        \label{eq:linearised_eom_relation_test}
        \chi_*^{(d - \Delta)} P_\mathcal{M}[\chi^*_{(\Delta)} \phi]
            =
        P_\mathcal{N}[\phi] \chi_*^{(\Delta)},
    \end{equation}
    where each differential operator has been implicitly restricted to the space of
    \textit{test} functions of the appropriate spacetime.
\end{prop}

\begin{proof}
    The proof is effectively a direct computation.
    Let $g \in \mathfrak{D}(\mathcal{M})$ and $h \in \mathfrak{D}(\mathcal{N})$.
    Recall from the definition of $P_\mathcal{N}$ that
    \begin{equation}
        \label{eq:conformally_invariant_proof_1}
        \left\langle
            P_\mathcal{N}[\phi] \chi_*^{(\Delta)}g,
            h
        \right\rangle_{\mathcal{N}}
            =
        \left\langle 
            S_\mathcal{N}''[\phi],
            \left( \chi_*^{(\Delta)}g \right) \otimes h
        \right\rangle_{\mathcal{N}}.
    \end{equation}
    This then allows us to employ \eqref{eq:conformal_lagrangian_cond} as
    \begin{align}
        \left\langle 
            S_\mathcal{N}''[\phi],
            \left( \chi_*^{(\Delta)}g \right) \otimes h
        \right\rangle_{\mathcal{N}}
       &    =
        \frac{d}{d \epsilon}
        \left\langle
            S_\mathcal{N}'
            \left[
                \phi + \epsilon h
            \right],
            \chi_*^{(\Delta)} g
        \right\rangle_\mathcal{N}
        \Big|_{\epsilon = 0}
            \nonumber \\
       &    =
        \frac{d}{d \epsilon}
        \left\langle
            S_\mathcal{M}'
            \left[
                \chi^*_{(\Delta)}\phi + \epsilon \chi^*_{(\Delta)}h
            \right],
            g
        \right\rangle_\mathcal{M}
        \Big|_{\epsilon = 0}
        \nonumber \\
       &    =
        \left\langle
            P_\mathcal{M}\left( \chi^*_{(\Delta)}\phi \right) g,
            \chi^*_{(\Delta)}h
        \right\rangle_\mathcal{M}
            \nonumber \\
        \label{eq:conformally_invariant_proof_2}
       &    =
        \left\langle
            \chi_*^{(d - \Delta)}
            P_\mathcal{M}[\chi^*_{(\Delta)} \phi] g,
            h
        \right\rangle_\mathcal{N}.
    \end{align}
    Note the first equality is not immediately obvious:
    rather, it follows from the locality of $\mathcal{L}_\mathcal{N}$.
    In the following line we use \eqref{eq:conformal_lagrangian_cond}
    and, for the final equality, we note that
    $\chi_*^{(d - \Delta)}$ is the adjoint of $\chi^*_{(\Delta)}$.
    As the choice of $h$ is arbitrary, we may then conclude that
    the two operators coincide.
\end{proof}

\begin{remark}
    As $P_\mathcal{N}[\phi]$ and $P_\mathcal{M}[\chi^*_{(\Delta)} \phi]$
    are both self-adjoint, we can write an equivalent form of
    \eqref{eq:linearised_eom_relation_test} for linear maps
    $\mathfrak{E}(\mathcal{N})$, namely
    \begin{equation}
        P_\mathcal{M}[\chi^*_{(\Delta)} \phi] \chi^*_{(\Delta)}
            =
        \chi^*_{(d - \Delta)} P_\mathcal{N}[\phi].
    \end{equation}
    Using this equation, we can immediately see that the solution spaces for
    these two operators are closely related:
    if $\psi$ is a solution to $P_\mathcal{N}[\phi]$,
    then $\chi^*_{(\Delta)} \psi$ is a solution to
    $P_\mathcal{M}[\chi^*_{(\Delta)}\phi]$.

    Moreover if, for $\lambda > 0$, we take
    $
        \mathcal{N}
            =
        (
            M,
            \lambda^2 g_\mathcal{M},
            \mathfrak{o}_\mathcal{M},
            \mathfrak{t}_\mathcal{M}
        )
    $,
    i.e. just $\mathcal{M}$ with the metric scaled by some factor $\lambda^2$
    and $\chi = \mathrm{Id}_M$,
    then $\chi^*_{(\Delta)} \psi = \lambda^\Delta \psi$.
    This indicates that $\Delta$ is
    what is typically referred to in the literature as
    the \textit{scaling dimension} of the field $\phi$.
\end{remark}

When a pair of normally-hyperbolic differential operators are
related in the above manner,
we can similarly relate their fundamental solutions.
The following proposition,
which reduces to \cite[Lemma 2.2]{pinamontiConformalGenerallyCovariant2009}
in the particular case of the conformally coupled Klein-Gordon field in 4D,
establishes the conformal covariance Pauli-Jordan function
arising from a suitable conformal natural Lagrangian.
To simplify notation, we shall refer only to
a single differential operator on each spacetime,
i.e. we suppress the dependence on an initial field configuration
$\phi$ or $\chi^*_{(\Delta)} \phi$,
though this does not mean that the scope of the result is limited to free theories.

\begin{prop}
    \label{prop:conformal_adv/ret}
    Let $\chi \in \mathrm{Hom}_{\mathsf{CLoc}}(\mathcal{M}; \mathcal{N})$,
    and let $P_\mathcal{M}$, $P_\mathcal{N}$ be
    a pair of symmetric, normally hyperbolic differential operators on
    $\mathcal{M}$ and $\mathcal{N}$ respectively such that
    \begin{equation}
        \label{eq:conformally_covariant_eom}
        P_\mathcal{M} \chi^*_{(\Delta)}
            =
        \chi^*_{(d - \Delta)} P_\mathcal{N}.
    \end{equation}
    If $E^{R/A}_{\mathcal{M} / \mathcal{N}}$ denotes
    the advanced/retarded propagator for
    $P_{\mathcal{M} / \mathcal{N}}$ as appropriate,
    then
    \begin{equation}
        \label{eq:conformal_causal_prop}
        E_\mathcal{M}^{R/A}
            =
        \chi^*_{(\Delta)}  
        E_\mathcal{N}^{R/A}
        \chi_*^{(d - \Delta)}.
    \end{equation}
\end{prop}

\begin{proof}
    Recall that the advanced and retarded propagators of $P_\mathcal{M}$
    are uniquely determined by
    their composition with $P_\mathcal{M}$ and their support properties.
    As such, we simply need to establish that the operator on the right-hand side of
    \eqref{eq:conformal_causal_prop}
    satisfies the relevant criteria.

    Firstly, if we act on this operator with $P_\mathcal{M}$ we see
    \begin{align*}
        P_\mathcal{M} \,
        \chi^*_{(\Delta)}
        E_\mathcal{N}^{R/A}
        \chi_*^{(d - \Delta)}
                =
        \chi^*_{(d - \Delta)}
        P_\mathcal{N}
        E^{R/A}_\mathcal{N}
        \chi_*^{(d - \Delta)}.
    \end{align*}
    By definition,
    $
        P_\mathcal{N} \circ E^{R / A}_\mathcal{N}
            =
        \mathds{1}_{\mathfrak{D}(\mathcal{N})},
    $
    and clearly
    $
        \chi^*_{(d - \Delta)} \chi_*^{(d - \Delta)}
            =
        \mathds{1}_{\mathfrak{D}(\mathcal{M})},
    $
    hence
    \begin{equation}
        P_\mathcal{M}
        \left(
            \chi^*_{(\Delta)}
            E_\mathcal{N}^{R/A}
            \chi_*^{(d - \Delta)}
        \right)
                =
        \mathds{1}_{\mathfrak{D}(\mathcal{M})}.
    \end{equation}

    If we denote by $P_{\mathcal{M}}^c$ the restriction of
    $P_{\mathcal{M}}$ to $\mathfrak{D}(\mathcal{M})$,
    and likewise $P_\mathcal{N}^c$,
    by the symmetry of these operators, we have that
    \begin{equation*}
        \chi_*^{(d - \Delta)} P^c_\mathcal{M}
            =
        P^c_\mathcal{N} \chi_*^{(\Delta)}.
    \end{equation*}
    Thus, acting on $P^c_\mathcal{M}$ with our candidate propagator yields
    \begin{equation*}
        \chi^*_{(\Delta)}
        E_\mathcal{N}^{R/A}
        \chi_*^{(d - \Delta)}
        P^c_\mathcal{M}
                =
        \chi^*_{(\Delta)}
        E_\mathcal{N}^{R/A}
        P_\mathcal{N}^c
        \chi_*^{(\Delta)},
    \end{equation*}
    which is again simply $\mathds{1}_{\mathfrak{D}(\mathcal{M})}$.

    Finally, we must determine the supports of these functions.
    Let $f \in \mathfrak{D}(\mathcal{M})$.
    Note that $\mathrm{supp} \, (\chi_*^{(d - \Delta)} f) = \chi(\mathrm{supp} \, f)$,
    hence, using the support property of $E^{R/A}$
    \begin{equation*}
        \mathrm{supp} \,
            \left(
                E^{R/A}_\mathcal{N} \chi_*^{(d - \Delta)} f
            \right)
                    \subseteq
        \mathscr{J}^\pm_\mathcal{N}\left(\chi \left( \mathrm{supp} \, f \right) \right).
    \end{equation*}
    Pulling this back to $\mathcal{M}$, we have
    \begin{equation*}
        \mathrm{supp} \,
        \left( 
            \chi^*_{(\Delta)} E^{R/A}_\mathcal{N} \chi_*^{(d - \Delta)} f
        \right)
            \subseteq
        \chi^{-1}\left(
            \mathscr{J}^\pm_\mathcal{N}
            \left( 
                \chi(\mathrm{supp} \, f)
            \right)
        \right).
    \end{equation*}
    Conformally admissible embeddings preserve causal structure.
    In particular, if $\gamma: [0, 1] \to \mathcal{M}$ is a causal,
    future/past-directed curve, then
    $\chi \circ \gamma$ is also causal and future/past-directed.
    This means that
    $
        \chi\left( \mathscr{J}^\pm_\mathcal{M}(\mathrm{supp} \, f) \right)
            =
        \mathscr{J}^\pm_\mathcal{N}\left( \chi(\mathrm{supp} \, f) \right).
    $
    Hence our candidate propagators also meet the desired support criteria,
    and must genuinely be the advanced and retarded propagators for $P_\mathcal{M}$
    as required.
\end{proof}

One can show that conformal invariance as defined in appendix D of
\cite{waldGeneralRelativity2010}
implies \eqref{eq:conformally_covariant_eom},
so long as it is also assumed that $P_\mathcal{M}$ and $P_\mathcal{N}$
are \textit{symmetric} in the sense that
$
    \left\langle f, P_\mathcal{M} \phi \right\rangle_\mathcal{M}
        =
    \left\langle P_\mathcal{M} f, \phi \right\rangle_\mathcal{M}
$ for all $f \in \mathfrak{D}(\mathcal{M})$, $\phi \in \mathfrak{E}(\mathcal{M})$.

Similar to the case of (isometric) local covariance,
the consequence of \cref{prop:conformal_adv/ret}
is that we can define a symplectomorphism from 
the solution space of $P_\mathcal{M}$ to that of $P_\mathcal{N}$.
Recall that we can identify the space of solutions to $P_\mathcal{M}$
with $\mathfrak{D}(\mathcal{M})/P_\mathcal{M}\left( \mathfrak{D}(\mathcal{M}) \right)$.
If $f, g \in \mathfrak{D}(\mathcal{M})$, then
\begin{equation}
    \label{eq:causal_prop_covar}
    \left\langle f, E_\mathcal{M} g \right\rangle
        =
    \left\langle \chi_*^{(d - \Delta)} f,
        E_\mathcal{N} \left( \chi_*^{(d - \Delta)} g \right) \right\rangle.
\end{equation}
Moreover, from \eqref{eq:linearised_eom_relation_test}, it follows that
$
    \chi_*^{(d - \Delta)} \left( P_\mathcal{M}
        \left( \mathfrak{D}(\mathcal{M}) \right) \right)
            \subseteq
    P_\mathcal{N}
        \left( \mathfrak{D}(\mathcal{N}) \right)
$,
hence $\chi^{(\Delta)}_*$ yields a well-defined map between the quotient spaces
$
    \mathfrak{D}(\mathcal{M}) / P_{\mathcal{M}}
        \left( \mathfrak{D}(\mathcal{M}) \right)
            \to
    \mathfrak{D}(\mathcal{N}) / P_{\mathcal{N}}
        \left( \mathfrak{D}(\mathcal{N}) \right).
$

As was the case in \Cref{sec:normal_ordering},
this symplectomorphism of solution spaces in turn gives rise to
a Poisson algebra homomorphism relating the Peierls brackets for each spacetime.
A quick calculation shows that the map $\mathfrak{F}_{\mu c}^{(\Delta)}\chi$
defined in \eqref{eq:F_mc_delta_morphisms}
is a Poisson algebra homomorphism:
for
$\mathcal{F}, \mathcal{G} \in \mathfrak{F}_{\mu c}(\mathcal{M})$,
$\phi \in \mathfrak{E}(\mathcal{N})$
we have that
\begin{align*}
    \left\{ 
        \mathfrak{F}_{\mu c}^{(\Delta)} \chi \mathcal{F},
        \mathfrak{F}_{\mu c}^{(\Delta)} \chi \mathcal{G}
    \right\}_\mathcal{N}[\phi]
   &    =
    \left\langle 
        \left( \mathfrak{F}_{\mu c}^{(\Delta)} \chi \mathcal{F} \right)^{(1)}[\phi],
        E_\mathcal{N}(\phi)
        \left( \mathfrak{F}_{\mu c}^{(\Delta)} \chi \mathcal{G} \right)^{(1)}[\phi]
    \right\rangle_\mathcal{N}
        \\
   &    =
    \left\langle 
        \chi_*^{(d - \Delta)} \mathcal{F}^{(1)}[\chi^*_{(\Delta)} \phi],
        E_\mathcal{N}(\phi)
        \chi_*^{(d - \Delta)} \mathcal{G}^{(1)}[\chi^*_{(\Delta)} \phi]
    \right\rangle_\mathcal{N}
        \\
   &    =
    \left\langle 
        \mathcal{F}^{(1)}[\chi^*_{(\Delta)} \phi],
        E_\mathcal{M}(\chi^*_{(\Delta)} \phi)
        \mathcal{G}^{(1)}[\chi^*_{(\Delta)} \phi]
    \right\rangle_\mathcal{M}
        \\
   &    =
    \left(
        \mathfrak{F}_{\mu c}^{(\Delta)} \chi
        \left\{ 
            \mathcal{F},
            \mathcal{G}
        \right\}_\mathcal{M}
    \right)[\phi].
\end{align*}

We may summarise the above results as ensuring that the following is well-defined:

\begin{definition}[Locally Conformally Covariant Classical Field Theory]
    \label{def:LCCCFT}
    For some $\Delta \in \mathbb{R}$, let $\mathcal{L}$
    be a conformal natural Lagrangian of weight $\Delta$.
    The \textit{locally conformally covariant classical field theory}
    associated to $\mathcal{L}$ is a functor
    $\mathfrak{P}: \mathsf{CLoc} \to \mathsf{Poi}$,
    which assigns
    \begin{itemize}
        \item   To every spacetime $\mathcal{M} \in \mathsf{CLoc}$,
                the algebra $\mathfrak{F}_{\mu c}(\mathcal{M})$
                equipped with the Peierls bracket
                $\left\{ \cdot, \cdot\right\}_\mathcal{M}$
                associated to the generalised Lagrangian $\mathcal{L}_\mathcal{M}$.
        \item   To every morphism
                $\chi \in \mathrm{Hom}_{\mathsf{CLoc}}(\mathcal{M}; \mathcal{N})$,
                the Poisson algebra homomorphism
                $\mathfrak{F}_{\mu c}^{(\Delta)}\chi$.
    \end{itemize}
\end{definition}

\begin{example}[The Conformally Coupled Scalar Field]
    The simplest example of a conformal natural Lagrangian is that of
    the conformally coupled scalar field.
    For spacetimes of dimension $d$, this is given by,
    for
    $\mathcal{M} \in \mathsf{CLoc}$,
    $f \in \mathfrak{D}(\mathcal{M})$,
    $\phi \in \mathfrak{E}(\mathcal{M})$
    \begin{equation}
        \mathcal{L}_\mathcal{M}(f)[\phi]
            :=
        \frac{1}{2}
        \int_\mathcal{M}
            f
            \left[
                g_\mathcal{M}\left( 
                    \nabla \phi, \nabla \phi
                \right)
                    +
                \xi_d R_\mathcal{M} \phi^2
            \right]
            \, \mathrm{dVol}_\mathcal{M},
    \end{equation}
    where $R_\mathcal{M}$ is the scalar curvature function
    for the spacetime $\mathcal{M}$ and
    $\xi_d = \tfrac{d-2}{4(d - 1)}$ is the conformal coupling constant.

    In this case, we can see that the Euler-Lagrange derivative
    satisfies the desired covariance property with
    $\Delta = \tfrac{(d - 2)}{2}$.

    Even in this example we see the necessity of phrasing
    \eqref{eq:conformal_lagrangian_cond}
    in terms of variations of the action.
    Na\"ively, we may have assumed conformal covariance to be given by
    $
        \mathcal{L}_\mathcal{M}(f)[\chi^*_{(\Delta)} \phi]
            =
        \mathcal{L}_\mathcal{N}(\chi_*^{(\Delta)} f)[\phi]
    $.
    However, the presence of the test function $f$ in the above Lagrangian
    prevents the integration by parts necessary for this equation to hold.
\end{example}

\subsubsection{Conformally Covariant Quantum Field Theory}

In order to discuss quantisation,
we must return our attention to free field theories.
In doing so we can once again refer unambiguously to a single operator
$P_\mathcal{M}$ producing the equations of motion on $\mathcal{M}$.

We saw in \cref{sec:deformation_quantisation} that
quantisation of a free field theory is achieved through the use of
arbitrarily selected Hadamard distributions for each $P_\mathcal{M}$.
The covariance of the quantum algebras was thus dependent on
the fact that, given an admissible embedding
$\chi: \mathcal{M} \to \mathcal{N}$,
the pullback of a Hadamard distribution on $\mathcal{N}$ by $\chi$ is
again a Hadamard distrbution on $\mathcal{N}$.
We have already seen that the \textit{weighted} pullback of
the causal propagator on $\mathcal{N}$ is the causal propagator on $\mathcal{M}$.
The following proof, again adapted from \cite{pinamontiConformalGenerallyCovariant2009},
gives the corresponding result for Hadamard distributions.

\begin{prop}
    Let $\chi \in \mathrm{Hom}_{\mathsf{CLoc}}(\mathcal{M}; \mathcal{N})$
    be a conformally admissible embedding with conformal factor $\Omega$,
    and let $P_\mathcal{M}, P_\mathcal{N}$ be a pair of normally hyperbolic
    differential operators satisfying
    \begin{equation*}
        P_\mathcal{M} \chi^*_{(\Delta)}
            =
        \chi^*_{(d - \Delta)} P_\mathcal{N}.
    \end{equation*}
    If $W_\mathcal{N}: \mathfrak{D}(\mathcal{N}) \to \mathfrak{E}(\mathcal{N})$
    is a Hadamard distribution for $P_\mathcal{N}$, then
    \begin{equation}
        W_\mathcal{M}
            :=
        \chi^*_{(\Delta)} W_\mathcal{N} \chi_*^{(d - \Delta)}
    \end{equation}
    is a Hadamard distribution for $P_\mathcal{M}$.
\end{prop}

\begin{proof}
    Firstly, \eqref{eq:causal_prop_covar}
    ensures that the anti-symmetric part of $W_\mathcal{M}$ is
    $\tfrac{i}{2}E_\mathcal{M}$.
    Secondly, by a direct computation, we can see that
    $P_\mathcal{M} W_\mathcal{M} \equiv 0$,
    hence $W_\mathcal{M}$ is a distibutional solution to $P_\mathcal{M}$.
    Thirdly, upon complexification of
    $\mathfrak{D}(\mathcal{M})$ and
    $\mathfrak{D}(\mathcal{N})$,
    we clearly have that
    $
        \chi_*^{(d - \Delta)} \overline{f}
            =
        \overline{\big( \chi_*^{(d - \Delta)} f \big)},
    $
    hence positivity of $W_\mathcal{M}$ follows directly from that of $W_\mathcal{N}$.

    Thus, all that remains to be shown is that $W_\mathcal{M}$ has
    the appropriate wavefront set:

    As a distribution in $\mathfrak{D}'(\mathcal{M}^2)$,
    as opposed to a continuous map
    $\mathfrak{D}(\mathcal{M}) \to \mathfrak{E}(\mathcal{M})$,
    $W_\mathcal{M}$ is defined on the dense subspace
    $\mathfrak{D}(\mathcal{M})^{\otimes 2} \subset \mathfrak{D}(\mathcal{M}^2)$ by
    \begin{equation}
        \left\langle 
            W_\mathcal{M},
            f \otimes g
        \right\rangle
            =
        \left\langle 
            W_\mathcal{N},
            \chi_*^{(d - \Delta)} f \otimes \chi_*^{(d - \Delta)} g
        \right\rangle.
    \end{equation}
    This differs from the usual pullback $\chi^*W_\mathcal{N}$ only in the
    multiplication by the smooth function
    $\Omega^{d - \Delta} \otimes \Omega^{d - \Delta}$,
    hence
    $
        \mathrm{WF}(W_\mathcal{M})
            =
        \mathrm{WF}\left((\chi^*)^{\otimes 2} W_\mathcal{N}\right)
    $.

    At this point it is convenient to regard $\chi(\mathcal{M})$ as
    a spacetime in its own right,
    with all the relevant data being that inherited from $\mathcal{N}$ by restriction.
    We then observe that $\chi$ factorises as $\iota \circ \xi$,
    where the inclusion
    $\iota: \chi(\mathcal{M}) \hookrightarrow \mathcal{N}$
    is an isometric embedding, and
    $\xi: \mathcal{M} \to \chi(\mathcal{M})$
    is a conformal \textit{diffeomorphism}.
    With this,
    we write $\chi^* W_\mathcal{N} = \xi^* \left( \iota^* W_\mathcal{N} \right)$.
    As $\xi$ is a diffeomorphism, we know that
    $
        \mathrm{WF}\left(\xi^*\left( \iota^* W_\mathcal{N} \right)\right)
            =
        \xi^* \mathrm{WF}(\iota^* W_\mathcal{N}),
    $
    and, since $\iota$ is an isometric admissible embedding
    $
        \mathrm{WF}(\iota^* W_\mathcal{N})
            =
        \Gamma_{\chi(\mathcal{M})},
    $
    where $\Gamma_\mathcal{M} = \mathrm{WF}(W)$ for any (and hence every)
    Hadamard distribution $W$ on $\mathcal{M}$.
    
    It is only left for us to show that
    $\xi^* \Gamma_{\chi(\mathcal{M})} = \Gamma_\mathcal{M}$.
    Let $(y_1, y_2; \eta_1, \eta_2) \in \Gamma_{\chi(\mathcal{M})}$,
    and let $\gamma: (-\epsilon, 1 + \epsilon)$ be a null geodesic satisfying
    $\gamma(0) = y_1$,
    $\gamma(1) = y_2$,
    $\dot{\gamma}^\flat(0) = \eta_1$,
    $\dot{\gamma}^\flat(1) = \eta_2$.
    It is then readily verified that $\xi^{-1} \circ \gamma$ is
    a null geodesic strip which demonstrates
    $(x_1, x_2; k_1, k_2) \in \Gamma_\mathcal{M}$,
    where $y_i = \xi(x_i)$, and $k_i = \eta_i \circ d\xi|_{x_i}$.
    Thus we see that $\xi^* \Gamma_{\chi(\mathcal{M})} \subseteq \Gamma_\mathcal{M}$.
    Similarly, if $\widetilde{\gamma}$ is a null geodesic stip demonstrating
    $(x_1, x_2; k_1, k_2) \in \Gamma_\mathcal{M}$,
    then $\gamma := \xi \circ \widetilde{\gamma}$ shows that
    $(y_1, y_2; \eta_1, \eta_2) \in \Gamma_{\chi(\mathcal{M})}$.
    From this we can conclude that
    $
        \mathrm{WF}(W_\mathcal{M})
            =
        \mathrm{WF}(\chi^* W_\mathcal{N})
            =
        \Gamma_\mathcal{M}
    $,
    hence $W_\mathcal{M}$ is indeed a Hadamard distribution for $P_\mathcal{M}$.
\end{proof}

If we, by a slight abuse of notation,
write $W_\mathcal{M} = \chi^*_{(\Delta)} W_\mathcal{N}$,
then the above proposition can be expressed as
$\chi^*_{(\Delta)}: \mathrm{Had}(\mathcal{N}) \to \mathrm{Had}(\mathcal{M})$.
This map, together with the map $\mathfrak{F}_{\mu c}^{(\Delta)} \chi$
defined in the previous section, creates the algebra homomorphism
required to make the quantum theory conformally covariant.

Firstly we observe that, if $H_\mathcal{M}$
is the symmetric part of $W_\mathcal{M}$ \textit{etc},
then a quick computation confirms that
\begin{equation*}
    \left( 
        \mathfrak{F}_{\mu c}^{(\Delta)} \chi
        \mathcal{F}
    \right)
        \star_{H_\mathcal{N}}
    \left( 
        \mathfrak{F}_{\mu c}^{(\Delta)} \chi
        \mathcal{G}
    \right)
        =
    \mathfrak{F}_{\mu c}^{(\Delta)} \chi \left( 
        \mathcal{F} \star_{H_\mathcal{M}} \mathcal{G}
    \right).
\end{equation*}
In other words, for a Hadamard distribution
$H_\mathcal{N} \in \mathrm{Had}(\mathcal{N})$,
$\mathfrak{F}_{\mu c}^{(\Delta)} \chi$ defines a $*$-algebra homomorphism
${
    \mathfrak{A}^{H_\mathcal{M}}(\mathcal{M}) \to
    \mathfrak{A}^{H_\mathcal{N}}(\mathcal{N})
}$,
using the notation introduced in \eqref{eq:A^H_algebras}.

To see that these maps define a homomorphism
$\mathfrak{A}(\mathcal{M}) \to \mathfrak{A}(\mathcal{N})$
note that, if $H'_\mathcal{N} \in \mathrm{Had}(\mathcal{N})$ and
$H'_\mathcal{M} := \chi^*_{(\Delta)} H'_\mathcal{N}$ then,
using \eqref{eq:F_mc_delta_derivatives}, one can show that
\begin{equation}
    \alpha_{H'_\mathcal{N} - H_\mathcal{N}}
        \circ
    \mathfrak{F}_{\mu c}^{(\Delta)} \chi
        =
    \mathfrak{F}_{\mu c}^{(\Delta)} \chi
        \circ
    \alpha_{H'_\mathcal{M} - H_\mathcal{M}},
\end{equation}
hence our homomorphisms are compatible with the isomorphisms
between different concrete realisations of $\mathfrak{A}(\mathcal{N})$
as required.

Thus we have shown that the following definition makes sense.

\begin{definition}[The Quantum Massless Scalar Field]
    Let
    $
        \mathcal{L}:
            \mathfrak{D} \Rightarrow
            \mathfrak{F}_\mathrm{loc}
    $
    be the conformal natural Lagrangian of the massless scalar field
    in spacetime dimension $d$.
    The \textit{locally conformally covariant quantum field theory}
    associated to $\mathcal{L}$ is a functor
    $\mathfrak{A}: \mathsf{CLoc} \to *\mbox{-}\mathsf{Alg}$,
    which assigns
    \begin{itemize}
        \item   To every spacetime $\mathcal{M} \in \mathsf{CLoc}$,
                the algebra $\mathfrak{A}(\mathcal{M})$
                defined in \Cref{sec:deformation_quantisation}.
        \item   To every morphism
                $\chi \in \mathrm{Hom}_{\mathsf{CLoc}}(\mathcal{M}; \mathcal{N})$,
                the $*$-algebra homomorphism defined, for
                $
                    \mathcal{F}
                        =
                    (\mathcal{F}_H)_{H \in \mathrm{Had}(\mathcal{M})}
                        \in
                    \mathfrak{A}(\mathcal{M})
                $
                and
                $H_\mathcal{N} \in \mathrm{Had}(\mathcal{N})$, by
                \begin{equation*}
                    \left( 
                        \mathfrak{A}\chi \mathcal{F}
                    \right)_{H_\mathcal{N}}
                        :=
                    \mathfrak{F}_{\mu c}^{(\Delta)} \chi
                    \left(
                        \mathcal{F}_{\chi^*_{(\Delta)} H_\mathcal{N}}
                    \right),
                \end{equation*}
                where $\Delta = \tfrac{d-2}{2}$.
    \end{itemize}
\end{definition}

\subsection{Primary and Quasi-Primary Fields}

Now that we have constructed the quantum theory of the massless scalar field,
we can begin comparing our formalism to the standard CFT literature.
In formulations of CFT descended from the Osterwalder-Schrader axioms,
one defines a field $\phi(z, \bar{z})$,
to be primary with conformal weights
$(h, \widetilde{h}) \in \mathbb{R}^2$ if,
for a holomorphic function $z \mapsto w(z)$
\begin{equation}
    \label{eq:Euclidian_primary}
    \phi(z, \bar{z})
        \mapsto
    \left( \frac{\partial w}{\partial z} \right)^{h}
    \left( \frac{\partial \bar{w}}{\partial \bar{z}} \right)^{\widetilde{h}}
    \phi(w(z), \bar{w}(\bar{z})).
\end{equation}

In order to reach an analogous definition of a primary field within the AQFT framework,
we must equip our spacetimes with frames.
As a motivating example, Minkowski space is naturally equipped with the frame
(in null coordinates) $(du, dv)$.
The Minkowski metric is then simply $ds^2 = du \odot dv$,
where $\odot$ denotes the symmetrised tensor product.
A general conformal automorphism, $\chi$,
of Minkowski space can be written in the form
\begin{equation}
    \chi: (u, v) \mapsto (\mu(u), \nu(v)),
\end{equation}
where either $\mu, \nu \in \mathrm{Diff}_+(\mathbb{R})$ or $\mathrm{Diff}_-(\mathbb{R})$. This is readily shown to be conformal as, for any $(u, v) \in {\mathbb{M}_2}$
\begin{equation}
    \chi^*(du \odot dv)_{(u, v)} = \mu'(u) \nu'(v) (du \odot dv)_{(u, v)}.
\end{equation}
Hence the conformal factor is the product $\Omega^2(u, v) = \mu'(u) \nu'(v)$. 
To generalise this concept to arbitrary globally-hyperbolic spacetimes,
we make the following definition.

\begin{definition}
    The category $\mathsf{CFLoc}$, consists of objects that are tuples
    $\mathscr{M} = (M, (e^\ell, e^\arr))$,
    where $M$ is a $2$-manifold, and $e^\ell, e^\arr$ are a pair of $1$-forms 
    such that, $\forall p \in M$, $\{{e^\ell}_p, {e^\arr}_p\}$ spans $T^*_pM$,
    subject to the condition that the map
    \begin{equation}
        \label{eq:induced_spacetime}
        (M, (e^\ell, e^\arr))
            \mapsto
        (M, e^\ell \odot e^\arr, [e^\ell \wedge e^\arr], [e^\ell + e^\arr])
    \end{equation}
    sends objects in $\mathsf{CFLoc}$ to objects in $\mathsf{Loc}$.

    A morphism
    $\chi: (M, (e^\ell, e^\arr)) \to (N, (\tilde{e}^\ell, \tilde{e}^\arr)))$
    is a smooth embedding $\chi: M \hookrightarrow N$
    such that if $\mathcal{M}$ and $\mathcal{N}$ are the spacetimes obtained
    in the above manner from
    ($M, (e^\ell, e^\arr))$ and
    $(N, (\tilde{e}^\ell, \tilde{e}^\arr)))$ respectively, 
    then
    $
        \chi
            \in
        \mathrm{Hom}_\mathsf{CLoc}\left(\mathcal{M}; \mathcal{N} \right).
    $
    In other words, $\chi$ is a conformally admissible embedding of
    $M$ into $N$ with respect to the metrics and orientations induced by their coframes.
\end{definition}

As every $2$D globally hyperbolic spacetime is parallelisable,
each may be expressed as the spacetime induced by some object of $\mathsf{CFLoc}$,
i.e. the map \eqref{eq:induced_spacetime} is surjective.
Furthermore, from the definition of the morphisms in $\mathsf{CFLoc}$,
it is evident that this map extends to a fully faithful functor
$\mathfrak{p}: \mathsf{CFLoc} \to \mathsf{CLoc}$,
hence we have an \textit{equivalence} between the two
in the sense of category theory.

Rather than relying solely on this equivalence, however,
the following proposition provides a test of whether an embedding
$\chi: M \hookrightarrow N$ is conformally admissible
with respect to the spacetime structure induced by the frames
$(e^\ell, e^\arr)$ and $(\tilde{e}^\ell, \tilde{e}^\arr)$.

\begin{prop}
    \label{prop:conformally_admissible_frames}
    Let
    $\mathscr{M} = (M, (e^\ell, e^\arr))$,
    $\mathscr{N} = (N, (\tilde{e}^\ell, \tilde{e}^\arr)))$
    be two objects in $\mathsf{CFLoc}$,
    a smooth embedding
    $\chi: M \hookrightarrow N$
    is then a $\mathsf{CFLoc}$ morphism
    between $\mathscr{M}$ and $\mathscr{N}$ if and only if
    there exists a pair of smooth, everywhere-positive functions
    $\omega_\ell, \omega_\arr \in \mathfrak{E}_{> 0}(M)$
    such that
    \begin{equation}
        \label{eq:conformally_admissible_frames}
        \chi^* \tilde{e}^{\ell / \arr}
            =
        \omega_{\ell / \arr} e^{\ell / \arr}.
    \end{equation}
\end{prop}

\begin{proof}
    Suppose first that the embedding $\chi$ satisfies
    \eqref{eq:conformally_admissible_frames},
    then it is clearly conformal, as
    \begin{equation}
        \chi^*(\tilde{e}^\ell \odot \tilde{e}^\arr)
            =
        \Omega^2 (e^\ell \odot e^\arr),
    \end{equation}
    where the conformal factor is $\Omega^2 = \omega_\ell \omega_\arr$.
    To show it is admissible, consider first
    \begin{equation}
        \chi^*[\tilde{e}^\ell \wedge \tilde{e}^\arr] := [\chi^*(\tilde{e}^\ell \wedge \tilde{e}^\arr)] = [\omega_\ell \omega_\arr (e^\ell \wedge e^\arr)] = [e^\ell \wedge e^\arr],
    \end{equation}
    where the final equality comes from the fact that the product
    $\omega_\ell \omega_\arr$ is everywhere positive.
    Hence, $\omega_\ell \omega_\arr (e^\ell \wedge e^\arr)$
    defines same orientation as $e^\ell \wedge e^\arr$,
    establishing that $\chi$ is orientation preserving.
    
    Next, to show $\chi$ preserves time orientation, consider
    \begin{equation}
        \chi^* (\tilde{e}^\ell + \tilde{e}^\arr)
            =
        \omega_\ell e^\ell + \omega_\arr e^\arr.
    \end{equation}
    For this $1$-form to define the same time orientation as $e^\ell + e^\arr$,
    first we need to prove it is timelike. Let $g = e^\ell \odot e^\arr$, then
    \begin{equation}
        g(\omega_\ell e^\ell + \omega_\arr e^\arr,
          \omega_\ell e^\ell + \omega_\arr e^\arr)
            =
        2 \omega_\ell \omega_\arr > 0,
    \end{equation}
    hence it is everywhere timelike.
    Next, we need to show it is compatible with the original orientation:
    \begin{equation}
        \label{eq:preserve_time_orientation}
        g(\omega_\ell e^\ell + \omega_\arr e^\arr,
          e^\ell + e^\arr)
            =
        \omega_\ell + \omega_\arr > 0.
    \end{equation}
    Thus \eqref{eq:conformally_admissible_frames} is a sufficient condition for
    $\chi$ to be a conformally admissible embedding.

    Conversely, let us now assume that $\chi$ is conformally admissible.
    Let $\tilde{e}^{\ell / \arr}|_{\chi(M)}$ denote the restriction of
    $\tilde{e}^{\ell / \arr}$ to the image of $M$ under $\chi$.
    As $\chi$ is conformal, the pull-back of each of these $1$-forms must be
    a null $1$-form on $M$ with respect to the induced metric.
    At every point $p \in M$, this tells us that
    $\chi^* \tilde{e}^{\ell}|_{\chi(M)} (p)$ must be colinear with either
    $e^\ell(p)$ or $e^\arr(p)$.
    That it must be colinear with $e^\ell(p)$ in particular is due to the fact that
    $\chi$ preserves orientation;
    a similar argument can then be made for $\tilde{e}^\arr$.
    Thus we have two functions
    $\omega_{\ell / \arr} \in \mathfrak{E}_{> 0}(M)$ such that
    $\chi^* \tilde{e}^{\ell}|_{\chi(M)} (p) = \omega_{\ell / \arr} e^{\ell / \arr}$.
    Their product is the conformal factor of $\chi$ and hence must be positive.
    Finally, for $\chi$ to preserve time orientation,
    $\omega_\ell$ and $\omega_\arr$ must satisfy
    \eqref{eq:preserve_time_orientation}, thus each function must be everywhere-positive.
\end{proof}

Using these frames, we can define a modified pushforward, similar to
\eqref{eq:weighted_pushforward},
except now with a \textit{pair} of weights
$(\lambda, \tilde{\lambda}) \in \mathbb{R}^2$
specified.
The weighted pushforward of a test function $f \in \mathfrak{D}(M)$
under a morphism $\chi: \mathscr{M} \to \mathscr{N}$
with left/right conformal factors $\omega_{\ell / \arr}$ is given by
\begin{equation}
    \chi_*^{(\lambda, \tilde{\lambda})} f
        =
    \chi_* \left(
            \omega_\ell^{- \lambda} \omega_\arr^{- \widetilde{\lambda}} f
        \right).
\end{equation}
We then construct the functor
$\mathfrak{D}^{(h, \widetilde{h})}: \mathsf{CFLoc} \to \mathsf{Vec}$,
for $(h, \widetilde{h}) \in \mathbb{R}^2$ as follows:
for an object $\mathscr{M} \in \mathsf{CFLoc}$,
define $\mathfrak{D}^{(h, \widetilde{h})}(\mathscr{M}) = \mathfrak{D}(M)$,
and for a morphism $\chi: \mathscr{M} \to \mathscr{N}$:
\begin{equation}
    \label{eq:two_weight_morphism}
    \mathfrak{D}^{(h, \widetilde{h})} \chi(f)
        =
    \chi_*^{(1 - h, 1 - \widetilde{h})} f.
\end{equation}

With this functor, we can finally define a \textit{primary field of weight}
$(h, \widetilde{h})$ to be a natural transformation
$\Phi: \mathfrak{D}^{(h, \widetilde{h})}
    \Rightarrow
\mathfrak{A}$,
where $\mathfrak{A}: \mathsf{CFLoc} \to \mathsf{Vec}$
is a locally covariant QFT, which may or may not
be the `pullback' $\widetilde{\mathfrak{A}}\circ \mathfrak{p}$
of some theory $\widetilde{\mathfrak{A}}: \mathsf{CLoc} \to \mathsf{Vec}$.
Explicitly, this means that, if $\mathcal{M}$ is
the spacetime constructed from $\mathscr{M} \in \mathsf{CFLoc}$ according to
\eqref{eq:induced_spacetime}, and likewise $\mathcal{N}$ arises from
$\mathscr{N} \in \mathsf{CFLoc}$, then
we have a pair of linear maps $\Phi_{\mathscr{M}/\mathscr{N}}$ such that,
for any $\chi \in \mathrm{Hom}_{\mathsf{CFLoc}}(\mathscr{M}; \mathscr{N})$,
the following diagram commutes
\begin{equation}
    \label{eq:primary}
    \begin{tikzcd}
        \mathfrak{D}(\mathcal{M})
            \rar["\mathfrak{D}^{(h, \widetilde{h})}\chi"]
            \dar["\Phi_\mathscr{M}"]
            &
        \mathfrak{D}(\mathcal{N})
            \dar["\Phi_\mathscr{N}"]
            \\
        \widetilde{\mathfrak{A}}(\mathcal{M})
            \rar["{\widetilde{\mathfrak{A}}}\chi"]
            &
        \widetilde{\mathfrak{A}}(\mathcal{N})
    \end{tikzcd}
\end{equation}

Heuristically, we can see how this definition relates to
\eqref{eq:Euclidian_primary}
by taking the `limit' of $\Phi_\mathscr{M}(f)$ as
$f \to \delta_x$,
the Dirac delta distribution localised at $x \in M$.
Whilst there is no guarantee that
$\Phi_\mathscr{M}(f)$ converges in this limit,
\eqref{eq:two_weight_morphism} \textit{does} converge
in the weak-$*$ topology to
$\omega_\ell(x)^h \omega_\arr(x)^{\widetilde{h}} \delta_{\chi(x)}$.
If we imagine for a moment that
$\Phi_\mathscr{M}(x) := \lim_{f \to \delta_x} \Phi_\mathscr{M}(f)$
is well-defined, the statement that $\Phi$ is primary with weights
$(h, \widetilde{h})$ implies
\begin{equation}
    \mathfrak{A}\chi \Phi_\mathscr{M}(x)
        =
    \lim_{f \to \delta_x}
    \Phi_\mathscr{N}\left( \mathfrak{D}^{(h, \widetilde{h})}\chi f \right)
        =
    \omega_\ell(x)^h \omega_\arr(x)^{\widetilde{h}} \Phi_\mathscr{N}(\chi(x)).
\end{equation}
Recalling that, if $\chi: \mathbb{M}_2 \to \mathbb{M}_2$
is expressed in null coordinates as
$\chi(u, v) = (\mu(u), \nu(v))$,
then $\omega_\ell = d\mu/du$ and
$\omega_\arr = d\nu/dv$,
we see that we have recovered a Lorentzian signature analogue of
\eqref{eq:Euclidian_primary} as deisred.

We can also recover the physical interpretations of the sum and difference of
$h$ and $\widetilde{h}$, referred to as the
\textit{scaling dimension} $\Delta$ and \textit{spin} $s$ of the field respectively.
For the scalar field, we have already encountered
the scaling dimension as the number $\Delta$ appearing in,
for example, \Cref{def:LCCCFT}.
If we consider a field with spin $s = 0$,
the action of the corresponding $\mathfrak{D}$ functor is
\begin{equation}
    \label{eq:spin_zero_morphism}
    \mathfrak{D}^{\left(\Delta / 2, \Delta / 2\right)} f
        =
    \chi_*^{(2 - \Delta)} f.
\end{equation}
The right hand side of which is precisely the action of the functor
$\mathfrak{D}^{(\Delta)}$ as defined in \cite{pinamontiConformalGenerallyCovariant2009}.
Hence, any primary field \textit{\`a la} Pinamonti's definition
$\Phi: \mathfrak{D}^{(\Delta)} \Rightarrow \mathfrak{A}$
defines a primary field of spin $0$ in our description:
$\widetilde{\Phi}: \mathfrak{D}^{\left( \Delta / 2, \Delta / 2\right)}
    \Rightarrow
\mathfrak{A} \circ \mathfrak{p}$
where $\widetilde{\Phi}_\mathscr{M} := \Phi_\mathcal{M}$.

Conversely, a choice of spin $0$ primary field
$
    \widetilde{\Phi}:
    \mathfrak{D}^{(\Delta/2, \Delta/2)}
        \Rightarrow
    \mathfrak{A} \circ \mathfrak{p}
$
unambiguously defines a natural transformation
$\Phi: \mathfrak{D}^{(\Delta)} \Rightarrow \mathfrak{A}$.
To see this, note that if $\mathscr{M}$ and $\widetilde{\mathscr{M}}$
represent different frames for the same spacetime
$\mathcal{M} = \mathfrak{p}(\mathscr{M}) = \mathfrak{p}(\widetilde{\mathscr{M}})$,
then the identity morphism of the underlying manifold
constitutes a $\mathsf{CFLoc}$ morphism $\mathscr{M} \to \widetilde{\mathscr{M}}$,
hence we can deduce from \eqref{eq:primary} that
$\widetilde{\Phi}_\mathscr{M} \equiv \widetilde{\Phi}_{\widetilde{\mathscr{M}}}$.
In other words, the spin of a primary field measures how it behaves under
a change of frame on a fixed spacetime.
Thus, if the spin vanishes, the primary field does not depend on the frame,
and can be defined in the same way as in
\cite{pinamontiConformalGenerallyCovariant2009}.

\begin{example}
    The null derivative of the scalar field defines a map
    $
        \partial \Phi_\mathscr{M}:
        \mathfrak{D}(\mathcal{M})
            \to
        \mathfrak{F}_{\mu c}(\mathcal{M})
    $
    \begin{equation*}
        \partial \Phi_\mathscr{M}(f)[\phi]
            =
        \int_\mathcal{M}
            f(x) (e_\ell \phi) e^\ell \wedge e^\arr,
    \end{equation*}
    where $e_\ell$ is the vector field dual to $e^\arr$.
    To see that this is a primary field consider the upper-right path through the diagram \eqref{eq:primary}:
    \begin{align*}
        \partial\Phi_\mathscr{N}\left(
                \mathfrak{D}^{(h, \widetilde{h})}\chi (f)
            \right)[\phi]
            &=
        \int_{\chi(M)}
            \left( \chi^{-1} \right)^*
                \left(
                    \omega_\ell^{h - 1}
                    \omega_\arr^{\widetilde{h} - 1}
                    f
                \right)
                \cdot
            (\tilde{e}_\ell \phi)
            \, \tilde{e}^\ell \wedge \tilde{e}^\arr,
            \\
            &=
        \int_M
            \left(
                \omega_\ell^{h - 1}
                \omega_\arr^{\widetilde{h} - 1}
                f
            \right)
                \cdot
            \chi^*(\tilde{e}_\ell \phi) \,
            \left(
                \omega_\ell \, \omega_\arr e^\ell \wedge e^\arr.
            \right)
    \end{align*}
    Next, using $\chi^*(\tilde{e}_\ell \phi) = (\chi^* \tilde{e}_\ell)(\chi^* \phi) = \omega_\ell^{-1} (e_\ell \chi^* \phi)$ we have
    \begin{equation*}
        \partial\Phi_\mathscr{N}\left(
                \mathfrak{D}^{(h, \widetilde{h})}\chi (f)
            \right)[\phi]
            =
        \int_M 
            \omega_\ell(x)^{h - 1}
            \omega_\arr(x)^{\widetilde{h}}
            f(x)
            (e_\ell(\chi^*\phi))
        e^\ell \wedge e^\arr
    \end{equation*}
    To compare this with the lower-left path,
    we first observe that the algebra isomorphisms
    $\alpha_{\chi^*H' - H}$ all act by identity on linear functionals,
    thus if $\mathcal{F}$ is linear,
    ${\mathfrak{A}}\chi(\mathcal{F})[\phi] = \mathcal{F}[\chi^*\phi]$.
    Hence the observable we obtain in this way is
    \begin{equation*}
        {\mathfrak{A}}\chi(\partial\Phi_\mathscr{M}(f))[\phi]
            =
        \int_M f(x) (e_\ell(\chi^* \phi)) e^\ell \wedge e^\arr.
    \end{equation*}
    By fixing $(h, \widetilde{h})$ such that the diagram commutes, we can therefore conclude that $\partial \Phi$ is a primary field of weight $(1, 0)$.
    Similarly, if we consider the field $\bar{\partial}\Phi$,
    obtained by acting with $e_\arr$ instead of $e_\ell$,
    we would obtain a primary field of weight $(0, 1)$.
\end{example}

The introduction of frames also allows us
to implement rigid transformations.
We define the \textit{boost and dilation} morphisms,
for $\alpha \in \mathbb{R} \setminus \{0\}$ as
\begin{align*}
    b_\alpha: (M, (e^\ell, e^\arr)) 
        &\mapsto
    (M, (\tfrac{1}{\alpha} e^\ell, \alpha e^\arr)),
\\
    d_\alpha: (M, (e^\ell, e^\arr)) 
        &\mapsto
    (M, (\alpha e^\ell, \alpha e^\arr)),
\end{align*}
where in each case,
the smooth embedding inducing the morphism is simply $\mathrm{Id}_M$.
If we denote the subcategory generated by these morphisms $\mathsf{CFLoc}_0$,
and the restrictions of
$\mathfrak{D}^{(h, \widetilde{h})}$ and $\mathfrak{A}$ to this subcategory
$\mathfrak{D}^{(h, \widetilde{h})}_0$ and $\mathfrak{A}_0$ respectively,
then a \textit{quasi-primary} field may be defined as a natural transformation
$\mathfrak{D}^{(h, \widetilde{h})}_0 \Rightarrow \mathfrak{A}_0$,
for some pair of weights $(h, \widetilde{h}) \in \mathbb{R}^2$.
In other words, a field is quasi-primary if it responds to
boosts and dilations in the same way a primary field would.

For the massless scalar field,
we identify several notable examples of primary and quasi-primary fields below:
\begin{enumerate}
    \item   As demonstrated in the above example,
            the derivative fields $\partial \Phi$ and
            $\bar{\partial}\Phi$ are both primary.
            Taking higher derivatives will produce
            quasi-primary fields of increasing weight.
            In general $\partial^n\bar{\partial}^m\Phi$ is
            quasi-primary with weight $(n, m)$,
            though note that if both $n$ and $m$ are non-zero,
            this field vanishes on-shell.
    \item   Higher powers of primary fields are again primary classically,
            but in the quantum case,
            they fail to be even quasi-primary in general.
            For instance, the specific case of
            $T = \tfrac{1}{2} (\partial \Phi)^2$
            shall be discussed in the next section.
    \item   The (smeared) vertex operator
            $e^{ia\Phi}_\mathcal{M}(f)$
            defined
            for $f \in \mathfrak{D}(\mathcal{M})$,
            $a \in \mathbb{R}$ by
    \begin{equation*}
        e^{ia\Phi}_\mathcal{M}(f)[\phi]
            :=
        \int_M f(x) e^{ia \phi(x)} \, \mathrm{dVol},
    \end{equation*}
    classically is neither primary nor quasi-primary.
    However, the covariantly normal-ordered field $\nord{e^{ia\Phi}}$
    is a quantum primary with spin $0$ and scaling dimension
    $\tfrac{\hbar a^2}{2 \pi}$

    To see this, consider the lower-left path of \eqref{eq:primary}.
    For
    $f    \in \mathfrak{D}(\mathcal{M})$,
    $\phi \in \mathfrak{E}(\mathcal{N})$,
    $H    \in \mathrm{Had}(\mathcal{M})$, and
    $H'   \in \mathrm{Had}(\mathcal{N})$,
    we have
    \begin{equation}
        \label{eq:vertex_primary_quantum_series}
        \mathfrak{A}\chi \left( \nord{e^{ia\Phi}(f)}_\mathcal{M} \right)_{H'} [\phi]
            =
        \sum_{n = 0}^\infty \left(
            \frac{\hbar}{2} \right)^n \frac{1}{n!}
            \left\langle 
                \left( \chi^* H' - H^\mathrm{sing}_\mathcal{M} \right)^{\otimes n},
                {e^{ia\Phi}_\mathcal{M}}(f)^{(2n)}[\chi^* \phi]
            \right\rangle.
    \end{equation}
    The functional derivatives of
    $e^{ia \Phi}_\mathcal{M}$ can be calculated straightforwardly,
    and yield, for any $n \in \mathbb{N}$
    \begin{align}
        \left\langle 
            \left( \chi^* H' - H^\mathrm{sing}_\mathcal{M} \right)^{\otimes n},
            {e^{ia\Phi}_\mathcal{M}}(f)^{(2n)}[\chi^* \phi]
        \right\rangle
           &= \nonumber \\
        (-a^2)^n \int_M e^{ia\chi^*\phi} f(x)
            \Big(
           &    \lim_{y \to x} \chi^*H'(x; y) - H^\mathrm{sing}_\mathcal{M}(x; y)
            \Big)^n
        \, \mathrm{dVol}
        \nonumber \\
        \label{eq:vertex_primary_order_n}
            =
        \int_M e^{ia\chi^*\phi} f(x)
            \Big(
           &    \!-a^2 \chi^* h'(x; x) +
                 \frac{a^2}{4\pi} \log(\Omega(x))
            \Big)^n
        \, \mathrm{dVol}.
    \end{align}
    Here, $h'$ is the smooth part of $H'$,
    and the $\log(\Omega(x))$ term arises from the difference
    in the local Hadamard form \eqref{eq:local_Hadamard_form} of
    $\chi^* H'$ and $H_\mathcal{M}^\mathrm{sing}$.
    We can then express the action of the morphism $\mathfrak{A}\chi$ as
    \begin{equation}
        \label{eq:vertex_primary_quantum}
        \mathfrak{A}\chi \left( \nord{e^{ia\Phi}(f)}_\mathcal{M} \right)_{H'} [\phi]
            =
        e^{ia\Phi}_\mathcal{M}\left(
            f
            e^{
                \left(
                    -\hbar
                    \frac{a^2}{2}
                    (\iota_\Delta \circ \chi)^* h'
                \right)
                }
            \Omega^{\hbar \frac{a^2}{4 \pi} }
        \right)[\chi^*\phi],
    \end{equation}
    where $\iota_\Delta(x) = (x, x)$,
    and we are using the linearity of $e^{ia\Phi}_\mathcal{M}$ in
    the test function to extend it
    \footnote{
        In doing so, we avoid any necessity to prove summation and integration
        may be interchanged, or that
        $
            \mathrm{Exp}(\hbar( A + B \log C ))
                =
            \mathrm{Exp}(\hbar A) C^{\hbar B}.
        $
        If one is not comfortable with such manipulations of formal series,
        reassurance may be found in the fact that,
        if the field configuration $\phi$ is held fixed,
        and $\hbar$ is chosen to be any positive number, then the series
        \eqref{eq:vertex_primary_quantum_series} converges absolutely,
        as a series of complex numbers,
        to the right hand side of \eqref{eq:vertex_primary_quantum}.
    }
    to a map
    $
        \mathfrak{D}(\mathcal{M})[[\hbar]]
            \to
        \mathfrak{F}_{\mu c}(\mathcal{M})[[\hbar]].
    $%
    
    We can compare this to $\nord{e^{ia\Phi}_\mathcal{N}}$,
    where we have, for $g \in \mathfrak{D}(\mathcal{N})$
    \begin{align*}
        \nord{e^{ia\Phi}_\mathcal{N}(g)}_{H'} [\phi]
            &=
        \sum_{n = 0}^\infty \left(
            \frac{\hbar}{2} \right)^n \frac{1}{n!}
            \left\langle 
                {h'}^{\otimes n},
                {e^{ia\Phi}_\mathcal{N}}(g)^{(2n)}[\phi]
            \right\rangle,
            \\
            &=
        \left\langle 
            e^{-\hbar \frac{a^2}{2}h'_\Delta}, e^{ia\Phi}_\mathcal{N}(g)
        \right\rangle,
            \\
            &=
        e^{ia\Phi}_\mathcal{N}\left( g e^{-\hbar \frac{a^2}{2} h'_\Delta} \right)[\phi].
    \end{align*}
    As $e^{ia\Phi}$ is a classical primary field of scaling dimension $0$, we have $e^{ia\Phi}_\mathcal{M}(f)[\chi^* \phi] = e^{ia \Phi}_\mathcal{N}(\chi_* \Omega^{-d} f)[\phi]$, hence
    \begin{equation*}
        \mathfrak{A}\chi \left( \nord{e^{ia\Phi}_\mathcal{M}(f)} \right)_{H'} [\phi]
            =
        \nord{e^{ia\Phi}_\mathcal{N}\left( \mathfrak{D}^{\left( \frac{\hbar a^2}{4 \pi} \right)} (f) \right)}_{H'}[\phi]
    \end{equation*}
    as required.
\end{enumerate}

\subsection{The Stress-Energy Tensor of the Massless Scalar Field}

A well known feature of chiral CFTs is the transformation law for the stress-energy tensor,
constrained by the famous L\"uscher-Mack theorem
\cite{luscherGlobalConformalInvariance1975}.
Here we shall show explicitly that,
for the free scalar field in $2$D Minkowski space,
the stress-energy tensor satisfies precisely this transformation law.
And, moreover, that there exist analogous transformation laws on
arbitrary globally-hyperbolic spacetimes.

The $uu$ component of the stress-energy tensor%
\footnote{
    We may also refer to $T_{uu}$ as the \textit{chiral} component of $T$,
    in which case $T_{vv}$ would be the \textit{anti-chiral} component.
    For ease of notation, we consider only the chiral component,
    dropping the subscript.
}
on a framed spacetime $\mathscr{M} = (M, (e^\ell, e^\arr)) \in \mathsf{CFLoc}$,
is a distribution valued in $\mathfrak{F}_\mathrm{loc}(M)$
defined, for $f \in \mathfrak{D}(M)$, $\phi \in \mathfrak{E}(M)$ by
\begin{equation}
    T_\mathscr{M}(f)[\phi]
        :=
    \frac{1}{2} \int_M f \cdot
        (e_\ell \phi)^2 e^\ell \wedge e^\arr.
\end{equation}
Note that we can replace the test function $f$ with a compactly supported distribution,
so long as its singularity structure is compatible with the constraint
that $T_\mathscr{M}(f)$ is a microcausal distribution.
In particular, the generators of the Virasoro algebra $B_n$
from \cref{sec:Virasoro} can be expressed as
$T_{\mathscr{M}_\mathrm{cyl}}(f_n)$,
where the integral kernel of $f_n$ is $e^{inu} \delta(u + v)$
in the null-coordinates for the cylinder.

Clasically, $T$ is a primary field with conformal weight $(2, 0)$,
i.e. $T: \mathfrak{D}^{(2, 0)} \Rightarrow \mathfrak{P} \circ \mathfrak{p}$,
where $\mathfrak{P}$ is the classical theory for
the massless scalar field, as given in \cref{def:LCCCFT}.
We shall now study how its quantisation $\nord{T}$ fails to be even quasi-primary.

In order to make our analysis more concrete,
we restrict our attention to the subcategory of $\mathsf{CFLoc}$
containing the single object $\mathbb{M}_2$.
Here, the locally covariant normal ordering prescription
$\nord{-}_{\mathbb{M}_2}$ is simply
$\altnord{-}_{H_{\mathbb{M}}}$,
where $H_{\mathbb{M}}$ is the symmetric part of the Minkowski vacuum.
Hence, if we work in the concrete algebra $\mathfrak{A}^{H_\mathbb{M}}(\mathbb{M})$
we can identify $T_{\mathbb{M}_2}(f)$ directly with its quantum counterpart
with no modification.

Given a $\mathsf{CFLoc}$ morphism $\chi: \mathbb{M}_2 \to \mathbb{M}_2$,
if the covariantly ordered field $\nord{T}$ was primary,
we would expect in particular that
$
    \mathfrak{A}\chi\left( \nord{T}_{\mathbb{M}_2}(f) \right)
        -
    \nord{T}_{\mathbb{M}_2}\left( \mathfrak{D}^{(2,0)}\chi f \right)
$
would vanish.
Upon making the identification
$\mathfrak{A}(\mathbb{M}_2) \simeq \mathfrak{A}^{H_\mathbb{M}}(\mathbb{M}_2)$
this term becomes
\begin{equation}
    \alpha_{\chi^*H_\mathbb{M} - H_\mathbb{M}}\left( T_{\mathbb{M}_2}(f) \right)
        -
    T_{\mathbb{M}_2}\left( \mathfrak{D}^{(2,0)}\chi f \right).
\end{equation}

We already know that this vanishes in the classical limit $\hbar \to 0$,
hence we only need to compute the $\mathcal{O}(\hbar)$ term.
Recall that in null coordinates we can express a $\mathsf{CFLoc}$
morphism $\mathbb{M}_2 \to \mathbb{M}_2$ using a pair of functions
$\mu, \nu \in \mathrm{Diff}_+(\mathbb{R})$ by
$\chi(u, v) = (\mu(u), \nu(v))$.
Upon doing so we see
\begin{align}
    \begin{split}
    \left\langle 
        (\chi^*H_{\mathbb{M}_2} - H_{\mathbb{M}_2}), T_{\mathbb{M}_2}(f)^{(2)}
    \right\rangle
        &= \\
    \int_{\mathbb{R}^2} \partial_u \partial_{u'} \big[ 
        &H_{\mathbb{M}_2}(\mu(u); \mu(u')) -
        H_{\mathbb{M}_2}(u; u')
    \big]
    f(u) \delta(u - u') \, \mathrm{d}u  \mathrm{d}u',
    \end{split}
\end{align}
where we have integrated out $v$ and $v'$ and defined
$f(u) := \int_\mathbb{R}f(u, v) \, \mathrm{d}v$.
It only remains to determine
\begin{align}
    \begin{split}
            \lim_{u' \to u} \big[ \mu'(u)\mu'(u')
                (H_{\mathbb{M}_2})_{uu'}(\mu(u); \mu(u')) -
                (H_{\mathbb{M}_2})_{uu'}(u; u')
            \big]
                = \\
            \lim_{u' \to u} \left[ 
                \frac{\mu(u)\mu(u')}{(\mu(u) - \mu(u'))^2} - 
                \frac{1}{(u - u')^2}
            \right].
    \end{split}
\end{align}
By Taylor expanding $\mu(u')$ around $u$, one eventually finds that the limit exists and is equal to
\begin{equation}
    \frac{1}{6}\left( 
        \frac{\mu'''(u)}{\mu'(u)} - \frac{3}{2}\left( \frac{\mu''(u)}{\mu'(u)} \right)^2
    \right)
        =:
    \frac{1}{6} S(\mu)(u),
\end{equation}
where $S(\mu)$ denotes the Schwarzian derivative of the function $\mu$.
From this it is clear that $\nord{T}$ is not primary, as
\begin{equation}
    \mathfrak{A}\chi\left( 
        \nord{T}_{\mathbb{M}_2}(f) \right)
        =
    \nord{T}_{\mathbb{M}_2}\left( \mathfrak{D}^{(2, 0)}\chi(f) \right)
        - \frac{1}{4 \pi} \frac{\hbar}{12} \left\langle S(\mu), f \right\rangle.
\end{equation}

Thus we recover the well-known result that, on Minkowski spacetime,
the quantum stress-energy tensor transforms almost as a primary of weight $(2, 0)$,
but is obstructed by an $\mathcal{O}(\hbar)$ correction proportional to
the Schwarzian derivative of the transformation.
We can now use our framework to generalise this result to any globally hyperbolic spacetime.
The failure for \eqref{eq:primary} to commute for
$\chi \in \mathrm{Hom}_{\mathsf{CFLoc}}(\mathscr{M}; \mathscr{N})$ is
\begin{equation}
    \label{eq:generic_schwarzian}
    \left\langle \widetilde{S}(\chi), f \right\rangle
        =
    \mathfrak{A} \chi \left( \nord{T}_\mathscr{M}(f) \right)
        -
    \nord{T}_\mathscr{N}\left( \mathfrak{D}^{(2, 0)}\chi(f) \right).
\end{equation}
Whilst the right hand side of this equation requires an arbitrary choice of
$H'  \in \mathrm{Had}(\mathcal{N})$ and
$\phi \in \mathfrak{E}(\mathcal{N})$,
$\widetilde{S}$ is actually independent of both of these choices.
As in Minkowski space, the classical term cancels and we are left to compute
\begin{equation*}
    \left\langle \widetilde{S}(\chi), f \right\rangle
        =
    \frac{\hbar}{2}\left[ 
        \left\langle
            \chi^*H' - H^\mathrm{sing}_\mathcal{M},
            T_\mathscr{M}(f)^{(2)}
        \right\rangle
            -
        \left\langle
            H' - H^\mathrm{sing}_\mathcal{N},
            T_\mathscr{N}\left( \mathfrak{D}^{(2, 0)}\chi (f) \right)^{(2)}
        \right\rangle
        \right],
\end{equation*}
where the choice of configuration $\phi$ has been suppressed as no remaining terms depend on it.
If we define $h' = H' - H^\mathrm{sing}_\mathcal{N}$,
then one can show that
$
    \big\langle
        h',
        T_\mathscr{N}\left( \mathfrak{D}^{(2, 0)}\chi (f) \right)^{(2)}
    \big\rangle
        =
    \left\langle
        \chi^*h', T_\mathscr{M}(f)^{(2)}
    \right\rangle,
$
which cancels with the smooth part of $\chi^*H'$,
and hence
\begin{equation}
    \label{eq:2D_Schwarzian}
    \widetilde{S}(\chi)
        =
    \frac{\hbar}{2} \iota_{\Delta}^*
    \left(
        (e_\ell \otimes e_\ell)
        \left( 
            \chi^* H^\mathrm{sing}_\mathcal{N} -
                   H^\mathrm{sing}_\mathcal{M}
        \right)
    \right),
\end{equation}
where we are again using the embedding
$\iota_\Delta: x \mapsto (x, x) \in \mathcal{M}^2$.
If we take $\chi: {\mathbb{M}_2} \to {\mathbb{M}_2}$ to be as above, we then see that $\widetilde{S}(\chi) = S(\mu)$, hence the original Schwarzian derivative is recovered.

Note that the right-hand side of \eqref{eq:generic_schwarzian}
can be defined for \textit{any} confomally covariant QFT.
A \textit{L\"uscher-Mack theorem} for pAQFT would then imply that,
as a distribution, this is equal to \eqref{eq:2D_Schwarzian}
up to multiplication by some constant,
which we could then interpret as the central charge of the theory.
We stress that such a result has not yet been found,
however we intend to return to this issue in future work.
\todo[color=todogreen, ]{
    Yeah, I didn't know how to end that last sentence.
    Is it even a good idea to say that an analoge LM theorem doesn't exist yet?
}

\section{Conclusion and Outlook}
In this paper we have shown how CFT fits into the framework of pAQFT. As an example application, we have proposed a fully Lorentzian treatment of 1+1 massless scalar field on the Minkowski cylinder and we have shown how the covariant choice of normal ordering of observables leads to correct commutation relations for Virasoro generators. We have also shown that a change of normal ordering leads to the appearance of an extra term $\zeta(-1)$, which is usually explained using the zeta regularisation trick. Here we derive this result completely rigorously, using the pAQFT framework.

In our future work we aim to study further how chiral algebras emerge naturally in our framework and how our approach relates to the standard AQFT treatment (local conformal nets) and the factorisation algebras approach \cite{costelloFactorizationAlgebrasQuantum2016}. We also plan to study OPEs and interacting theories.

\section{Acknowledgements}
We would like to thank Sebastiano Carpi, Chris Fewster and Robin Hillier for very inspiring discussions.

\appendix

\section{Method of Images}
\label{sec:images}

It is well-known that if a space $Y$ can be expressed as
the quotient of some other space $X$ under the action of some group
(satisfying certain properties),
then we can use this relation in order to build Green's functions on $Y$
out of Green's functions.
Here we give a coordinate-free account of some of the necessary results,
then explain how this method may be used to construct the
retarded/advanced propagators of the cylinder from those of Minkowski space.

\begin{lemma}
    \label{lem:extension_of_propagator}
    Let $P$ be a differential operator on a smooth manifold $\mathcal{M}$ and
    let $G: \mathfrak{D}(\mathcal{M}) \to \mathfrak{E}(\mathcal{M})$ be
    a fundamental solution to $P$, i.e. $P G f = G P f = f$ for all
    $f \in \mathfrak{D}(\mathcal{M})$.
    For $U \subset \mathcal{M}$ open, define
    \begin{equation}
        \mathcal{M} \setminus \mathrm{supp}_U \, G
            =
        \bigcup
        \left\{
            V \subset \mathcal{M} \text{ open}
                \, | \,
            \mathrm{supp} \, f \subset V
                \Rightarrow
            (Gf)|_U \equiv 0
        \right\}.
    \end{equation}
    Let $\phi \in \mathfrak{E}(\mathcal{M})$,
    if there exists an open cover
    $\bigcup_{\alpha \in \mathcal{A}} U_\alpha = \mathcal{M}$ such that
    $\mathrm{supp} \, \phi \cap \mathrm{supp}_{U_\alpha} \, G$ is compact,
    then one can define a function $G \phi \in \mathfrak{E}(\mathcal{M})$ such that
    $P G \phi = G P \phi = \phi$.
\end{lemma}

\begin{proof}
    We claim that the local definitions
    \begin{equation*}
        G\phi|_{U_\alpha} := G(\rho_\alpha \phi)|_{U_\alpha},
    \end{equation*}
    where $\rho_\alpha \in \mathfrak{D}(\mathcal{M})$ such that
    $\rho_\alpha \equiv 1$ on $\mathrm{supp} \, \phi \cap \mathrm{supp}_{U_\alpha} \, G$
    can be glued together to form the desired map.
    Suppose $\alpha, \beta \in \mathcal{A}$ such that
    $U_{\alpha\beta} = U_\alpha \cap U_\beta \neq \emptyset$.
    One can quickly verify that
    $
        \mathrm{supp}_{U_{\alpha\beta}} \, G
            \subseteq
        \mathrm{supp}_{U_{\alpha}} \, G
            \cap
        \mathrm{supp}_{U_{\beta}} \, G,
    $
    hence $\rho_\alpha \phi |_{U_{\alpha\beta}} = \rho_\beta \phi |_{U_{\alpha\beta}}$.
    In particular this means that
    $
        \mathrm{supp} \, ((\rho_\alpha - \rho_\beta) \phi)
            \subset
        \mathcal{M} \setminus \mathrm{supp}_{U_{\alpha\beta}} \, G
    $
    and hence
    $G(\rho_\alpha \phi)|_{U_{\alpha\beta}} = G(\rho_\beta \phi)|_{U_{\alpha\beta}}$,
    thus $G\phi$ is a well-defined function.

    Next, to show that $PG\phi = GP\phi = \phi$,
    note that for every $x \in \mathcal{M}$ there must be a neighbourhood
    $U' \ni x$ such that $U' \subset \mathrm{supp}_{U'} \, G$,
    otherwise we could not have that $PGf = GPf = f$ even for
    $f \in \mathfrak{D}(\mathcal{M})$.
    As such, we may assume that the cover
    $\left\{ U_\alpha \right\}_{\alpha \in \mathcal{A}}$
    satisfies $U_\alpha \subset \mathrm{supp}_{U_\alpha} \, G$ for every $\alpha$.
    We then use the locality of differential operators, namely that
    $(P \psi)|_U = P|_{\mathfrak{E}(U)} \psi|_U$
    for any $\psi \in \mathfrak{E}(\mathcal{M})$,
    to see that
    $(PG\phi)|_{U_\alpha} = (\rho_\alpha \phi)|_{U_\alpha}$.
    As we have assumed $U_\alpha \subset \mathrm{supp}_{U_\alpha} \, G $,
    for any $x \in U_\alpha$ we must either have
    $x \in \mathrm{supp} \, \phi \cap \mathrm{supp}_{U_\alpha} \, G$,
    in which case $\rho_\alpha(x) = 1$
    or $\phi(x) = 0$.
    In both cases, we have $\rho_\alpha(x) \phi(x) = \phi(x)$,
    hence $(PG\phi)|_{U_\alpha} = \phi|_{U_\alpha}$.
    For the same reasons, we have that
    $(\rho_\alpha(P\phi))|_{U_\alpha} = (P(\rho_\alpha \phi))|_{U_\alpha}$
    and hence
    $(GP\phi)|_{U_\alpha} = \phi|_{U_\alpha}$ concluding the proof.
\end{proof}

\begin{theorem}[The Method of Images]
    \label{thm:method_of_images}
    Let $\pi: \widetilde{\mathcal{M}} \to \mathcal{M}$ be
    a regular covering of $\mathcal{M}$ by $\widetilde{\mathcal{M}}$.
    Further, let $P$ and $\widetilde{P}$ be a pair of differential operators for
    $\mathcal{M}$ and $\widetilde{\mathcal{M}}$ respectively, such that
    $\pi^* P = \widetilde{P} \pi^*$.
    Further, let $\widetilde{G}$ be a fundamental solution to $\widetilde{P}$ such that
    \begin{enumerate}
        \item   \label{itm:finite_intersection}
                There exists a covering
                $\bigcup_{\alpha \in \mathcal{A}} U_\alpha = \widetilde{\mathcal{M}}$
                such that, $\forall K \subset \mathcal{M}$ compact,
                $\pi^{-1}(K) \cap \mathrm{supp}_{U_\alpha} \, \widetilde{G}$ is compact,
        \item   \label{itm:deck_equivariance}
                $\forall \rho \in \mathrm{Aut}(\pi)$,
                $\rho^* \widetilde{G} = \widetilde{G} \rho^*$.
    \end{enumerate}
    Then there exists a fundamental solution $G$ for $P$ such that
    $\pi^* G = \widetilde{G} \pi^*$
\end{theorem}

\begin{proof}
    Because $\mathrm{supp} \,  \pi^* f = \pi^{-1}(\mathrm{supp} \, f)$,
    condition \ref{itm:finite_intersection} tells us that
    $\widetilde{G} \pi^* f$ is well defined and satisfies
    $
        \widetilde{P} \widetilde{G} \pi^* f
            =
        \widetilde{G} \widetilde{P} \pi^* f
            = 
        \pi^* f
    $

    Next, \ref{itm:deck_equivariance} ensures that for any $\rho \in \mathrm{Aut}(\pi)$
    \begin{equation}
        \rho^* \widetilde{G} \pi^* f
            =
        \widetilde{G} \rho^* \pi^* f
            =
        \widetilde{G} (\pi \circ \rho)^* f
            =
        \widetilde{G} \pi^* f,
    \end{equation}
    i.e. $\widetilde{G} \pi^* f$ is a $\mathrm{Aut}(\pi)$ invariant,
    and hence can be expressed as $\pi^* F$ for some $F \in \mathfrak{E}(\mathcal{M})$.
    As our choice of $f$ was arbitrary, this defines a map $f \mapsto F$,
    which is clearly linear.
    As such we denote it $G: \mathfrak{D}(\mathcal{M}) \to \mathfrak{E}(\mathcal{M})$.

    To show that $G$ is then a fundamental solution for $P$ is a
    fairly mechanical process:
    \begin{equation}
        \pi^* P G f
            =
        \widetilde{P} \pi^* G f
            =
        \widetilde{P} \widetilde{G} \pi^* f
            =
        \pi^* f.
    \end{equation}
    From the injectivity of $\pi^*$, we may then conclude $P G f = f$.
    Next, using the same trick
    \begin{equation}
        \pi^* G P f
            =
        \widetilde{G} \pi^* P f
            =
        \widetilde{G} \widetilde{P} \pi^* f
            =
        \pi^* f,
    \end{equation}
    which again shows $G P f = f$.
\end{proof}

The following lemma shows how this applies to the equations of motion
of a locally covariant (classical) field theory.

\begin{lemma}
    Let $\mathcal{L}: \mathfrak{D} \Rightarrow \mathfrak{F}_\mathrm{loc}$
    be a natural Lagrangian such that,
    for any $\mathcal{M} \in \mathsf{Loc}$, $\phi \in \mathfrak{E}(\mathcal{M})$,
    $
        \left\langle S_\mathcal{M}''[\phi], h \otimes g \right\rangle
            =
        \left\langle P_\mathcal{M}[\phi] h, g \right\rangle
    $
    where $P_\mathcal{M}[\phi]$ is some differential operator.
    If $\widetilde{\mathcal{M}}, \mathcal{M} \in \mathsf{Loc}$ and
    $\pi: \widetilde{\mathcal{M}} \to \mathcal{M}$ is such that for every
    $x \in \widetilde{\mathcal{M}}$,
    there exists a subspacetime\footnote{
        i.e. the inclusion $\mathcal{N} \hookrightarrow \widetilde{\mathcal{M}}$
        is an admissible embedding of spacetimes
    }
    $\mathcal{N} \ni x$ such that
    $\pi|_\mathcal{N}$ is an admissible embedding, then
    \begin{equation}
        \pi^* P_\mathcal{M}[\phi] = P_{\widetilde{\mathcal{M}}}[\pi^* \phi] \pi^*.
    \end{equation}
\end{lemma}

\begin{proof}
    Recall that the naturality of $\mathcal{L}$ implies that,
    for every admissible embedding $\chi: \mathcal{M} \hookrightarrow \mathcal{N}$,
    $\chi^* P_\mathcal{N}[\phi] = P_{\mathcal{M}}[\chi^* \phi] \chi^*$.
    Applying this to
    and the composed map
    $\pi|_\mathcal{N} = \pi \circ \iota$
    and then to the inclusion
    $\iota: \mathcal{N} \hookrightarrow \mathcal{M}$, we have,
    for $\phi \in \mathfrak{E}(\mathcal{M})$ and $g \in \mathfrak{D}(\mathcal{M})$
    \begin{align*}
        (\pi^* (P_\mathcal{M}[\phi] g))|_\mathcal{N}
            &=
        P_\mathcal{N}[(\pi^* \phi)|_\mathcal{N}] (\pi^* g)|_\mathcal{N}\\
            &=
        (P_{\widetilde{\mathcal{M}}}[\pi^* \phi] \, \pi^* g)|_{\mathcal{N}}.
    \end{align*}
    Given that $\widetilde{\mathcal{M}}$ is covered by
    $\mathcal{N} \subseteq \widetilde{\mathcal{M}}$
    for which this holds,
    we may conclude
    $\pi^*(P_\mathcal{M}[\phi] \, g) = P_{\widetilde{\mathcal{M}}}[\pi^* \phi] \pi^* g$
    as desired.
\end{proof}

Given that the equations of motion are related in this way,
we can now show that the propagators are as well:
For any $f \in \mathfrak{D}(\mathcal{M}_\mathrm{cyl})$,
$\mathrm{supp} \, \pi^* f$ is clearly timelike compact,
i.e. there exists a pair of Cauchy surfaces $\Sigma_\pm \in {\mathbb{M}_2}$ such that
$
    \mathrm{supp} \, \pi^* f
        \subseteq
    \mathscr{J}^+(\Sigma_-) \cap \mathscr{J}^-(\Sigma_+).
$
From this it follows that $\mathrm{supp} \, \pi^* f$
is both past-compact and future-compact.
The support properties \eqref{eq:propagator_supports} of $E^{R/A}$ imply that
$\mathrm{supp}_U \, E^{R/A} = \mathcal{J}^{\mp}(\overline{U})$,
where $\overline{U}$ is the closure of $U$.

Next, as the symmetries of the covering map
$(x, t) \mapsto (x + 2\pi n, t)$ are translations,
and $E^{R/A}$ are both equivariant under translations,
we have also satisfied condition \ref{itm:deck_equivariance}.
Applying \Cref{thm:method_of_images}, we thus have a pair of propagators
$
    E^{R/A}_\mathrm{cyl}:
        \mathfrak{D}(\mathcal{M}_{\mathrm{cyl}})
            \to
        \mathfrak{E}(\mathcal{M}_{\mathrm{cyl}})
$
which satisfy
\begin{equation}
    \pi^* E^{R/A}_\mathrm{cyl} = E^{R/A} \pi^*.
\end{equation}
It is straightforward to verify that these satisfy the support criteria
\eqref{eq:propagator_supports},
hence they are \textit{the} retarded/advanced propagators for the cylinder.

\section{Closure Proofs for Microcausal Functionals}
\label{sec:closure_proofs}

\begin{prop}
    \label{prop:peierls_closure}
    Let $\mathcal{M}$ be a globally hyperbolic spacetime,
    let $S$ be a quadratic action on $\mathcal{M}$,
    then
    $
        \{ \cdot, \cdot \}_S:
        \mathfrak{F}_{\mu c}(\mathcal{M}) \times 
        \mathfrak{F}_{\mu c}(\mathcal{M}) \to
        \mathfrak{F}_{\mu c}(\mathcal{M})
    $.
\end{prop}

\begin{proof}
    We shall only prove this fact for $\mathcal{M} \subseteq \mathbb{R}^d$,
    but it is possible to `patch together'
    the results over an atlas for a more general $\mathcal{M}$.
    We begin by rephrasing Theorem 8.2.13 of
    \cite{hormanderAnalysisLinearPartial2015}:

    Suppose that
    $X \subseteq \mathbb{R}^n$, and
    $Y \subseteq \mathbb{R}^m$.
    Let $K \in \mathfrak{D}'(X \times Y)$
    and $u \in \mathfrak{E}'(Y)$.
    Theorem 8.2.13 allows us to define a new distribution $K \circ u$,
    with integral kernel
    \begin{equation}
        (K \circ u)(x)
            =
        \int_Y
            K(x, y) u(y)
        \, \mathrm{d}y,
    \end{equation}
    and estimate its wavefront set.
    Namely, $K \circ u$ exists whenever
    $\mathrm{WF}'(K)_Y \cap \mathrm{WF}(u) = \emptyset$, where
    \begin{equation*}
        \mathrm{WF}'(K)_Y
            :=
        \left\{
            (y; \eta) \in T^*Y \setminus \underline{0}_Y
                \, | \,
            \exists x \in X,
            (x, y; 0, -\eta) \in \mathrm{WF}(K)
        \right\},
    \end{equation*}
    is the wavefront set of $K$ \textit{twisted w.r.t.} $Y$
    (and $\underline{0}_Y$ denotes the zero section of $T^*Y$).

    Moreover, whenever $K \circ u$ does exist, we have
    \begin{equation}
        \label{eq:hormander_composition_estimate}
        \mathrm{WF}(K \circ u)
            \subseteq
        \left\{
            (x, \xi) \in T^*X
                \, | \,
            \exists (y, \eta) \in \mathrm{WF}(u) \cup \underline{0}_Y,
            (x, y; \xi, \eta) \in \mathrm{WF}(K)
        \right\}
    \end{equation}

    Let $\mathcal{F}, \mathcal{G} \in \mathfrak{F}_{\mu c}(\mathcal{M})$,
    the $m^\text{th}$ functional derivative of their Peierls bracket can be written,
    ommitting the dependence on a field configuration
    $\phi \in \mathfrak{E}(\mathcal{M})$,
    as follows:
    \begin{equation}
        \label{eq:peierls_bracket_derivative}
        \left( \left\{ \mathcal{F}, \mathcal{G}\right\}_S \right)^{(m)}
            =
        \sum_{\left\{ J_1, J_2 \right\} \in P_m}
            \left[
                \left(
                    \mathcal{F}^{(|J_1| + 1)} \otimes
                    \mathcal{G}^{(|J_2| + 1)}
                \right)
                \circ E
            \right]
                s_{J_1, J_2},
    \end{equation}
    where the sum runs over partitions $J_1 \sqcup J_2 = \{1, \ldots, m\}$,
    $\circ$ is the operation described above, and
    $
        s_{J_1, J_2}:
            \mathfrak{D}(\mathcal{M}^m) \to
            \mathfrak{D}(\mathcal{M}^m)
    $
    is an operation permuting the variables of a given test function
    according to a permutation $\sigma_{J_1, J_2} \in S_m$ such that
    $i \in J_1 \Rightarrow \sigma_{J_1, J_2}(i) \leq |J_1|$.
    (As $\mathcal{F}^{(m)}$ is permutation invariant as a distribution,
    this is a sufficient characterisation of $\sigma_{J_1, J_2}$.)
    In fact, as we are only testing for microcausality,
    the only property we need of these distributions is that, for $0 \leq k \leq m$,
    the wavefront set of
    $
        \left( \mathcal{F}^{(k + 1)} \otimes \mathcal{G}^{(m - k + 1)} \right)
            \circ
        E
    $
    is disjoint from
    the cones $\overline{V}_\pm^m$, defined by
    \begin{equation}
        \label{eq:lightcones}
        \overline{V}_+^m
            =
        \left\{ 
            (x_1, \ldots, x_m; \xi_1, \ldots, \xi_m) \in T^*\mathcal{M}
                \, | \,
            \xi_i \in \overline{V}_+(x_i) \forall i \leq m
        \right\},
    \end{equation}
    where $\overline{V}_+(x)$ denotes
    the closed future/past lightcone in $T^*_x \mathcal{M}$,
    and similar for $\overline{V}_-^m$.

    We set $X = \mathcal{M}^n$, $Y = \mathcal{M}^2$,
    $K = \mathcal{F}^{(k + 1)} \otimes \mathcal{G}^{(m - k + 1)}$, and $u = E$.
    Using \cite[Theorem 8.2.9]{hormanderAnalysisLinearPartial2015},
    we can estimate
    $\mathrm{WF}(\mathcal{F}^{(k + 1)} \otimes \mathcal{G}^{(m - k + 1)})$ by
    \begin{equation}
        \label{eq:tensor_prod_wavefront_estimate}
        \mathrm{WF}\left(
            \mathcal{F}^{(k + 1)} \otimes
            \mathcal{G}^{(m - k + 1)}
        \right)
            \subseteq
        \left( 
            \mathrm{WF}(\mathcal{F}^{(k + 1)})
                \cup
            0_{\mathcal{M}^{k + 1}}
        \right) \times
        \left(
            \mathrm{WF}(\mathcal{G}^{(m - k + 1)})
                \cup
            0_{\mathcal{M}^{m - k + 1}}
        \right),
    \end{equation}
    where $0_\mathcal{M} = \mathcal{M} \times \{0\} \subseteq T^*\mathcal{M}$
    denotes the zero section of $T^*\mathcal{M}$ \textit{etc}.
    Let
    $
        (y_\mathcal{F}, y_\mathcal{G};
        \eta_\mathcal{F}, \eta_\mathcal{G}) \in
        T^*Y \setminus 0_Y.
    $

    The wavefront set of the causal propagator,
    as may be found in \cite[\S4.4.1]{rejznerPerturbativeAlgebraicQuantum2016},
    can be written as
    \begin{equation}
        \mathrm{WF}(E)
            =
        \left\{
            (x, y; \xi, \eta) \in T^* \mathcal{M}^2
                \,|\,
            (x, \xi) \in \overline{V}_+ \cup \overline{V}_-,
            (x, \xi) \sim (y, -\eta)
        \right\},
    \end{equation}
    where the relation $(x, \xi) \sim (y, \eta)$ means there exists a null geodesic
    $\gamma: [0, 1] \to \mathcal{M}$
    connecting $x$ to $y$ and such that the parallel transport of $\xi$
    along $\gamma$ is $\eta$.
    However, for our purposes, we can use the much simpler estimate
    \begin{equation}
        \mathrm{WF}(E)
            \subset
        (V_+ \times V_-)
            \cup
        (V_- \times V_+),
    \end{equation}
    i.e. if $(x, y; \xi, \eta) \in \mathrm{WF}(E)$ then either
    $(x, \xi) \in V_+$ and $(y, \eta) \in V_-$,
    \textit{or}
    $(x, \xi) \in V_-$ and $(y, \eta) \in V_+$.

    Suppose there exists
    $\underline{x}_\mathcal{F} \in \mathcal{M}^k$ and
    $\underline{x}_\mathcal{G} \in \mathcal{M}^{n - k}$ such that
    \begin{equation*}
        (
            \underline{x}_\mathcal{F},     y_\mathcal{F},
            \underline{x}_\mathcal{G},     y_\mathcal{G};
            0            , -\eta_\mathcal{F},
            0            , -\eta_\mathcal{G}
        ) \in
        \mathrm{WF}'\left(
            \mathcal{F}^{(k + 1)} \otimes
            \mathcal{G}^{(m - k + 1)}
        \right)_Y,
    \end{equation*}
    then this estimate indicates that either $\eta_\mathcal{F} = 0$,
    or $(y_\mathcal{F}; \eta_\mathcal{F}) \notin \overline{V}_\pm$.
    The same is also true of $(y_\mathcal{G}; \eta_\mathcal{G})$,
    though at least one of $\eta_\mathcal{F}$ and $\eta_\mathcal{G}$
    must be non-zero.
    Thus we see that the intersection of
    $\mathrm{WF}(\mathcal{F}^{(k + 1)} \otimes \mathcal{G}^{(m - k + 1)})$
    with $\mathrm{WF}(E)$ must be trivial, as
    $
        (y_\mathcal{F}, y_\mathcal{G}; \eta_\mathcal{F}, \eta_\mathcal{G}) \in 
        \mathrm{WF}(E) \Rightarrow
        (y_\mathcal{F}; \eta_\mathcal{F}),
        (y_\mathcal{G}; \eta_\mathcal{G}) \in
        (\overline{V}_+ \cup \overline{V}_-) \setminus 0_\mathcal{M}.
    $

    Thus we can apply theorem 8.2.13 and conclude not only that
    $(\mathcal{F}^{(k + 1)} \otimes \mathcal{G}^{(m - k + 1)}) \circ E$
    is well defined, but also that its wavefront set has trivial intersection with both
    $\overline{V}_+^m$ and $\overline{V}_-^m$.
    To see this, let
    $
        (
                \underline{x}_\mathcal{F}, \underline{x}_\mathcal{G};
                \underline{\xi}_\mathcal{F}, \underline{\xi}_\mathcal{G}
        ) \in
        \overline{V}_+^m.
    $
    Any
    $
        (y_\mathcal{F}, y_\mathcal{G}; \eta_\mathcal{F}, \eta_\mathcal{G}) \in
        \mathrm{WF}(E) \cup 0_Y
    $
    necessarily belongs also to either $ \overline{V}_+ \times \overline{V}_- $
    or $ \overline{V}_- \times \overline{V}_+ $.
    Suppose it is the former, then, by microcausality,
    $
        (\underline{x}_\mathcal{G}, y_\mathcal{G};
         \underline{\xi}_\mathcal{G}, -\eta_\mathcal{G}) \notin
        \mathrm{WF}(\mathcal{G}^{(m - k + 1)}).
    $
    Recalling \eqref{eq:tensor_prod_wavefront_estimate},
    this means there is only a chance that
    $
        (
            \underline{x}_\mathcal{F}, y_\mathcal{F},
            \underline{x}_\mathcal{G}, y_\mathcal{G};
            \underline{\xi}_\mathcal{F}, \eta_\mathcal{F},
            \underline{\xi}_\mathcal{G}, \eta_\mathcal{G}
        ) \in
        \mathrm{WF}(\mathcal{F}^{(k + 1)} \otimes \mathcal{G}^{(m - k + 1)})
    $
    if $\underline{\xi}_\mathcal{G}$ and $\eta_\mathcal{G}$ are both zero.
    However, this still fails, as $\eta_\mathcal{G} = 0 \Rightarrow \eta_\mathcal{F} = 0$,
    which in turn implies that
    $
        (\underline{x}_\mathcal{F}, y_\mathcal{F};
         \underline{\xi}_\mathcal{F}, -\eta_\mathcal{F}) \notin
        \mathrm{WF}(\mathcal{F}^{(k + 1)}).
    $
    The wavefront set estimate from 8.2.13 then allows us to conclude that
    $
        (
                \underline{x}_\mathcal{F}, \underline{x}_\mathcal{G};
                \underline{\xi}_\mathcal{F}, \underline{\xi}_\mathcal{G}
        ) \notin
        \mathrm{WF}\left( 
            (\mathcal{F}^{(k + 1)} \otimes \mathcal{G}^{(m - k + 1)})
            \circ E
        \right).
    $
    Applying the corresponding argument to $\Gamma_-^m$,
    we see that all derivatives of $\left\{ \mathcal{F}, \mathcal{G}\right\}_S$
    satisfy the requisite wavefront set condition to be declared microcausal.
\end{proof}

\begin{prop}
    \label{prop:quantum_closure}
    Let $\mathcal{M}$ be a globally hyperbolic spacetime,
    $P$ a normally hyperbolic operator on $\mathcal{M}$,
    and $W = \tfrac{i}{2}E + H$ a Hadamard distribution for $P$,
    then $\mathfrak{F}_{\mu c}(\mathcal{M})[[\hbar]]$ is closed under $\star_H$.
\end{prop}

\begin{proof}
    Let $\mathcal{F}, \mathcal{G} \in \mathfrak{F}_{\mu c}(\mathcal{M})$,
    the $m^\text{th}$ derivative of the $\mathcal{O}(\hbar^n)$ term of
    $\mathcal{F} \star_H \mathcal{G}$ is,
    \begin{equation}
        \label{eq:star_prod_derivative}
        \left(
            \tfrac{d^n}{d \hbar^n}
                (\mathcal{F} \star_H \mathcal{G})
            |_{\hbar = 0}
        \right)^{(m)}
            =
        \sum_{\left\{ J_1, J_2 \right\} \in P_m}
            \left[
                \left(
                    \mathcal{F}^{(|J_1| + n)} \otimes
                    \mathcal{G}^{(|J_2| + n)}
                \right)
                \circ W^{\otimes n}
            \right]
                s_{J_1, J_2},
    \end{equation}
    where all notation is the same as in the previous proof,
    and the contraction $\circ$ is computed in the expected way, namely
    \begin{multline*}
        \left[
            \left( 
                \mathcal{F}^{(|J_1| + n)} \otimes
                \mathcal{G}^{(|J_2| + n)}
            \right)
            \circ W^{\otimes n}
        \right](x_1, \ldots x_m)
            \\ =
        \int_{\mathcal{M}^{2n}} \Big[
            \mathcal{F}^{(|J_1| + n)}(x_1, \ldots x_{|J_1|}, y_1, \ldots y_n)
            \mathcal{G}^{(|J_2| + n)}(
                x_{|J_1| + 1}, \ldots x_{m}, y_{n + 1}, \ldots y_{2n}
            )
            \\
        W(y_1, y_{n + 1}) \cdots W(y_n, y_{2n})
        \Big] \, \mathrm{d}y_1 \cdots \, \mathrm{d}y_{2n}.
    \end{multline*}
    In order to apply theorem 8.2.13 to
    $
        \left( \mathcal{F}^{(k + n)} \otimes \mathcal{G}^{(m - k + n)} \right) \circ
        (\chi W)^{\otimes n}
    $
    for $0 \leq k \leq m$, we must show that
    \begin{equation*}
        \mathrm{WF}'\left( 
            \mathcal{F}^{(k + n)} \otimes \mathcal{G}^{(m - k + n)}
        \right)_Y
            \cap
        \mathrm{WF}((\chi W)^{\otimes n})
            =
        \emptyset,
    \end{equation*}
    where $Y = \mathcal{M}^{2n}$ comprises the $y_i$ variables in the above integral.
    The justification of this proceeds similarly to before.
    Firstly, we note the following estimate, obtained by repeated application of
    8.2.9 from \cite{hormanderAnalysisLinearPartial2015}
    \begin{equation*}
        \mathrm{WF}(W^{\otimes n})
            \subseteq
        \left( 
            \mathrm{WF}(W) \cup 0_{\mathcal{M}^2}
        \right)^n
            \setminus
        0_{\mathcal{M}^{2n}}.
    \end{equation*}
    Hence, if
    $
        (y_1, \ldots, y_{2n}; \eta_1, \ldots, \eta_{2n})
            \in
        \mathrm{WF}((\chi W)^{\otimes n}),
    $
    then for each $i \in \{1, \ldots, n\}$, either
    $\eta_i$ and $\eta_{n + i}$ are both zero,
    or $(y_i; \eta_i) \in \overline{V}_+$ and $(y_{n + i}; \eta_{n + i}) \in \overline{V}_-$,
    moreover, $\eta_i$ must be non-zero for at least one $i$.
    Denote
    $\underline{y}_\mathcal{F} = (y_i)_{i = 1}^n$ and
    $\underline{y}_\mathcal{G} = (y_i)_{i = n + 1}^{2n}$,
    and similarly
    $\underline{\eta}_\mathcal{F}$ and
    $\underline{\eta}_\mathcal{G}$.
    Then we have that
    $(\underline{y}_\mathcal{F}; \underline{\eta}_\mathcal{F}) \in \overline{V}_+^n$, and
    $(\underline{y}_\mathcal{G}; \underline{\eta}_\mathcal{G}) \in \overline{V}_-^n$,
    hence neither can
    $
        (
            \underline{x}_\mathcal{F},  \underline{y   }_\mathcal{F};
            0                        , -\underline{\eta}_\mathcal{F}
        )
    $
    belong to $\mathrm{WF}(\mathcal{F}^{(k + n)})$,
    for any $\underline{x}_\mathcal{F} \in \mathcal{M}^k$, nor
    $
        (
            \underline{x}_\mathcal{G},  \underline{y   }_\mathcal{G};
            0                        , -\underline{\eta}_\mathcal{G}
        )
    $
    belong to $\mathrm{WF}(\mathcal{G}^{(m - k + n)})$,
    for any $\underline{x}_\mathcal{F} \in \mathcal{M}^{m - k}$.
    \footnote{
        Note that here we required the tighter restriction on $\mathrm{WF}(W)$
        relative to $E$:
        if we had covectors $(y_i; \eta_i) \in \overline{V}_+$ and $(y_j; \eta_j) \in \overline{V}_-$,
        for $i, j \in \{1, \ldots, n\}$,
        then it might be possible to find
        $
            (
                \underline{x}_\mathcal{F},  \underline{y   }_\mathcal{F};
                0                        , -\underline{\eta}_\mathcal{F}
            )
                \in
            \mathrm{WF}(\mathcal{F}^{(k + n)}),
        $
        hence the above intersection would in general be non-empty,
        preventing us from proceeding any further.
    }

    Now we must show that 8.2.13 precludes $\overline{V}_\pm^m$ from
    $
        \mathrm{WF}(
        ( \mathcal{F}^{(k + n)} \otimes \mathcal{G}^{(m - k + n)} )
        \circ (\chi W)^{\otimes n})
    $.
    Let
    $
        (
            \underline{x}_\mathcal{F},  \underline{x}_\mathcal{G};
            \underline{\xi}_\mathcal{F},  \underline{\xi}_\mathcal{G}
        )
            \in
        \overline{V}_+^{m}
    $
    and
    $
        (
            \underline{y   }_\mathcal{F}, \underline{y   }_\mathcal{G};
            \underline{\eta}_\mathcal{F}, \underline{\eta}_\mathcal{G}
        )        
            \in
        \left( 
            \mathrm{WF}((\chi W)^{\otimes n})
                \cup
            0_Y
        \right).
    $
    Then, just as before
    $
        (
            \underline{y}_\mathcal{F}, \underline{y}_\mathcal{G};
            \underline{\eta}_\mathcal{F}, \underline{\eta}_\mathcal{G}
        )
            \in
        \overline{V}_+^n \times \overline{V}_-^n
            \Rightarrow
        (
            \underline{x}_\mathcal{G},  \underline{y}_\mathcal{G};
            \underline{\xi}_\mathcal{G},  -\underline{\eta}_\mathcal{G}
        )
            \in
        \overline{V}_+^{m - k + n}.
    $
    Similarly to the final part of the proof of \Cref{prop:peierls_closure},
    one can then show $\underline{\xi}_\mathcal{F}$ cannot be zero,
    hence
    \begin{equation*}
        (
            \underline{x  }_\mathcal{F},   \underline{y   }_\mathcal{F},
            \underline{x  }_\mathcal{G},   \underline{y   }_\mathcal{G};
            \underline{\xi}_\mathcal{F},  -\underline{\eta}_\mathcal{F},
            \underline{\xi}_\mathcal{G},  -\underline{\eta}_\mathcal{G}
        )
            \notin
        \mathrm{WF}(\mathcal{F}^{(k + n)} \otimes \mathcal{G}^{(m - k + n)}),
    \end{equation*}
    whence \eqref{eq:hormander_composition_estimate} allows us to conclude
    \begin{equation*}
        (
            \underline{x}_\mathcal{F},  \underline{x}_\mathcal{G};
            \underline{\xi}_\mathcal{F},  \underline{\xi}_\mathcal{G}
        )
            \notin
        \mathrm{WF}\left( 
            ( \mathcal{F}^{(k + n)} \otimes \mathcal{G}^{(m - k + n)} )
            \circ (\chi W)^{\otimes n}
        \right).
    \end{equation*}
    To carry out the analogous argument for $\overline{V}_-^m$,
    one instead starts with the observation that
    \begin{align*}
        (
            \underline{x}_\mathcal{F},  \underline{x}_\mathcal{G};
            \underline{\xi}_\mathcal{F},  \underline{\xi}_\mathcal{G}
        )
            \in
        \overline{V}_-^{m} 
            &\text{ and }
        (
            \underline{y   }_\mathcal{F}, \underline{y   }_\mathcal{G};
            \underline{\eta}_\mathcal{F}, \underline{\eta}_\mathcal{G}
        )        
            \in
        \left( 
            \mathrm{WF}((\chi W)^{\otimes n})
                \cup
            0_Y
        \right)
            \\
            &\Rightarrow
        (
            \underline{x}_\mathcal{F},  \underline{y}_\mathcal{F};
            \underline{\xi}_\mathcal{F},  -\underline{\eta}_\mathcal{F}
        )
            \in
        \overline{V}_-^{k + n}
    \end{align*}
    and proceeds accordingly.

    This proves
    \begin{equation*}
        \mathrm{WF}\left(
            \left(
                \tfrac{d^n}{d\hbar^n}(
                    \mathcal{F} \star_H \mathcal{G}
                )|_{\hbar = 0}
            \right)^{(m)}
        \right)
            \cap
        \overline{V}_\pm^m
            =
        \emptyset,
    \end{equation*}
    thus each coefficient of $\mathcal{F} \star_H \mathcal{G}$ is a microcausal functional.
\end{proof}

\section{Squaring the Propagator}
\label{sec:cauchy_prod_prop}

In this section, we explain in detail why the expression \eqref{eq:squared_propagator}
for $[(\partial_u \otimes \partial_u) W_\mathrm{cyl}]^2$ is valid.
To simplify notation, we shall write
$(\partial_u \otimes \partial_u) W_\mathrm{cyl} =: w$,
and denote by $w_N$ the truncation of the series defining
$w$ to the first $N$ terms.

Theorem 8.2.4 of \cite{hormanderAnalysisLinearPartial2015} gives
the necessary conditions for the square of a distribution to exist.
However, it does not provide a convenient integral kernel with which to
evaluate such products on test functions.
A good starting point to this end may be found on page 526 of
\cite{chazarainIntroductionTheoryLinear1982},
where it is stated that for any pair of cones
$\Gamma_a, \Gamma_b \subseteq \dot{T}^*\mathcal{M}$
such that $\Gamma_a \cap - \Gamma_b = \emptyset$,
the multiplication of distributions, considered as a map
$
    \mathfrak{D}'_{\Gamma_a}(\mathcal{M}) \times
    \mathfrak{D}'_{\Gamma_b}(\mathcal{M}) \to
    \mathfrak{D}'(\mathcal{M})
$
is continuous in each of its arguments.
In other words, if we take some fixed
$u \in \mathfrak{D}'_{\Gamma_a}(\mathcal{M})$,
and a sequence $v_n$ converging to $v$ in the sense of
$\mathfrak{D}'_{\Gamma_b}(\mathcal{M})$,
then $u \cdot v_n$ weakly converges to $u \cdot v$,
and \textit{vice versa} for a sequence in $\mathfrak{D}'_{\Gamma_a}(\mathcal{M})$.

Let $\Gamma \subseteq \dot{T}^*\mathcal{M}_\mathrm{cyl}^2$ be a cone
which both contains $\mathrm{WF}(w)$ and satisfies
$\Gamma \cap - \Gamma = \emptyset$.
We can show that the smooth distributions $w_N$ obtained by
truncating the sum appearing in \eqref{eq:cylinder_2_point} converge
to $w$ in $\mathfrak{D}'_\Gamma$.

Firstly, we shall pick an open subset
$U \subset \mathcal{M}_{\mathrm{cyl}}^2$
which can be identified with an open subset of $\mathbb{R}^4$.
We shall only prove convergence for the restriction of $w_N$ to $U$,
though the full result follows from this with little trouble.
Following \cite[Definition 8.2.2]{hormanderAnalysisLinearPartial2015}
for sequential convergence, we must show that, for all $\chi \in \mathfrak{D}(U)$
and conic $V \subseteq \mathbb{R}^4$ such that
$\mathrm{supp} \, \chi \times V \cap \Gamma = \emptyset$,
\begin{equation*}
    \sup_{\xi \in V}
        \left|
            (1 + |\xi|)^k
            \left(
                \widehat{\chi w}(\xi) - \widehat{\chi w_N}(\xi)
            \right)
        \right|
    \to 0
    \text{ as } N \to \infty.
\end{equation*}

If we choose our coordinates for $U$ appropriately,
we can express this Fourier transform as
\begin{equation}
    \label{eq:convergence_integral}
    \widehat{\chi w}(\xi) - \widehat{\chi w_N}(\xi)
        =
    \sum_{n = N + 1}^\infty
        n \int_U \chi(x) e^{-in (\underline{u}, x)} e^{-i(\xi, x)} \, \mathrm{d}x,
\end{equation}
where $\underline{u} = (1, 0, -1, 0)$ is a constant vector.
If we set $F(x) := -(\underline{u}, x)$,
then each integral appearing in \eqref{eq:convergence_integral} can be expressed as
$T_\chi(n, \xi)$ using the notation in 
\cite[\S 4.3.2]{barQuantumFieldTheory2009}.
One can then show that the conditions are met for the stronger estimate
of corollary 2 from the same source to apply, i.e. for any $k \in \mathbb{N}$
\begin{equation*}
    |T_\chi(n, \xi)|
    \leq C_{\chi, V, k} (1 + n + |\xi|)^{-2k}
    \leq C'_{\chi, V, k} (1 + n)^{-k} (1 + |\xi|)^{-k},
\end{equation*}
for some appropriate choice of positive constants.
This allows us to uniformly bound the original expression in $\xi$ as
\begin{equation*}
    \sup_{\xi \in V}
        \left|
            (1 + |\xi|)^k
            \left(
                \widehat{\chi w}(\xi) - \widehat{\chi w_N}(\xi)
            \right)
        \right|
    \leq
    C'_{\chi, V, k} \sum_{n = N + 1}^\infty
    (1 + n)^{1 - k}.
\end{equation*}
For $k \geq 3$, this establishes the convergence desired.
For $k = 1, 2$, we simply pick a stronger bound for $T_\chi$.

Thus we can write, for $f \in \mathfrak{D}(\mathcal{M}_\mathrm{cyl}^2)$
\begin{equation*}
    \left\langle
        w^2,
        f
    \right\rangle
        =
    \lim_{N \to \infty}
    \left\langle
        w_N \cdot w,
        f
    \right\rangle,
\end{equation*}
which allows us to bring all summation outside of
the integrals arising from the duality pairing.
Noting that $w_N$ is a smooth function for all finite $N$,
we can hence evaluate this pairing directly as
\begin{align*}
    \left\langle
        w^2,
        f
    \right\rangle
        &=
    \lim_{N \to \infty}
    \sum_{m = 0}^\infty m
    \int_{\mathcal{M}_\mathrm{cyl}^2}
        e^{-im(u - u')}
        \left[
            \sum_{n = 0}^N n
            e^{-in(u - u')}
            f(u, v, u', v')
        \right]
    \, \mathrm{dVol}^2 \\
        &=
    \sum_{n = 0}^\infty
    \sum_{m = 0}^\infty nm
    \int_{\mathcal{M}_\mathrm{cyl}^2}
        e^{-i(n + m)(u - u')}
        f(u, v, u', v')
    \, \mathrm{dVol}^2,
\end{align*}
where, \textit{a priori}, the sum over $m$ must be performed first.

As $f$ is smooth, the integral is rapidly decaying as a function of $n + m$,
hence the sum is absolutely convergent.
Rearranging the double sum accordingly,
it is then clear that the sequence of partial sums
\begin{equation}
    w^2_N(u, v, u', v')
        :=
    \sum_{k = 0}^N \sum_{l = 0}^k
    l(k - l) e^{-ik(u - u')}
\end{equation}
converges to $w^2$ in the weak topology of $\mathfrak{D}'(\mathcal{M}_\mathrm{cyl}^2)$.

\bibliographystyle{alpha}
\bibliography{VOAs_pAQFT}

\end{document}